\documentclass[parskip=half, twoside=false]{scrartcl}

\usepackage{amssymb}
\usepackage{amsthm}
\usepackage{graphicx}
\usepackage[colorlinks,citecolor=blue,urlcolor=blue]{hyperref}
\usepackage{mathtools}
\usepackage{natbib}

\newtheorem{theorem}{Theorem}
\newtheorem{proposition}{Proposition}
\newtheorem{lemma}{Lemma}
\newtheorem{example}{Example}
\newtheorem{remark}{Remark}

\DeclarePairedDelimiter{\abs}{\lvert}{\rvert}
\DeclarePairedDelimiter{\norm}{\lVert}{\rVert}
\DeclarePairedDelimiter{\braces}{\lbrace}{\rbrace}
\DeclarePairedDelimiter{\parens}{\lparen}{\rparen}
\DeclarePairedDelimiter{\bracks}{\lbrack}{\rbrack}
\DeclarePairedDelimiter{\iprod}{\langle}{\rangle}
\DeclarePairedDelimiter{\sn}{\lvert}{\rvert_{\infty}}
\DeclarePairedDelimiter{\lip}{\lvert}{\rvert_{\mathrm{Lip}}}

\DeclareMathOperator{\argmax}{argmax}
\DeclareMathOperator{\diag}{diag}
\DeclareMathOperator{\tr}{tr}
\DeclareMathOperator{\Var}{Var}
\DeclareMathOperator{\Cov}{Cov}

\DeclareMathOperator{\E}{\mathbb{E}}

\newcommand{\arr}[1]{\boldsymbol{#1}}
\newcommand{\pder}[2]{\frac{\partial #1}{\partial #2}}
\newcommand{\T}{^{\mathsf{T}}}
\newcommand{\dd}{\operatorname{d}\!}
\newcommand{\Ls}{\Lambda^s}
\newcommand{\Le}{\Lambda^e}
\newcommand{\Nor}{\mathcal{N}}
\newcommand{\Ga}{\operatorname{Ga}}
\renewcommand{\P}{\mathbb{\mathop{P}}}

\newcommand{\hf}{\widehat{f}}
\newcommand{\hu}{\widehat{u}}

\newcommand{\hh}{\widehat{h}}
\newcommand{\bone}{\boldsymbol{1}}
\newcommand{\bph}{\boldsymbol{\phi}}
\newcommand{\bome}{\boldsymbol{\omega}}
\newcommand{\bsig}{\boldsymbol{\sigma}}
\newcommand{\ba}{\boldsymbol{a}}
\newcommand{\bb}{\boldsymbol{b}}

\newcommand{\bu}{\boldsymbol{u}}
\newcommand{\bv}{\boldsymbol{v}}
\newcommand{\bw}{\boldsymbol{w}}
\newcommand{\bz}{\boldsymbol{z}}
\newcommand{\bF}{\boldsymbol{F}}
\newcommand{\bR}{\boldsymbol{R}}

\newcommand{\bI}{\boldsymbol{I}}
\newcommand{\bA}{\boldsymbol{A}}
\newcommand{\bB}{\boldsymbol{B}}
\newcommand{\bE}{\boldsymbol{E}}
\newcommand{\bQ}{\boldsymbol{Q}}
\newcommand{\hbF}{\widehat{\bF}}
\newcommand{\hbu}{\widehat{\bu}}
\newcommand{\hbw}{\widehat{\bw}}
\newcommand{\hbz}{\widehat{\bz}}

\title{Functional Estimation of the Marginal Likelihood}
\author{
  Omiros Papaspiliopoulos \\
  \href{mailto:omiros.papaspiliopoulos@unibocconi.it}{\texttt{[omiros.papaspiliopoulos@unibocconi.it]}} \and
  Timothée Stumpf-Fétizon \\\href{mailto:timothee.stumpffetizon@unibocconi@unibocconi.it}{\texttt{[timothee.stumpffetizon@unibocconi.it]}} \and
  Jonathan Weare \\
  \href{mailto:weare@nyu.edu}{\texttt{[weare@nyu.edu]}}
}

\begin{document}
\maketitle

\begin{abstract}
  We propose a framework for computing, optimizing and integrating with respect to a smooth marginal likelihood in statistical models that involve high-dimensional parameters/latent variables and continuous low-dimensional hyperparameters. The method requires samples from the posterior distribution of the parameters for different values of the hyperparameters on a simulation grid and returns inference on the marginal likelihood defined everywhere on its domain, and on its functionals. We show how the method relates to many of the methods that have been used in this context, including sequential Monte Carlo, Gibbs sampling, Monte Carlo maximum likelihood, and umbrella sampling. We establish the consistency of the proposed estimators as the sampling effort increases, both when the simulation grid is kept fixed and when it becomes dense in the domain. We showcase the approach on Gaussian process regression and classification and crossed effect models.
  \par\vskip\baselineskip\noindent
  \textbf{Code:} \url{https://github.com/timsf/infemus}
  \par\vskip\baselineskip\noindent
  \textbf{Keywords:} Markov chain Monte Carlo, umbrella sampling, normalizing constants, perturbation theory
  \par\vskip\baselineskip\noindent
  \textbf{MSC2020 subject classifications:} Primary 62-08, 65C05; Secondary 62F15
\end{abstract}

\section{Introduction}
\label{sec:intro}

In this article we tackle a problem of fundamental importance in statistics and machine learning, that of exploring, optimizing and averaging with respect to the \emph{marginal likelihood}. Our contribution is two-fold. First, we develop a functional estimator that interpolates pointwise estimates of the marginal likelihood delivered by well-established umbrella sampling techniques, which include bridge sampling as a special case. Second, we establish functional consistency properties of the estimator when Monte Carlo effort increases.  The new functional estimator and its analysis are made possible by a novel derivation of the existing umbrella sampling technique we build upon, which also creates insightful links to common methods in the computational statistics literature.

The setting is as follows. A probabilistic model for data $y$ is specified in terms of a likelihood function $p(y | \theta, \lambda)$, with high-dimensional \emph{parameters/latent variables/stochastic processes} $\theta$ with density $p(\theta | \lambda)$, and low-dimensional \emph{hyperparameters} $\lambda \in \Lambda \subset \Re^{p}$, where $\Lambda$ is a connected set. The function of interest in this article is the marginal likelihood
\begin{equation}
  \label{eq:ml}
  \lambda \mapsto p(y | \lambda), \quad p(y | \lambda) = \int p(y | \theta,\lambda) p(\theta|\lambda) \dd \theta, \quad \lambda \in \Lambda.
\end{equation}
We are interested in plotting the marginal likelihood, when possible, or its one-dimensional profiles and marginals (when integrable). The gradients of this function are also of interest, as are the locations of local maxima of the marginal likelihood. The latter can be used to obtain empirical Bayes estimates for the hyperparameters. By integration we mean computing marginal expectations of scalar-valued functions $\phi(\theta)$, which we will denote by $\pi(\phi)$:
\begin{equation}
  \label{eq:me}
  \pi(\phi) = \frac{\int \phi(\theta) p(y | \theta,\lambda) p(\theta | \lambda) p(\lambda) \dd{\theta} \dd {\lambda}}{\int p(y | \lambda) p(\lambda) \dd{\lambda}},
\end{equation}
where $p(\lambda)$ is a prior on the hyperparameters, although often this might be just a constant in cases where $\Lambda$ is a compact set. It is instructive to consider a few different approaches to this problem that are closely related to, as well as unified by, the one we develop in this work.

One approach is the joint sampling of $p(\theta, \lambda | y)$ using Markov Chain Monte Carlo (MCMC). For a recent review of various such methods, especially with view towards dealing with multi-modality,  see for example \cite{latuszynski2025mcmc}. For many classes of structured models in statistics and machine learning (such as Gaussian process regression and mixed models, as in our numerical examples, but also in  latent class models or stochastic differential equations) there exist very efficient sampling methods for $p(\theta | y, \lambda)$, some with provably linear complexity in $\theta$-dimension. Since, in our context,  $\lambda$ is low-dimensional, a common strategy is to use a coordinate-wise updating sampling scheme, a special case of which is the Gibbs sampler when both conditional densities can be sampled directly. A major disadvantage of this approach is that it has slow convergence where there is strong posterior dependence between $\lambda$ and $\theta$, that is when  the conditional densities $p(\theta | y, \lambda)$ have little mutual support for different $\lambda$'s. Reparameterization strategies have been successful in reducing this dependence \citep{papaspiliopoulos2007general, robust}, but oftentimes this is not enough.  Further difficulties arise when there are multiple modes in the marginal likelihood and/or the joint density $p(\theta, \lambda | y)$.

For the discretized setting, where the hyperparameter space is discretized and restricted to a finite set $\Ls = \braces{\lambda_{\ell}}_{\ell=1}^L$, a different approach is to resort  to a vast literature on estimating normalizing constants and ratios thereof. This is particularly feasible and natural when $\lambda$ is one-dimensional and hence the elements in $\Ls$ are ordered, e.g., with $\lambda_{1}$ in the first and $\lambda_{L}$ in the last position, and the object of interest is $p(y | \lambda_{L})/p(y | \lambda_{1})$. A common instance is for $\lambda \in \Re_+$ to play the role of a ``temperature'', such that $p(\theta | y,\lambda_{1})$ is easily sampled from, while the object of interest $p(\theta | y, \lambda_{L})$ is difficult to sample from. Then, \emph{Sequential Monte Carlo (SMC)} methods  can be applied to this task, see for example Chapters 3 and 17 of \cite{smc_book}. An alternative Monte Carlo approach is  the \emph{Vardi estimator}, as discussed in \cite{vardi} and the contributed discussion therein. Under the term \emph{umbrella sampling}, similar approaches have been developed and they are regularly used in computational chemistry for the estimation of parameter and collective variable free energies, as in \cite{FrenkelSmit2002}. Note that umbrella sampling techniques are also used for joint sampling of $p(\theta, \lambda | y)$, especially when there are multi-modalities \citep{DinnerThiede2020emus2}. For an application of this within statistics, see \cite{Chopin2012fe}. Within the umbrella sampling framework, the so-called \emph{Eigenvector Method for Umbrella Sampling (EMUS)} was developed in \cite{thiede2016eigenvector}, and this turns out to be directly related to the Vardi estimator. A special case of the Vardi estimator is bridge sampling, as in \cite{gelman-meng}, a method that that also has links to SMC. We discuss the connections among all the above in Section \ref{ssec:lit}.

For functional estimation of $p(y | \lambda)$, a classic stochastic approximation approach known as \emph{Monte Carlo maximum likelihood} uses samples from $p(\theta | y, \lambda_{i})$ for a single grid point $\lambda_{i}$ \citep{mcmle-cg,mcmle-om}. This is based on an importance sampling identity and yields an estimator of $p(y | \lambda) / p(y | \lambda_{i})$ for $\lambda \in \Lambda$.

The approach we develop in this article takes EMUS, which is a grid-based estimator defined on $\Ls$, and lifts it up to a functional estimator over the whole of $\Lambda$, by exploiting a surprising reproducing property. Given that bridge sampling turns out to be a special, but sub-optimal, version of EMUS (see Section \ref{ssec:lit}) our approach can also be used to extrapolate bridge sampling estimates beyond the grid they have been constructed. In short, the approach takes as input ``local samples'',  $\theta_{\ell}^{(n)} \sim p(\theta | y,\lambda_{\ell})$, for $n=1,\ldots N_{\ell}$, and $\lambda_{\ell} \in \Ls$, which are typically generated using MCMC, and it returns a functional estimator of $\braces{p(y | \lambda): \lambda \in \Lambda}$, up to normalization. Therefore, the main requirement is that  reasonably efficient  local samplers exist, which are treated as black boxes. There can be efficiency gains by coupling the local samplers, but here we opt for a plug-and-play approach. The estimator we develop is novel, easy to compute, and, as we show, it unifies and connects most of the schemes discussed earlier in this section.

An important contribution of this article is the careful theoretical analysis of the asymptotic properties of the proposed estimator. We obtain uniform convergence results in two asymptotic regimes, one we call fixed-grid asymptotics, where the size of the grid ($L$) is kept fixed and the Monte Carlo effort per grid point ($N_{\ell}$) increases, and one we call dense-grid asymptotics,  where $N_{\ell}$ is kept fixed (e.g., $N_\ell =1$) and $L$ increases. For practical purposes, the design of the simulation grid $\Ls$ is an important consideration. Using the functional estimate and expressions for its asymptotic variance, we propose an iterative scheme that can be used for optimal design. We demonstrate a pilot implementation that shows promise but we largely leave its full implementation to future work.

There are conceptual links between our approach and that in \cite{inla}, as both are based on the principle of stratification according to, and numerical integration with respect to the hyperparameters.

\section{Eigenvector methods}
\label{sec:emus}

As mentioned above, our approach in this article relies on the EMUS implementation of umbrella sampling. We provide a description of EMUS in the context of marginal likelihood estimation in Section \ref{ssec:emus} and elucidate its connection to the \emph{Gibbs sampler} in Section \ref{ssec:emus-gibbs}. This connection has the pedagogical value of linking a standard method from computational physics/chemistry to one from computational statistics. Crucially, it allows us to develop the functional estimator and the dense-grid consistency theory by relating the continuous and discrete spectra of the associated Markov processes. Section \ref{ssec:emus-error} presents some error analysis for EMUS, which our consistency theory for the functional estimator builds upon. Section \ref{ssec:lit} draws  further connections between EMUS and other established methods for estimating normalizing constants, and it discusses computational complexity and optimality properties of EMUS.

\subsection{The EMUS estimator}
\label{ssec:emus}

EMUS was first introduced as a method for chemical free energy calculations in \cite{thiede2016eigenvector} and then suggested as a variance reduction strategy for general expectations in \cite{DinnerThiede2020emus2}. As explained in those works, and discussed further in Section \ref{ssec:lit}, an iteration of the EMUS estimator gives an efficient algorithm to compute the Vardi estimator. The Vardi estimator was first introduced in the context of biased sampling in \cite{Vardi:EmpDistSelectionBias1985} and discussed as a tool for Monte Carlo integration and normalizing constant estimation in \citep{vardi}.

When cast as a tool for normalizing constant estimation, EMUS can be described as follows. Let $\psi_{\lambda_\ell}(\theta)$, for $\ell=1,\ldots, L$, be positive integrable functions, indexed by $\lambda_\ell \in \Ls = \braces{\lambda_{\ell}}_{\ell=1}^{L}$. Although it would be simpler to  index these by $\ell=1,\ldots,L$, but we opt for what appears to be a more cumbersome notation to facilitate the use of this method for functional estimation of the marginal likelihood later. Then, we define
\begin{equation}
  \label{eq:main-q-1}
  \begin{aligned}
    z(\lambda) & = \int \psi_{\lambda}(\theta) \dd \theta,\quad \pi_{\lambda}(\theta) = \frac{\psi_{\lambda}(\theta)}{z(\lambda)},
  \end{aligned}
\end{equation}
with $z(\lambda_\ell)$ for $\ell=1,\ldots,L$ being the normalizing constants of interest for estimation. 

EMUS is based on the following two ingredients. The first is a fundamental, albeit elementary, identity established in the following Proposition, whose proof is omitted.

\begin{proposition}
  \label{prop:eigenvec}
  Let $\bz$ be the vector of normalizing constants,
  \begin{equation*}
  \bz_{\ell} = z(\lambda_{\ell}),\quad \lambda_{\ell} \in \Ls,
  \end{equation*}
  and $\bF$  the stochastic matrix with elements
  \begin{equation*}
    \bF_{i,j} = \int \frac{\psi_{\lambda_{j}}(\theta) }{\sum_{\ell=1}^{L} \psi_{\lambda_{\ell}}(\theta) } \pi_{\lambda_{i}}(\theta) \dd \theta,\quad \lambda_{i}, \lambda_{j} \in \Ls. 
  \end{equation*}
  Then, $\bz$ is in detailed balance with $\bF$:
  \begin{equation}
    \label{eq:detailed}
    \bz_i \bF_{i,j} = \bz_j \bF_{j,i},
  \end{equation}
  and therefore it solves the eigenvector problem
  \begin{equation*}
    \bz = \bF\T \bz.
  \end{equation*}
\end{proposition}

\begin{remark}
  Above and in the rest of the article we use bold-face letters to refer to vectors (lower case) and matrices (upper case) for quantities defined on the finite grid $\Ls$. This will differentiate such quantities from functions defined on the infinite compact set $\Lambda$, for which regular fonts will be used. 
\end{remark}

The second ingredient of EMUS is a local sampler from the densities $\pi_{\lambda_{\ell}}(\theta)$, for $\lambda_{\ell} \in \Ls$. These are used to form an estimated matrix $\hbF$, with elements:
\begin{equation*}
  \hbF_{i,j} = \frac{1}{N_{i}} \sum_{n=1}^{N_{i}} \frac{\psi_{\lambda_{j}}(\theta_{i}^{(n)}) }{\sum_{\ell = 1}^{L} \psi_{\lambda_{\ell}}(\theta_{i}^{(n)}) }, \quad \theta_{i}^{(n)} \sim \pi_{\lambda_{i}}(\theta), \quad 
  n = 1, \ldots, N_{i}.
\end{equation*}
It is easy to check that $\hbF$ is also a stochastic matrix. With these ingredients, the EMUS estimator is defined as the solution to the eigenvector problem, $\hbz = \hbF\T \hbz$, which can be seen as a random perturbation of that in Proposition \ref{prop:eigenvec}. The eigenvector is best found by a direct linear solve of $\hbz = \hbF\T \hbz$ through a QR decomposition of $\bI-\hbF\T$ (see Supplement \ref{app:compemus} for details).

\subsection{Connection to the Gibbs sampler}
\label{ssec:emus-gibbs}

In this section we make a novel connection between the EMUS estimator and the Gibbs sampler. This observation links methods used extensively in the computational physics/chemistry community with those used in the computational statistics community. The purpose of this connection is not merely pedagogical, as it is a crucial ingredient of our functional estimator in Section \ref{sec:nystrom} and its theoretical analysis in Section \ref{sec:theory}. 

We extend slightly the framework of Section \ref{ssec:emus} and introduce a prior density $p(\lambda)$ on the space $\Lambda$. This density might be part of the statistical modelling, or it might be a computational tool to make the local function $\psi_\lambda(\theta)$ also integrable in $\lambda$, and it might be taken to be simply 1 when this is possible and desirable. We then define
\begin{equation}
  \label{eq:main-q-2}
  \begin{aligned}
    \pi_\theta(\lambda) & = \frac{\psi_{\lambda}(\theta) p(\lambda)}{\int \psi_{\lambda'}(\theta) p(\lambda') \dd \lambda'}, \quad u(\lambda) \propto z(\lambda) p(\lambda),
  \end{aligned}
\end{equation}
with the normalization of $u(\lambda)$ discussed below in Remark~\ref{rmk:conventions}. Note that computing $z(\lambda)$ is equivalent to computing $u(\lambda)$ since $p(\lambda)$ is explicitly known.

Consider, now, the two-component Gibbs sampler that updates iteratively $\theta$ and $\lambda$ according to the two conditional densities in \eqref{eq:main-q-1} and \eqref{eq:main-q-2} respectively. By construction, it generates a Markov chain on $\Lambda$ with transition kernel
\begin{equation}
  \label{eq:infemus}
  g(\lambda,\lambda') = \int \pi_{\lambda}(\theta) \pi_\theta(\lambda') \dd \theta, 
\end{equation}
which is reversible (hence invariant) with respect to $u(\lambda)$, hence solves the following eigenproblem:
\begin{equation*}
  u(\lambda) = \int u(\lambda') g(\lambda',\lambda) \dd \lambda'.
\end{equation*}

A common practice among practitioners is the so-called \emph{griddy Gibbs sampler} approach \citep{griddy}, where the hyperparameter space is discretized and restricted to a finite set $\Ls = \braces{\lambda_{\ell}}_{\ell=1}^L$, and the update of $\lambda$ is done by sampling from $\pi_{\theta}(\lambda_{\ell})$ for $\lambda_{\ell} \in \Ls$, normalized on the grid. Consequently, the transition kernel becomes a stochastic matrix and a direct calculation shows that is precisely:
\begin{equation}
  \label{eq:kernel-discr}
  \bF_{i,j} = \int \frac{\psi_{\lambda_{j}}(\theta) p(\lambda_{j})}{\sum_{\ell=1}^{L} \psi_{\lambda_{\ell}}(\theta) p(\lambda_{\ell})} \pi_{\lambda_{i}}(\theta) \dd \theta,\quad \lambda_{i}, \lambda_{j} \in \Ls,
\end{equation}
which is the direct extension of the matrix defined in Proposition \ref{prop:eigenvec} allowing for $p(\lambda) \neq 1$. In an analogous way, the griddy Gibbs sampler generates samples according to a probability vector proportional to:
\begin{equation}
  \label{eq:bfu}
  \bu_{\ell} = u(\lambda_{\ell}),\quad \lambda_{\ell} \in \Ls,
\end{equation}
which is in detailed balance with $\bF$,  $\bu_i \bF_{i,j} = \bu_j \bF_{j,i}$,
and therefore also solves the eigenproblem $\bu = \bF^T \bu$.

\begin{remark}
\label{rmk:conventions}
  In the rest of the article we will use \eqref{eq:kernel-discr} as the definition of $\bF$, define $\hbF$ with elements:
  \begin{equation}
  \label{eq:hat-f}
    \hbF_{i,j} = \frac{1}{N_{i}} \sum_{n=1}^{N_{i}} \frac{\psi_{\lambda_{j}}(\theta_{i}^{(n)}) p(\lambda_{j})}{\sum_{\ell = 1}^{L} \psi_{\lambda_{\ell}}(\theta_{i}^{(n)}) p(\lambda_{\ell})}, \quad \theta_{i}^{(n)} \sim \pi_{\lambda_{i}}(\theta), \quad 
  n = 1, \ldots, N_{i},
  \end{equation}
  which is again a stochastic matrix, and define the EMUS estimator as the solution (by a direct solve) to the following eigenproblem: 
  \begin{equation}
  \label{eq:emus}
    \hbu = \hbF\T \hbu.
  \end{equation}
  While for practical use we may fix $\sum_{\ell} \bu_{\ell} = \sum_{\ell} \hbu_{\ell}$ to any arbitrary number, it is helpful to the theoretical discussion of Section \ref{sec:dense-g} to adopt the convention
\begin{equation}
    \label{eq:sumtol}
    (SumToL) \quad \sum_{\ell=1}^{L} \bu_{\ell} = \sum_{\ell=1}^{L} \hbu_{\ell} = L,
\end{equation}
as we do for the remainder of the article. Due to the constraint imposed by \eqref{eq:bfu},  this normalization for $\bu$ implies a normalization for $u(\lambda)$. 
 All the above directly generalize in a trivial manner the quantities defined in Section \ref{ssec:emus} when $p(\lambda) \neq 1$.
\end{remark}

It follows that the griddy Gibbs sampler and the EMUS estimator use the same ingredients (i.e. local sampling) and solve the same eigenproblems, albeit in different ways. The griddy Gibbs sampler solves it by sampling from both conditional densities, and we can think of the empirical distribution on $\Ls$ as an approximation $\hbu$ to $\bu$, up to normalization. The EMUS estimator instead uses a direct solve, effectively stratifying in the $\Ls$ space instead of sampling. This has important consequences for the errors of the respective methods, in particular their sensitivity to the \emph{spectral gap} of $\bF$, which we analyze below.

\subsection{EMUS error analysis}
\label{ssec:emus-error}

As previously observed by \cite{DinnerThiede2020emus2, thiede2015sharp}, one of the most striking properties of EMUS is that it can be precise even if the spectral gap of $\bF$ is small. The spectral gap is defined as the difference between 1 and the second largest modulo eigenvalue of $\bF$ \cite[see for example][]{liu}, and  it is small in situations of strong dependence between $\theta$ and $\lambda$, or in the presence of multimodality. The point is best illustrated by an expression of the asymptotic variance of EMUS in Theorem \ref{thm:avar-emus} below. Its proof is given in Supplement \ref{app:avar-emus}, and builds upon  results obtained in \cite{DinnerThiede2020emus2} and \cite{thiede2015sharp}. We use this for the error analysis of the functional estimator in Section \ref{sec:theory}. 

\begin{theorem}
  \label{thm:avar-emus}
  Assume that:
  \begin{itemize}
  \item[(GridIrred)] $\bF$ is irreducible;
  \item[(IndSample)] The rows of $\hbF$ are independent;
  \item[(PosWeight)] For the total amount of Monte Carlo sampling $N = \sum_{\ell=1}^{L} N_{\ell}$, we ask that $w_{\ell} = \lim_{N \to \infty} N_{\ell}/N > 0$.
  \end{itemize}
  Assume that the $\theta^{(n)}_{\ell}$ are independent across $\ell$ and $n$. Let
  \begin{equation*}
    \bR_{i,j} = N_{i}\Var(\hbF_{i,j}) =  \Var\parens*{ \frac{\psi_{\lambda_{j}}(\theta_{i}^{(1)}) p(\lambda_{j})}{\sum_{\ell=1}^{L} \psi_{\lambda_{\ell}}(\theta_{i}^{(1)}) p(\lambda_{\ell})}}.
  \end{equation*}
  Then, as the number of samples $N\to \infty$, $\sqrt{N} \left(\hbu_{\ell} - \bu_{\ell}\right)$ converges in distribution to a Gaussian variable with mean 0 and variance $\bsig_{\ell}^2$. The \emph{relative} asymptotic variance $\bsig_{\ell}^2/\bu_{\ell}^2$ satisfies the bound
  \begin{equation*}
   \sup_{\ell} \frac{\bsig_{\ell}^2}{\bu_{\ell}^2} \leq  L \sum_{i=1}^{L} w_{i}^{-1} \sum_{j \neq i} \frac{\bR_{i,j}}{\bQ_{i,j}^{2}},
  \end{equation*}
  where $\bQ_{i,j}$ is the probability that a Markov chain on state-space $\Ls$ with transition matrix $\bF$ starting from $\lambda_{i}$ first visits $\lambda_{j}$ before returning to $\lambda_{i}$.
\end{theorem}

The theorem as stated assumes independent local samples at each grid point, but it can be extended to the case where these are generated according to a Markov chain, provided a CLT holds. As studied in detail in \cite{DinnerThiede2020emus2}, the fraction on the right hand side of the last display can be of moderate size even when both numerator and denominator are small, and in particular when the spectral gap of $\bF$ is very small. We further illustrate this robustness to spectral gap decay in Supplement \ref{app:toy}, which contains a toy example and some numerical comparisons between the griddy Gibbs sampler and the EMUS estimator.

\subsection{Related normalizing constant estimators}
\label{ssec:lit}

Here we clarify connections between the EMUS estimator and, among others, SMC, safe importance sampling, bridge sampling, and the Vardi estimator. Connections between our functional estimator, which builds upon EMUS, and Monte Carlo maximum likelihood are discussed in Section \ref{sec:nystrom}.

As remarked earlier  and explained in \cite{DinnerThiede2020emus2}, the Vardi estimator can be obtained by a recursive application of the the EMUS estimator, which can be seen as a linearization of the non-linear system of equations required to obtain the Vardi estimator. 
The Vardi estimator is in a specific sense the optimal way to use the samples from the local densities to estimate the ratios of normalizing constants. The precise result is discussed in \cite{vardi} according to which the estimator can be obtained as  maximum likelihood for a semi-parametric model in which the dominating measure is unknown to the analyst. For a precise result on its optimality in asymptotic variance see Theorem 1 in \cite{tan2004likelihood}, a result to which we return below.

One can arrive at the EMUS estimator by an argument different from that used in Section \ref{ssec:emus}, exploiting the following identity:
\begin{equation*}
  z(\lambda) = \int \frac{\psi_{\lambda}(\theta) }{\sum_{\ell=1}^{L} w_{\ell} \pi_{\lambda_{\ell}}(\theta)}  \sum_{\ell=1}^{L} w_{\ell} \pi_{\lambda_{\ell}}(\theta) \dd \theta.
\end{equation*}
For a specific choice of the weights $w_{\ell}$ (which would depend, though, on the normalizing constants $z(\lambda_{\ell})$), and by a re-arrangement the terms on the right-hand-side, we can obtain the same eigenequation as in Proposition \ref{prop:eigenvec}.  This connects EMUS to safe importance sampling, e.g. \cite{owen2000safe}, which is based on simulating from mixtures of local densities. However, we do not implement EMUS by simulating from a mixture distribution, which would be impossible because the component weights would be intractable.

More generally, the use of importance sampling to build estimators of ratios of normalizing constants has a long history, see for example Section 3 in \cite{gelman-meng} for an early but influential account. The most elementary use is via the identity:
\begin{equation}
  \label{eq:is-id-grid}
  \frac{\bu_{j}}{\bu_{i}} = \int \frac{\psi_{\lambda_{j}}(\theta) p(\lambda_{j})}{\psi_{\lambda_{i}}(\theta) p(\lambda_{i})} \pi_{\lambda_{i}}(\theta) \dd{\theta}.
\end{equation}
It is known that the resulting estimator can easily have high or even infinite variance. EMUS, instead, averages  bounded functions. For completeness, an illustrative toy example in Supplement \ref{app:toy} compares EMUS to this naive use of importance sampling. Better schemes than \eqref{eq:is-id-grid} are based on the identity studied in \cite{meng-wong}:
\begin{equation}
  \label{eq:bridge}
  \frac{\bu_{j}}{\bu_{i}} = \frac{\int \alpha(\theta) \psi_{\lambda_{j}}(\theta) p(\lambda_{j}) \pi_{\lambda_{i}}(\theta) \dd{\theta}}{\int \alpha(\theta) \psi_{\lambda_{i}}(\theta) p(\lambda_{i}) \pi_{\lambda_{j}}(\theta) \dd{\theta}},
\end{equation}
where $\alpha(\theta)$ is sometimes called the bridge function. 

We can relate this paradigm to EMUS by choosing $\alpha(\theta) = 1/\sum_{\ell=1}^L \psi_{\lambda_\ell}(\theta) p(\lambda_\ell)$, which plugged into \eqref{eq:bridge} leads to the identity $\bu_{j}/\bu_{i} = \bF_{i,j}/\bF_{j,i}$. This identity corresponds to the detailed balance equations obtained in Proposition \ref{prop:eigenvec} and extended in Section \ref{ssec:emus-gibbs}. The corresponding estimator of the ratio of normalizing constants is $\hbF_{i,j}/\hbF_{j,i}$. This estimator is different, and typically worse, than EMUS. Notice that whereas $\bu$ and $\bF$ are in detailed balance the same does not necessarily hold for $\hbu$ and $\hbF$. Therefore, in general, ratios of estimated normalizing constants, $\hbu_i/\hbu_j$, as obtained by EMUS,  depend on all elements of $\hbF$, and by extension on samples generated from all local densities, not just from the corresponding pairs $\pi_{\lambda_{i}}(\theta)$ and $\pi_{\lambda_{j}}(\theta)$. In fact, in many cases $\hbF_{j,i} = \hbF_{i,j} = 0$, and yet EMUS yields ratios $\hbz_i/\hbz_j$ with controlled errors.

Bridge sampling chooses some ordering of the elements in $\Ls$  and constructs an estimator based on a telescoping product of terms of the form \eqref{eq:bridge}. Given such ordering,  SMC can be used to obtain samples from the corresponding sequence of local densities, and returns as a by-product an estimator of ratios of their normalizing constants, see for example Chapter 17 
in \cite{smc_book}. Bridge sampling and SMC are related but they are distinct, and can be complementary to each other. Bridge sampling can use SMC for obtaining the necessary local samples, but it is based on a ``filtering'' approach to estimating ratios of normalizing constants (in the sense that local samples from $\pi_{\lambda_\ell}$ are used when estimating ratios of its normalizing constant), as opposed to the ``prediction'' approach used  in SMC (where only samples from previous densities in the ordering are used). For this reason, bridge sampling will typically result in better estimates compared to SMC.  

When all entries beyond the first sub- and super-diagonals of $\bF$ are zero, bridge sampling with index ordering $1,2,\dots,L$ and EMUS produce the same estimator. However, in many applications (e.g. when $\lambda$ is not a scalar) the choice of an ordering  is artificial and some choices can result in high variance in bridge sampling, whereas the choice of index ordering has no impact on EMUS. For more detail, the reader can refer to
the insightful analysis in Theorem 1 of \cite{tan2004likelihood} which, among other things, compares bridge sampling and the Vardi estimator in terms of their asymptotic variance. 

EMUS solves the eigenequation by a direct linear solve. In the worst case (when $\hbF$ is dense), the cost of this linear solve is $\mathcal{O}\left(L^3\right)$ floating point operations. When evaluation of the likelihoods  is expensive, this cost can be dominated by the $\mathcal{O}\left(L^2\right)$ cost to form $\hbF$. The cost of forming $\hbF$ and solving for $\hbz$ are both reduced when $\bF$, and therefore also $\hbF$, is sparse, which is easily achieved by a minor extension of  EMUS, though typically at the cost of a smaller spectral gap for $\bF$. However, for the values of $L$ used in practice, the  $\mathcal{O}\left(L\right)$ cost of sampling from local densities often dominates the overall computational cost.

\section{The functional estimator}
\label{sec:nystrom}

Our proposed functional estimator builds upon the eigenvector methods that operate on the lattice $\Ls$ and the following rather surprising reproducing property. 

\begin{proposition}
  \label{prop:reproduce}
  We define 
  \begin{equation}
    \label{eq:kernel}
    f(\lambda_{i}, \lambda) = \int \frac{\psi_{\lambda}(\theta) p(\lambda)}{\sum_{\ell=1}^{L} \psi_{\lambda_{\ell}}(\theta) p(\lambda_{\ell})} \pi_{\lambda_{i}}(\theta) \dd \theta,\quad \lambda \in \Lambda, \quad \lambda_{i} \in \Ls,
  \end{equation}
  which is finite for all $\lambda_{i}$ and $\lambda$. Then,
  \begin{equation}
    \label{eq:repr}
    u(\lambda) = \sum_{\ell=1}^{L} \bu_{\ell} f(\lambda_{\ell},\lambda). 
  \end{equation}
  \begin{proof}
    We begin by observing that
    \begin{equation*}
      f(\lambda_{i}, \lambda) = \frac{1}{z(\lambda_{i}) p(\lambda_{i})} \int \psi_{\lambda}(\theta) p(\lambda) \frac{\psi_{\lambda_{i}}(\theta) p(\lambda_{i})}{\sum_{\ell=1}^{L} \psi_{\lambda_{\ell}}(\theta) p(\lambda_{\ell})} \dd{\theta} \leq \frac{p(\lambda) \int \psi_{\lambda}(\theta) \dd{\theta}}{z(\lambda_{i}) p(\lambda_{i})} = \frac{u(\lambda)}{u(\lambda_{i})},
    \end{equation*}
    so $f(\lambda_{i}, \lambda)$ is finite for all $\lambda_{i}$ and $\lambda$. Note that
    \begin{equation*}
      u(\lambda) = \int \psi_{\lambda}(\theta) p(\lambda) \dd \theta = \int \psi_{\lambda}(\theta) p(\lambda) \frac{\sum_{i=1}^{L} \psi_{\lambda_{i}}(\theta) p(\lambda_{i})}{\sum_{\ell=1}^{L} \psi_{\lambda_{\ell}}(\theta) p(\lambda_{\ell})} \dd{\theta},
    \end{equation*}
    and the proof is completed by a re-organization and re-normalization of the terms, and the definition of \eqref{eq:kernel}.
  \end{proof}
\end{proposition}

The function $f(\lambda_{\ell},\lambda)$ in \eqref{eq:kernel}, extends the matrix $\bF$ defined in \eqref{eq:kernel-discr} off the grid, in the sense that $f(\lambda_{i},\lambda_{j}) = \bF_{i,j}$. Accordingly, the result \eqref{eq:repr} in Proposition \ref{prop:reproduce} for $\lambda \in \Ls$ is also satisfied due to Proposition \ref{prop:eigenvec}. 

Our proposed functional estimator of the marginal likelihood replaces the terms on the right-hand side of \eqref{eq:repr} by Monte Carlo estimates:
\begin{equation*}
  \label{eq:femus}
  \hu(\lambda) = \sum_{\ell=1}^{L} \hbu_{\ell} \hf(\lambda_{\ell},\lambda), 
\end{equation*}
where
\begin{equation}
  \label{eq:kernel-mc}
  \hf(\lambda_{i},\lambda) = \frac{1}{N_{i}} \sum_{n=1}^{N_{i}}  \frac{\psi_{\lambda}(\theta_{i}^{(n)}) p(\lambda)}{\sum_{\ell=1}^{L} \psi_{\lambda_{\ell}}(\theta_{i}^{(n)}) p(\lambda_{\ell})},\quad \theta_{i}^{(n)} \sim \pi_{\lambda_{i}}(\theta),
\end{equation}
and $\hbu_{\ell}$ are the estimates obtained using EMUS. It is again a by-product of the whole construction that this is a proper interpolant of the EMUS solutions on $\Ls$. This is formalized in the following Proposition, the proof of which follows rather easily from previous results. 

\begin{proposition}
  \label{prop:inter-EMUS}
  For $\hbu$ as defined via \eqref{eq:emus}, and $\hu(\cdot)$ via \eqref{eq:femus}, we have that $\hu(\lambda_{\ell}) = \hbu_{\ell}$ for each $\lambda_{\ell} \in \Lambda^s$. 
\end{proposition}

Relative to EMUS, the functional estimator requires no further simulation. We can also directly obtain estimates of the gradient of $u(\cdot)$. Letting $\nabla \psi_{\lambda}(\theta)$ denote the gradient of the function with respect to $\lambda$, we directly obtain $\nabla  \hf(\lambda_{i},\lambda)$ from \eqref{eq:kernel-mc} in terms of $\nabla \psi_{\lambda}(\theta_{i}^{(n)})$, and by plugging this in \eqref{eq:femus} we obtain $\nabla \hu(\lambda)$, which can be used to estimate $\nabla u(\lambda)$. Our preferred way to locate local maxima of $u(\cdot)$ when $\lambda$ is low-dimensional is by evaluating the functional estimator on a fine grid, taking advantage of the fact that doing so is relatively cheap since it involves no further sampling.

We turn to the problem of estimating expectations of a test function $\phi(\theta)$:
\begin{equation*}
  \pi(\phi) = \frac{\int \phi(\theta) \pi_{\lambda}(\theta) u(\lambda) \dd \theta \dd \lambda}{\int u(\lambda) \dd \lambda}.  
\end{equation*}
We also define the conditional expectation
\begin{equation}
  \label{eq:c-e}
  \pi_{\lambda}(\phi) = \int \phi(\theta) \pi_{\lambda}(\theta) \dd \theta,
\end{equation}
so that
\begin{equation}
  \label{eq:m-c-e}
  \pi(\phi) = \frac{\int \pi_{\lambda}(\phi) u(\lambda) \dd \lambda}{\int u(\lambda) \dd \lambda}. 
\end{equation}
Our approximation scheme is based on the following identity.

\begin{proposition}
  \label{prop:c-e}
  For the quantities as defined in \eqref{eq:main-q-1} and \eqref{eq:main-q-2}, we have
  \begin{equation*}
    \pi_{\lambda}(\phi)u(\lambda) = \sum_{\ell=1}^{L} \bu_{\ell} h(\lambda_{\ell},\lambda),
  \end{equation*}
  where
  \begin{equation*}
    h(\lambda_{i},\lambda) = \int \phi(\theta) \frac{\psi_{\lambda}(\theta) p(\lambda)}{\sum_{\ell=1}^{L}  \psi_{\lambda_{\ell}}(\theta) p(\lambda_{\ell})} \pi_{\lambda_{i}}(\theta) \dd \theta.
  \end{equation*}
  \begin{proof}
    The proof follows by multiplying and dividing the integrand in \eqref{eq:c-e} by $\sum_{\ell} \bu_{\ell} \pi_{\lambda_{\ell}}(\theta)$ and re-arranging the terms.
  \end{proof}
\end{proposition}

The function $h$ defined in the proposition, whose dependence on $\phi$ is not reflected in the notation, becomes precisely $f$ in \eqref{eq:kernel} when $\phi=1$. On the basis of sampling done only on the grid points $\lambda_{\ell} \in \Ls$, we can obtain estimates $\hu(\lambda_m)$ and $\hh(\lambda_{\ell},\lambda_m)$, for $\lambda_m \in \Le$ and $\lambda_{\ell} \in \Ls$, and for appropriate numerical integration weights $\Delta_m$, we can approximate $\pi(\phi)$ by
\begin{equation*}
  \label{eq:exp-est}
  \frac{\sum_{\ell=1}^{L} \hbu_{\ell} \sum_{m=1}^M \hh(\lambda_{\ell},\lambda_m) \Delta_m}{\sum_{m=1}^M \hu(\lambda_m) \Delta_m}. 
\end{equation*}
The essence of our approach is that we can increase the resolution of $\Le$ to arbitrary precision at minimal additional cost. Hence, for practical purposes, the integral is computed to arbitrary accuracy and the only source of error is that of Monte Carlo.

Arguably, it is only sensible to use Monte Carlo estimates when they have finite variance. We can relate the variance of the functional estimator to that of the estimators in \eqref{eq:kernel-mc} through their relation in \eqref{eq:femus}. Notice that $r_{i}(\lambda)$ defined below relates to $\bR$ defined in Theorem \ref{thm:avar-emus} in the same way that $f(\lambda_i,\lambda)$ relates to $\bF$, in the sense that $r_{i}(\lambda_j) = \bR_{i,j}$ for $\lambda_j\in \Ls$.

\begin{proposition} 
\label{prop:var-femus}
Under the same assumptions as in Theorem \ref{thm:avar-emus} and assuming, in addition, that
\begin{equation*}
  r_{i}(\lambda) = N_{i} \Var(\hf(\lambda_{i}, \lambda)) = \Var\parens*{\frac{\psi_{\lambda}(\theta_{i}^{(1)}) p(\lambda)}{\sum_{\ell=1}^{L} \psi_{\lambda_{\ell}}(\theta_{i}^{(1)}) p(\lambda_{\ell})}} < \infty, \quad i = 1, \dots, L.
\end{equation*}
then as $N \rightarrow \infty$, $\sqrt{N} \left(\hu(\lambda) - u(\lambda)\right)$ converges in distribution to a Gaussian variable with mean 0 and variance $\sigma^2(\lambda)$. 
The asymptotic variance $\sigma^{2}(\lambda)$ satisfies the upper bound
  \begin{equation}
  \label{eq:siglambnd}
   \frac{\sigma^{2}(\lambda)}{u^{2}(\lambda)} \leq 2 \sum_{i = 1}^{L} w_{i}^{-1} \parens*{L  \sum_{j \neq i} \frac{\bR_{i,j}}{\bQ_{i,j}^{2}} + \frac{\bu_{i}^{2}}{u^{2}(\lambda)} r_{i}(\lambda)}.
  \end{equation}
\end{proposition}
The first term on the right hand side of \eqref{eq:siglambnd} can be bounded above using results for the EMUS errors on the grid, such as Theorem \ref{thm:avar-emus}. Regarding the additional assumption that $r_{i}(\lambda) < \infty$ for all $i$ in the statement of Proposition \ref{prop:var-femus}, notice that 
\begin{equation*}
  r_{i}(\lambda)
  \leq \E \bracks*{\parens*{\frac{\psi_{\lambda}(\theta_{i}^{(1)}) p(\lambda)}{\sum_{\ell=1}^{L} \psi_{\lambda_{\ell}}(\theta_{i}^{(1)}) p(\lambda_{\ell})}}^{2}}.
\end{equation*}
Therefore, $r_{i}(\lambda) < \infty$ for all $i$ under the following condition (where we have upper-bounded a term above by 1):
\begin{equation}
  \label{eq:2Int}
  \int \frac{\psi_{\lambda}(\theta)^{2} p(\lambda)^{2}}{\sum_{\ell=1}^{L} \psi_{\lambda_{\ell}}(\theta) p(\lambda_{\ell})} \dd \theta < \infty.
\end{equation}
To gain some understanding of how restrictive or otherwise condition \eqref{eq:2Int} is we return to Example \ref{ex:toy-is}. Consider a $\lambda \notin \Ls$ such that $\psi_{\lambda}(\theta_{i})=0$ for $i=1,5$ and $\psi_{\lambda}(\theta_{i})>0$ for $i=2,3,4$. Then, appropriately adapted to a discrete $\theta$, \eqref{eq:2Int} holds whereas importance sampling between $\pi_{\lambda}(\theta)$ and any of $\pi_{\lambda_{\ell}}(\theta)$ is not well-defined. We can think of this as an ``interpolation'' example whereby the local densities, although with little overlap among them, have sufficient overlap with a density off-the-grid for the method to return high-quality estimates. 

We close the section by connecting to Monte Carlo maximum likelihood (MCML). This approach is based on the importance sampling identity for ratios of normalizing constants \eqref{eq:is-id-grid}, which we discussed earlier. MCML is based on samples generated from a single grid point, say from $\pi_{\lambda_{\ell}}(\theta)$, and yields a functional estimator of $u(\lambda)/u(\lambda_{\ell})$. The resulting estimator is actually $\hf(\lambda_{\ell},\lambda)$ as defined in \eqref{eq:kernel-mc} for $L=1$. For $L=1$, \eqref{eq:2Int} does not hold in Example \ref{ex:toy-is} above. Generally speaking, there is a kind of ``phase transition'' in the performance of the approach between $L=1$ and $L>1$, and MCML requires very strong conditions to achieve consistency \citep[see][]{mcmle-om}, much stronger than those we establish in Section \ref{sec:theory} for $L>1$.

\section{Numerical experiments}
\label{sec:numerics}

We conduct a number of numerical experiments with Gaussian process regression, classification, and crossed random effect models. When we compare our functional EMUS estimator to a griddy Gibbs sampler, we ensure that the overall simulation effort is the same for both. We extrapolate the griddy Gibbs estimates from $\Ls$ to $\Le$ by nearest neighbor interpolation of estimates on $\Ls$.

Our examples involve a high-dimensional $\theta$ and a two-dimensional hyperparameter $\lambda=(\tau_1,\tau_2)$. Besides estimates of $u(\lambda)$ we also report estimates of the two profiles, defined as
\begin{equation}
  \label{eq:profiles}
  \tau_1 \mapsto \max_{(\tau_1,\tau_2)\in \Le} u(\tau_1,\tau_2),\quad  \tau_2 \mapsto \max_{(\tau_1,\tau_2)\in \Le} u(\tau_1,\tau_2),
\end{equation}
which are a non-linear functional of $u(\lambda)$.

\subsection{Gaussian process regression}
\label{subsec:lingp}

The primary objective of this subsection is to test the methodology on a previously studied case where the  marginal likelihood is multimodal due to weak identifiability. A secondary objective is to consider a situation where the marginal likelihood can be computed exactly, allowing for a more refined numerical analysis of the errors induced by our methods. The underlying model is the Gaussian process regression setup previously considered in \cite{yao2022bayesian}, which can be consulted for further details. The $n = 32$ data points $(x_{i},y_{i})$ are shown in  Figure \ref{fig:lingp-data}, with the true conditional expectation $\E[y | x]$ shown as the solid line. We posit a misspecified homoscedastic Gaussian processes regression model, where the residual variance is assumed known and equal to 1/16, and the unknown regression function is a Gaussian process with covariance kernel
\begin{equation*}
  (\tau_1 / \tau_2) \exp(-\tau_2 |x' - x''|_{2}^{2}).
\end{equation*}
The latent variable $\theta$ is the $n$-dimensional vector that corresponds to the point evaluations of the Gaussian process on the $x_{i}$'s, and $\lambda = (\tau_1, \tau_2)$. We apply an improper flat prior to the log-parameters. The resulting $u(\lambda)$ is bimodal and can be computed exactly (see Figures \ref{fig:lingp-gt-variate} and \ref{fig:lingp-profs}).

\begin{figure}[htbp]
  \centering
  \includegraphics[scale=.66]{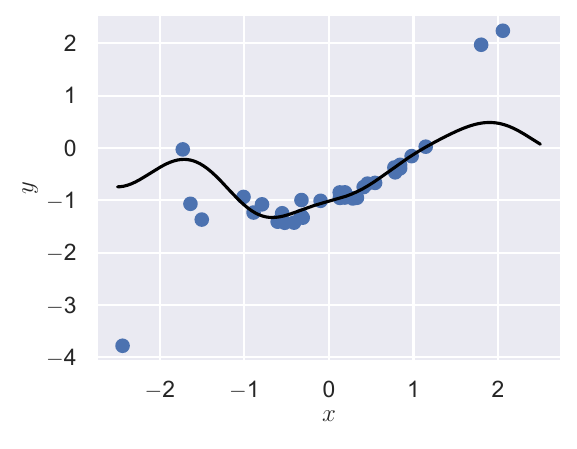}
  \caption{The data with the generating regression function $\E[y|x]$.}
  \label{fig:lingp-data}
\end{figure}

Our sampling design $\Ls$ consists of a $17 \times 17$ grid with uniform spacing in $(\log \tau_1, \log \tau_2)$ space, and $\Le$ is the $33 \times 33$ superset obtained in an analogous manner. At each $\lambda_{\ell} \in \Ls$, we obtain 16 \emph{independent} samples from $p(\theta | y, \lambda_{\ell})$. Figure \ref{fig:lingp-gt-variate} shows an  estimate of $u(\lambda_{m})$ for $\lambda_{m} \in \Le$ for functional EMUS and griddy Gibbs. Figure \ref{fig:lingp-profs} illustrates the variability of the method over 128 runs for each of the two profiles. 

\begin{figure}[htbp]
  \centering
  \includegraphics[scale=.66]{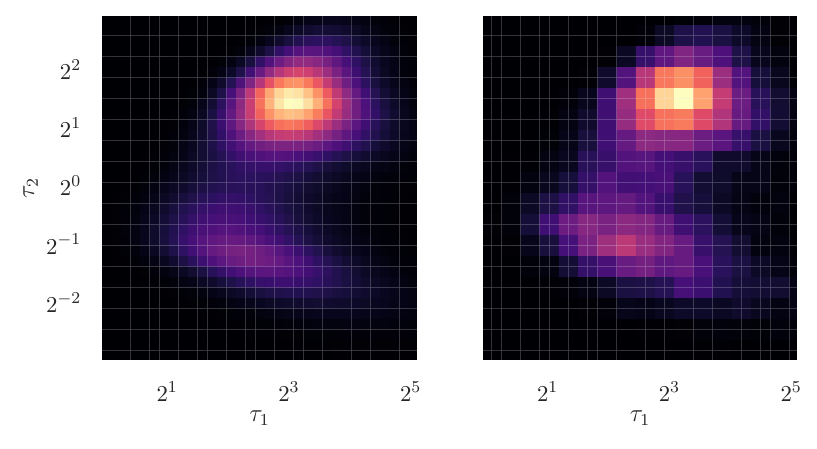}
  \caption{Typical estimates of $u(\lambda)$ according to functional EMUS (left) and extrapolated griddy Gibbs (right, see text for how this is obtained), both on a $17 \times 17$ simulation grid and a $33 \times 33$ evaluation grid.}
  \label{fig:lingp-gt-variate}
\end{figure}

\begin{figure}[htbp]
  \centering
  \includegraphics[scale=.66]{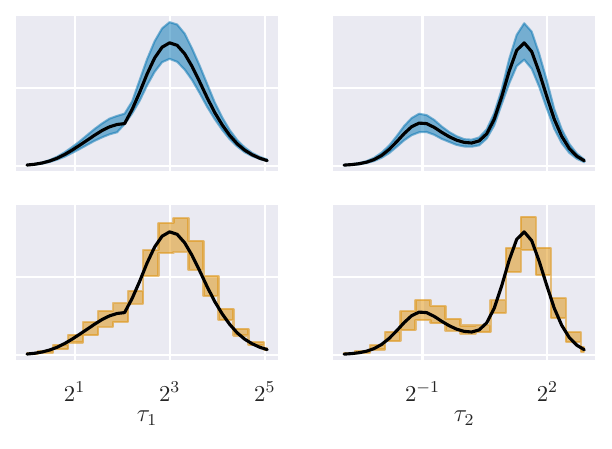}
  \caption{True profiles for $\tau_{1},\tau_{2}$ (solid black line), and 75\%-intervals of functional EMUS estimator (blue) and griddy Gibbs (orange) over 128 runs of each method.}
  \label{fig:lingp-profs}
\end{figure}

We now investigate the numerical performance of functional EMUS in two different asymptotic regimes, which we consider from a theoretical point of view in Section \ref{sec:theory}. One regime (that we later call fixed-grid) keeps $\Lambda^s$ fixed and increases sampling effort at each location in $\Ls$; the other (that we later call dense-grid) holds the sampling effort per grid point constant, and increases the density of the grid. In both cases $N$ denotes the total number of samples generated. We consider two fixed-grid regimes, one with lower and another with higher resolution. Results shown in Figure \ref{fig:lingp-scale} demonstrate that functional EMUS can reach the standard Monte Carlo asymptotic rate in both regimes. The comparatively higher error for the fixed low-density grid (green in Figure \ref{fig:lingp-scale}) highlights the importance of a sufficiently dense sampling grid. In Figure \ref{fig:lingp-scale}, the slope for both the high-resolution fixed design (blue) and the dense design (orange) eventually becomes the Monte Carlo rate of $1/\sqrt{N}$, while the slope for the low-resolution fixed design (green) is only slowly bending towards its supposed asymptotic rate. Therefore, it is more robust to increase the density of the grid, as long as the cost of storing a larger stochastic matrix and solving for its leading eigenvector is not prohibitive. In practice, we recommend increasing $L$ until the computational effort of the matrix decomposition dominates the effort spent on sampling, and increasing sampling intensity on the grid the rest of the way.

\begin{figure}[htbp]
  \centering
  \includegraphics[scale=.66]{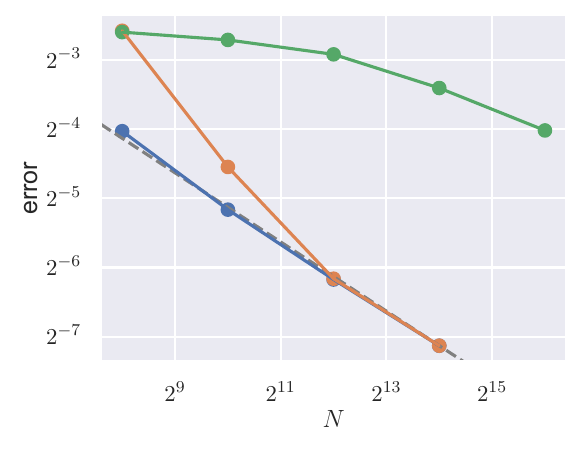}
  \caption{Total number of samples $N$ vs. error, defined as $\E |\hbu^{e} / |\hbu^{e}|_{1} - \bu^{e} / |\bu^{e}|_{1}|_{2}$, where $\hbu^{e}$ and $\bu^{e}$ are the evaluations of $\hu$ and $u$ on $\Le$, and $\Le$ is a $33 \times 33$ grid. The error is approximated by averaging over 128 estimates. The blue line corresponds to the regime with a fixed $17 \times 17$ grid and increasing $N$, the green line to the regime with a fixed $5 \times 5$ grid and increasing $N$, and the orange line to increasingly fine grids with $16$ samples per grid point. The grids of the orange line span from $5 \times 5$ on the left to $33 \times 33$ (same as $\Le$) on the right. The slope of the dashed line corresponds to the asymptotic Monte Carlo rate of $1/\sqrt{N}$.}
  \label{fig:lingp-scale}
\end{figure}

\subsection{Gaussian process classification}
\label{subsec:loggp}

Here we consider an example from \cite{titsias2018auxiliary}, where sampling from $p(\theta | y,\lambda)$, and even more so  from $p(\theta, \lambda | y)$, is challenging and computationally expensive. We use the ``Heart Disease'' dataset, collected by \cite{heart_disease_45} and hosted on the UC Irvine machine learning Repository. The model is a Gaussian process classifier with logistic link and binary outcome. In this context the marginal likelihood is intractable.

We use the same Gaussian process as in Section \ref{subsec:lingp}. We also use the same sampling design grids. At each $\lambda \in \Ls$, we run an MCMC chain with equilibrium distribution $p(\theta | \lambda, y)$ for 256 iterations, discarding half to avoid initialization bias. This Markov chain is constructed according to the \emph{marginal algorithm} of \cite{titsias2018auxiliary}. To assess the precision of the functional EMUS estimator, we obtain an estimate from much longer MCMC runs of 8192 samples at each $\lambda \in \Ls$, with 256 discarded for initialization. This high-precision estimate of $u(\lambda)$ is shown in Figure \ref{fig:loggp-gt-variate} together with a typical estimate produced by functional EMUS. The variability of the functional EMUS estimate of the two profiles over 128 runs is shown in Figure \ref{fig:loggp-profs}.

\begin{figure}[htbp]
  \centering
  \includegraphics[scale=.66]{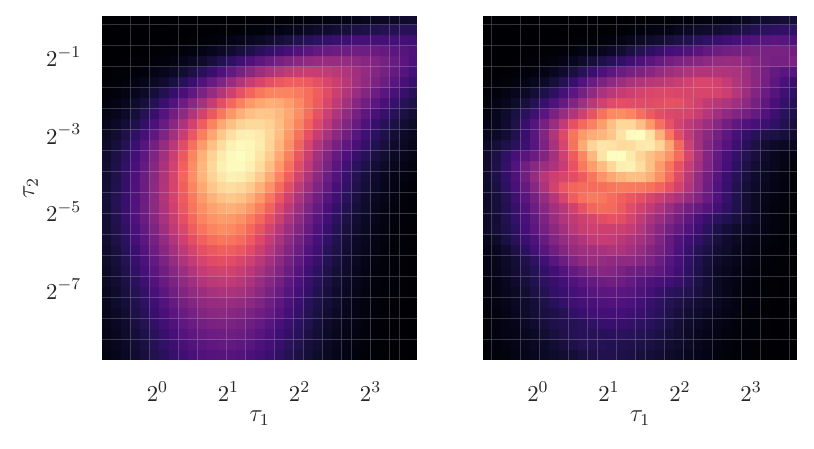}
  \caption{An expensive high-precision estimate of $u(\lambda)$ by functional EMUS (left, 8192 MCMC samples per grid point), and a typical estimate (right, 256 samples per grid point), both on a $17 \times 17$ simulation grid and a $33 \times 33$ evaluation grid.}
  \label{fig:loggp-gt-variate}
\end{figure}

\begin{figure}[htbp]
  \centering
  \includegraphics[scale=.66]{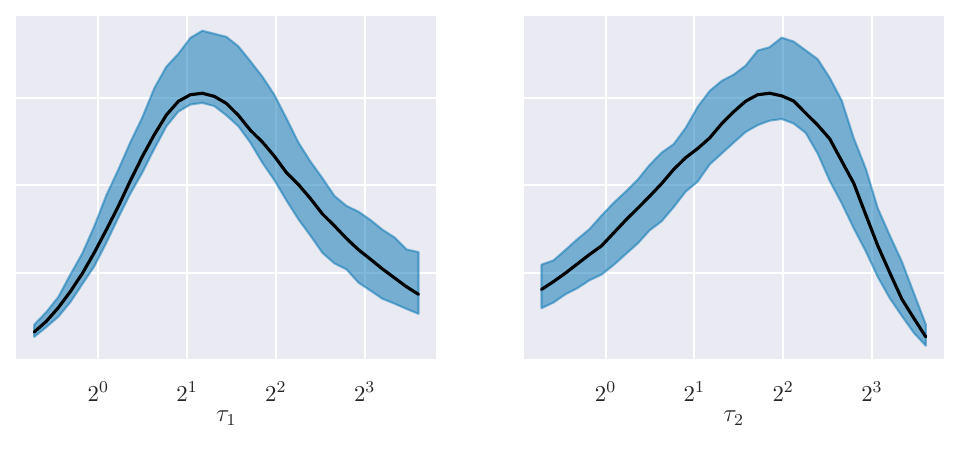}
  \caption{Profiles for $\tau_{1}, \tau_{2}$ (solid black line black, estimated to high precision as in Figure \ref{fig:loggp-gt-variate}) and 75\%-intervals of the functional EMUS estimator over 128 runs.}
  \label{fig:loggp-profs}
\end{figure}


\subsection{Crossed random effect models}
\label{subsec:crossed}

Here we consider an example from \cite{papaspiliopoulos2023scalable} that involves high-dimensional random effects and a binary output with logistic link. We focus on a relatively simple setting as the goal is to assess the performance of functional EMUS as the dimension of $\theta$ gets large. The data structure is as follows: there are two categorical factors, where for simplicity we set both to have the same number of levels $d$. For certain combinations of level $i$ of factor 1 (e.g. a customer) and level $j$ of factor 2 (e.g. a product) we observe a binary $y_{ij}$ (e.g. whether the customer rated the product favorably). We only observe a fraction of these crossings. In our simulation study we miss observations completely at random with probability 0.5. The linear predictor is $c + \alpha_{i} +\beta_{j}$ for $i,j = 1,\ldots,d$, hence each level of each factor is assigned a random effect, with $\alpha_{i} \sim \Nor(0,\tau_{1}^{-1})$ and $\beta_{j} \sim \Nor(0,\tau_{2}^{-1})$, and $c$ is an intercept with a flat prior. In our notation $\theta = (c,\alpha_1,\ldots,\alpha_d,\beta_1,\ldots,\beta_d)$ and $\lambda = (\tau_1,\tau_2)$, where $\tau_{1}, \tau_{2} \sim \Ga(1/2, 1/2)$.  In the simulation study below the data are generated from the model with $\tau_1=\tau_2=1$. We consider the increasing dimensions $d \in \braces{10^2, 10^3, 10^4}$.

We adopt the provably scalable computational framework of \cite{papaspiliopoulos2023scalable} for generalized linear mixed models that combines ideas from collapsed Gibbs sampling for linear mixed models in \cite{crossed} with the marginal sampler of \cite{titsias2018auxiliary}, which was used also in the Gaussian process classification examples. The resultant method is meant to sample from $p(\theta | y,\lambda)$ with an efficiency (e.g. measured by effective sample size) that does not degrade with $d$. 

Our sampling design $\Ls$ consists of a $17 \times 17$ grid with uniform spacing in $(\log \tau_{1}, \log \tau_{2})$ space, which we further resolve to $33 \times 33$ in the evaluation grid $\Le$. Due to the data generating mechanism, larger $d$ are associated with larger $n$, resulting in an increasingly peaked marginal likelihood; this can be seen by plotting the posterior standard deviation of $p(\lambda | y)$ versus $d$ in Figure \ref{fig:xfx-err}. To account for this we contract the range of the grids at rate $1/\sqrt{d}$. At each grid location we generate a single MCMC sample after 8 samples as burn in. The performance of functional EMUS versus $d$ is shown in Figure \ref{fig:xfx-err}. Since the ground truth is not known in this example, it is approximated by a functional EMUS estimate with 128 samples in each location of $\Ls$. The results show increasing precision as $d$ increases. 

\begin{figure}[htbp]
  \centering
  \includegraphics[scale=.66]{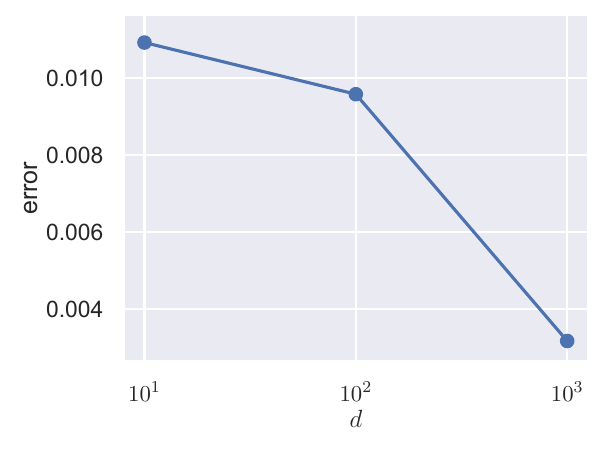}
  \includegraphics[scale=.66]{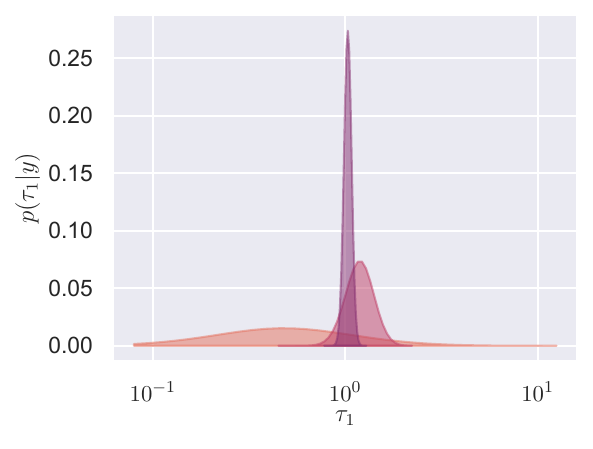}
  \caption{Performance of functional EMUS (left, defined as in Figure \ref{fig:lingp-scale}) and marginals $p(\tau_{1} | y)$ for the three settings of $d$, showing posterior concentration (right).}
  \label{fig:xfx-err}
\end{figure}

\subsection{Optimal Design}
\label{sec:design}

A natural question is how to choose $\Ls$ optimally, and how to distribute sampling effort across the grid points. Thus far we have taken $\Ls$ as given and distributed effort equally. Here and the in the associated Supplement \ref{app:design}, we outline a principled approach for sequential optimal design and show results for the toy Example \ref{ex:toy-ms}. More thorough evaluation of the approach will be considered in future work.

There are two main ingredients to our optimal design approach. First, the extrapolation of functional EMUS estimates beyond the simulation grid, and second, an expression for $\lim_{N\rightarrow \infty} \E[N|\hbu - \bu|_{2}^{2}]$. This was first obtained in \cite{thiede2016eigenvector} and it is given in Supplement \ref{app:design}. We can analytically optimize the expression over the sampling intensities. The resulting formula for the optimal weights can be extended to give a sampling intensity, $w(\lambda)$, for any $\lambda$. The right panel of Figure \ref{fig:toy-design} shows the optimal weights in the context of Example \ref{ex:toy-ms}. For this simple example, the intensities can be computed accurately by quadrature. Notice that the optimal sampling intensity is 0 in the tails of $u(\lambda)$, though it may be large even if $u(\lambda)$ is small, e.g. in the valley in between two modes. 

Our sequential strategy for estimating the optimal weights is as follows: we define a dense evaluation grid; start with a coarse simulation grid to compute a functional EMUS estimate; use the estimate to approximate the optimal intensities on the evaluation grid; allocate new sampling effort according to these intensities and adjust the simulation grid if needed; recompute the functional EMUS estimates using the new samples; and we iterate the above procedure until the available simulation budget is spent. Experimentation has suggested the following two small modifications. One is to stabilize the weights by taking their square root. Another is to include a mechanism that over-samples (relative to the optimized weights) in locations that have been under-sampled in previous iterations, and under-samples accordingly. The resulting formulae are included in Supplement \ref{app:design}. We use the pivotal sampling algorithm \citep{deville1998unequal} for generating the number of samples per grid point according to the sampling probabilities obtained as described above. 

\begin{figure}[htbp]
  \centering
  \includegraphics[scale=.66]{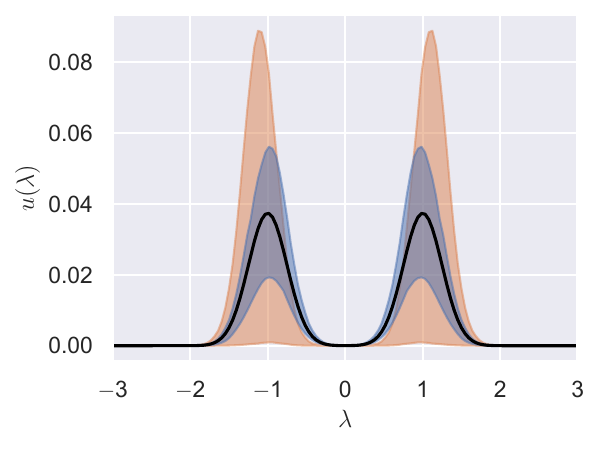}
  \includegraphics[scale=.66]{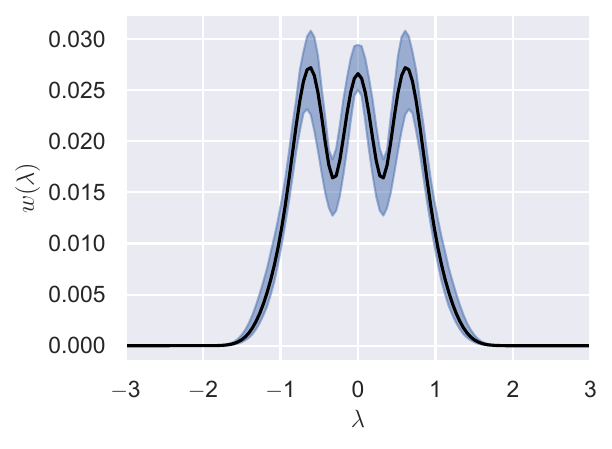}
  \caption{Optimal sampling in the setting of Example \ref{ex:toy-ms}. Left: $u(\lambda)$ (black) and 75\% intervals for the iterative procedure (blue) and sampling on a uniform grid (orange). Right: optimal sampling intensity $w(\lambda)$ (black, as computed accurately by quadrature) and 75\% intervals for the iterative procedure, estimated over 128 runs, half of which are obtained by symmetry. The x-axis has been clipped to $[-3, 3]$, down from $\Lambda = [-6,6]$.}
  \label{fig:toy-design}
\end{figure}

We test this procedure in the setting of Example \ref{ex:toy-ms}, with $\Lambda = [-6,6]$ and $q = \tau = 2^5$. This results in a bimodal $u$, with all mass located within $[-2,2]$, and a region of 0 density around 0. The corresponding optimal weight function, computed accurately through quadrature, is concentrated within $[-1.5,1.5]$, with a mode also at 0. At each of the 8 iterations of the procedure, we choose 8 elements of $\Le$, and generate 8 samples at each of those elements. We compare this to estimates that result from spreading the same total number of samples over 8 uniformly spaced elements of $\Lambda$. The left panel of Figure \ref{fig:toy-design} demonstrates that the iterative procedure can strongly reduce variance relative to the equal allocation of effort over a uniform $\Ls$.

\section{Consistency theory}
\label{sec:theory}

In this section we establish two fundamental consistency properties of the functional estimator. Our theory uses an existing consistency result for EMUS which is stated as Theorem \ref{thm:cons-emus} below, the proof of which can be obtained by the results in the Supplement of \cite{thiede2016eigenvector}.  

\begin{theorem}
    \label{thm:cons-emus}
    Let $\bu$ and $\hbu$ be defined as in Section \ref{ssec:emus-gibbs}, with normalization (SumToL). Under assumptions (GridIrred), (IndSample) and (PosWeight),
    \begin{equation*}
        \lim_{N \to \infty} |\hbu_{\ell} - \bu_{\ell}| = 0 \; (a.s.), \quad \lim_{N \to \infty} \E[|\hbu_{\ell}-\bu_{\ell}|] = 0, \quad \ell = 1, \ldots, L.
    \end{equation*}
\end{theorem}

Our intention is not to give the most general statements possible, but to provide theoretical justification for what we have observed in numerical experiments, and convey the intuition behind the consistency properties of the estimator. Some of the assumptions we make below, especially those in Section \ref{sec:dense-g}, are chosen to simplify the derivation and communication of the main results but can likely be relaxed. In view of these considerations, we draw the same number of samples from each local density: 
\begin{equation}
  \label{eq:eqsample}
  (EqSample) \quad N_{\ell} = N/L, \quad \ell = 1, \ldots, L,
\end{equation}
in which case (PosWeight) is trivially satisfied. We will also assume throughout that $u(\lambda)$ and $\psi_{\lambda}(\theta) p(\lambda)$ for all $\theta$, are continuous in $\lambda$. 

In the sequel we establish uniform almost sure and in expectation fixed-grid consistency, where $L$ is fixed and $N\to \infty$, and dense-grid consistency, where $N/L$ is fixed and $L\to \infty$. The precise statements are given in Theorems \ref{thm:fixed-g} and \ref{thm:dense-g} below. A standard argument shows that uniform almost sure convergence implies convergence of maximizers; for completeness we state and prove this in Proposition \ref{prop:const-max} [Supplement]. We have also obtained analogous  consistency theory for the estimator of expectations \eqref{eq:exp-est} but we have not included these results due to space limitations.

\subsection{Fixed-grid consistency}
\label{sec:fixed-g}

The main result in this section is Theorem \ref{thm:fixed-g} below that establishes uniform convergence of the functional estimator for fixed lattice $\Ls$. Its proof is based on the following key lemma.

\begin{lemma}
  \label{lem:kernel}
  Consider the following two integrability conditions: 
  \begin{gather}
    (SupInt) \quad \int \sup_{\lambda \in \Lambda} \psi_{\lambda}(\theta) p(\lambda) \dd \theta < \infty, \\
    (Sup2Int) \quad \int \frac{\sup_{\lambda \in \Lambda} \psi_{\lambda}(\theta)^{2} p(\lambda)^{2}}{\sum_{\ell=1}^{L} \psi_{\lambda_{\ell}}(\theta) p(\lambda_{\ell})} \dd \theta < \infty.
  \end{gather}
  Assuming (SupInt) we have that  
  \begin{equation*}
    \lim_{N\to\infty} \sup_{\lambda \in \Lambda} |\hf(\lambda_{\ell},\lambda) - f(\lambda_{\ell},\lambda)| = 0 \; (a.s.), 
    \quad \ell = 1, \ldots, L,
  \end{equation*}
  whereas assuming (Sup2Int) we have that 
  \begin{equation*}
    \lim_{N\to\infty} \E\bracks*{\sup_{\lambda \in \Lambda} |\hf(\lambda_{\ell},\lambda) - f(\lambda_{\ell},\lambda)|^{2}} = 0, \quad 
    \ell = 1, \ldots, L,.
  \end{equation*}
\end{lemma}
(Sup2Int) is stronger than (SupInt) and relates to \eqref{eq:2Int}, which was discussed in Section \ref{sec:nystrom} in relation to $\Var[\hf(\lambda_{i}, \lambda)] < \infty$.

The proof of the lemma appeals to the strong law of large numbers on a Banach space and it is included in the Supplement, where it is also shown that, for many examples in the important class of Gaussian process priors, the assumptions (SupInt) and (Sup2Int) are satisfied.

\begin{theorem}
  \label{thm:fixed-g}
  For $u(\cdot)$ and $\hu(\cdot)$ defined in \eqref{eq:repr} and \eqref{eq:femus} respectively, and assuming (SupInt), (SumToL), (GridIrred), (IndSample), (EqSample)
  \begin{equation*}
    \lim_{N\to\infty} \sup_{\lambda \in \Lambda} |\hu(\lambda) - u(\lambda)| = 0 \; (a.s.),
  \end{equation*}
  In addition, under (Sup2Int),
  \begin{equation*}
    \lim_{N\to\infty} \E\bracks*{\sup_{\lambda \in \Lambda} |\hu(\lambda) - u(\lambda)|} = 0.
  \end{equation*}
  \begin{proof}
    Notice first that
    \begin{equation*}
      \sup_{\lambda} |\hu(\lambda) - u(\lambda)|  \leq \sum_{\ell=1}^{L} \sup_{\lambda} | \hbu_{\ell} \hf(\lambda_{\ell},\lambda) - \bu_{\ell} f(\lambda_{\ell},\lambda)|.
    \end{equation*}
    Since $L$ is fixed, we now focus on each of the terms in the sum. Then, by adding and subtracting a term, and exploiting that $\hbu_{\ell} \leq L$ due to (SumToL), we get
    \begin{equation*}
      \sup_{\lambda} | \hbu_{\ell} \hf(\lambda_{\ell},\lambda) -  \bu_{\ell} f(\lambda_{\ell},\lambda)| \leq |\hbu_{\ell} - \bu_{\ell}| \sup_{\lambda} f(\lambda_{\ell},\lambda) + L \sup_{\lambda} |\hf(\lambda_{\ell},\lambda) - f(\lambda_{\ell},\lambda)|. 
    \end{equation*}
    We establish convergence (both almost surely and in L1) of the first term by Theorem \ref{thm:cons-emus} and of the second by Lemma \ref{lem:kernel}. 
  \end{proof}
\end{theorem}

\subsection{Dense-grid consistency}
\label{sec:dense-g}

In this section we consider a different asymptotic regime where the amount of Monte Carlo sampling per grid is kept fixed but the grid $\Ls$ becomes dense in $\Lambda$. In view of the variable grid, we replace (GridIrred) with an assumption on the continuous Markov chain:
\begin{equation}
  \label{eq:irred}
  (Irred) \quad \text{The Markov chain with transition kernel \eqref{eq:infemus} is irreducible in $\Lambda$.}
\end{equation}
In conjunction with (SupInt) and by way of Lemma \ref{lem:gapexist} [Supplement], (Irred) implies that the continuous chain has positive spectral gap, which is then inherited by the projected chain under some smoothness conditions, as seen in Lemma \ref{lem:convgap} [Supplement].

To simplify notation and exposition and focus on the key reasons the result we will establish holds, we apply the following simplifying conventions throughout the section:
\begin{equation}
  \label{eq:reggrid}
  (1dRegGrid) \quad 
  \Lambda = [0,1], \quad 
  \Ls = \braces{\ell/L}_{\ell=1}^{L}.
\end{equation}
Notably, in order to obtain sensible limits in this context, the $(SumToL)$-normalization is critical. The normalization is inherited by $u(\lambda)$ and $\hu(\lambda)$ through \eqref{eq:repr} and \eqref{eq:femus}. Since in this section $L$ will be taken to its limit $L\to \infty$, it is important to recognize the dependence of $u(\lambda)$ on $L$, although the notation does not make this explicit. In fact, due to Lemma \ref{lem:cvgdens} [Supplement],
\begin{equation}
  \label{eq:l-marg}
  \lim_{L \to \infty} u(\lambda) =   \pi(\lambda) = \frac{z(\lambda) p(\lambda)}{\int z(\lambda') p(\lambda') \dd \lambda'},
\end{equation}
where in keeping with the notation in \eqref{eq:main-q-2}, $\pi(\lambda)$ is the marginal density of $\lambda$.

The main result in this section is Theorem \ref{thm:dense-g}, which establishes uniform convergence. A key stepping stone is that the norm of EMUS errors does not explode, in spite of the increasing number of elements in the EMUS vector, which we prove below. An essential ingredient to the proof is the imposition of stronger smoothness assumptions on the model, beyond the mere continuity required in the fixed-grid case. This smoothness effectively allows the method to borrow information from samples at nearby grid points when learning the function at a given grid point. 

\begin{theorem}
  \label{thm:stable-errors}
  We assume (SumToL), (1dRegGrid), (IndSample), (Irred), (SupInt). Moreover, we assume
  \begin{equation}
    \label{eq:lipint}
    (LipInt) \quad
    \int \sup_{\lambda', \lambda''} \abs*{\frac{\psi_{\lambda'}(\theta) p(\lambda') - \psi_{\lambda''}(\theta) p(\lambda'')}{\lambda' - \lambda''}} \dd{\theta} < \infty,
  \end{equation}
  and 
  \begin{equation}
    \label{eq:cond}
    (Cond) \;
    \sup_{\lambda, \lambda'} \int \pi_{\theta}(\lambda')^{8} \pi_{\lambda}(\theta) \dd{\theta} < \infty, \quad 
    \sup_{\lambda, \lambda', \lambda''} \int \abs*{\frac{\pi_{\theta}(\lambda') - \pi_{\theta}(\lambda'')}{\lambda' - \lambda''}}^{8} \pi_{\lambda}(\theta) \dd{\theta} < \infty.
  \end{equation}
  Then,
  \begin{equation*}
    \lim_{L\to\infty} \frac{|\hbu - \bu|_{2}}{\sqrt{L}} = 0 \; (a.s.), \quad
    \limsup_{L\to\infty} \E[|\hbu - \bu|_{2}^{2}] < \infty.
  \end{equation*}
  \begin{proof}
    First, note that \eqref{eq:lipint} implies Lipschitz-continuity of $u(\lambda)$ by Jensen's inequality.
    The starting point of the proof is the perturbation identity
    \begin{equation*}
       \hbu-\bu = \braces{\hbu\T (\hbF-\bF)(\bI-\bF)^{\#}}\T,
    \end{equation*}
    where the group inverse $\bB^{\#}$ of a matrix $\bB$ is the unique matrix for which $\bB \bB^{\#} \bB = \bB$, $\bB^{\#} \bB \bB^{\#} = \bB^{\#}$ and $\bB \bB^{\#} = \bB^{\#} \bB$. See \cite{golub1986using} for further discussion of the group inverse in the context of Markov chain perturbation theory. In this perturbation identity, the random contribution $\hbu\T (\hbF-\bF)$ is multiplied by the deterministic group inverse $(\bI-\bF)^{\#}$, which suggests a strategy of analyzing the terms separately, thereby containing any randomness in the analysis of the first term.
    
    The key step towards analyzing the random part consists of substituting $\hbu = \hbF\T \hbu$, so that
    \begin{equation*}
      \hbu-\bu = \bA \hbu, \quad \bA = \braces{\hbF(\hbF-\bF)(\bI-\bF)^{\#}}\T,
    \end{equation*}
    and, by subadditivity and submultiplicity of norms,
    \begin{equation*}
        |\hbu-\bu|_{2} \leq |\bA|_{2} |\hbu|_{2}.
    \end{equation*}
    Moreover, by submultiplicity, $|\bA|_{2} \leq |\hbF(\hbF-\bF)|_{2} |(\bI-\bF)^{\#}|_{2}$. For the first term, by Lemma \ref{lem:randopnorm} [Supplement], $|\hbF(\hbF-\bF)|_{2} \xrightarrow{a.s.} 0$ as well as $\limsup_{L \to \infty} L^{2} \E[|\hbF(\hbF-\bF)|_{2}^{4}] < \infty$. Notice that while $|\hbF-\bF|_{2}$ does \emph{not} vanish in the $L$-limit, $|\bb\T (\hbF-\bF)|_{2}$ \emph{does} vanish for appropriately smooth vectors $\bb$, due to $\hbF-\bF$ having independent rows whose entries vary smoothly and are 0 in expectation. Introduction of the additional factor of  $\hbF$, allows us to exploit this smoothness property, which is a consequence of (Cond). For the second term, by Lemma \ref{lem:convgap} [Supplement], $\limsup_{L \to \infty} |(\bI-\bF)^{\#}|_{2} < \infty$. Here, we exploit reversibility of $\bF$ to bound $(\bI-\bF)^{\#}$ in terms of its spectral gap, which itself can be controlled in the $L$-limit. We again leverage (Cond) to ensure that the spectral gap remains stable as $L$ grows, while (Irred) and (SupInt) guarantee the existence of the limiting spectral gap. In conjunction,
    \begin{equation*}
        |\bA|_{2} \xrightarrow{a.s.} 0, \quad
        \limsup_{L \to \infty} L^{2} \E[|\bA|_{2}^{4}] < \infty.
    \end{equation*}
    With $|\bA|_{2}$ controlled, we can proceed with proving the claim. Beginning with the almost sure statement, by the triangle inequality,
    \begin{equation}
      \label{eq:errident}
      |\hbu - \bu|_{2}
      \leq |\hbu|_{2} |\bA|_{2}
      \leq (|\hbu - \bu|_{2} + |\bu|_{2}) |\bA|_{2},
    \end{equation}
    so $|\hbu - \bu|_{2} (1 - |\bA|_{2}) \leq |\bu|_{2} |\bA|_{2}$. 
    Since $(1 - |\bA|_{2}) \xrightarrow{a.s.} 1$ and $\limsup_{L \to \infty} L^{-1/2} |\bu|_{2} \leq \lim_{L \to \infty} \sn{\bu} = \sn{\pi}$ by Lemma \ref{lem:cvgdens} [Supplement],
    \begin{equation*}
      \limsup_{L \to \infty} L^{-1/2} |\hbu - \bu|_{2}
      \leq \limsup_{L \to \infty} L^{-1/2} |\bu|_{2} \lim_{L \to \infty} |\bA|_{2} = 0 \; (a.s.),
    \end{equation*}
    With the limit supremum being almost surely equal to the lower bound, $L^{-1/2} |\hbu - \bu|_{2} \xrightarrow{a.s.} 0$.
    Proceeding with convergence in expectation, we note that \eqref{eq:errident} and subadditivity and submultiplicity of norms implies 
    \begin{equation*}
      |\hbu - \bu|_{2}
      \leq |\bu|_{2} |\bA|_{2} + |\hbu|_{2} |\bA|_{2}^{2}
      \leq \sn{\bu} \sqrt{L} |\bA|_{2} + L |\bA|_{2}^{2},
    \end{equation*}
    where we applied the bounds $|\bu|_{2} \leq \sqrt{L} \sn{\bu}$ and $|\hbu|_{2} \leq |\hbu|_{1} = L$. We now raise the power on the LHS by applying Young's inequality, yielding
    \begin{equation*}
      |\hbu - \bu|_{2}^{2}
      \leq 2 \sn{\bu}^{2} L |\bA|_{2}^{2} + 2 L^{2} |\bA|_{2}^{4}.
    \end{equation*}
    Recalling that $\sn{\bu}^{2} \to \sn{\pi}^{2}$ by Lemma \ref{lem:cvgdens} [Supplement],
    \begin{equation*}
        \limsup_{L \to \infty} \E |\hbu - \bu|_{2}^{2}
         \leq 2 \lim_{L \to \infty} \sn{\bu}^{2} \limsup_{L \to \infty} L \E |\bA|_{2}^{2}      
         + 2 \limsup_{L \to \infty} L^{2} \E |\bA|_{2}^{4} 
         < \infty.
    \end{equation*}
  \end{proof}
\end{theorem}

On the basis of Theorem~\ref{thm:stable-errors}, it is apparent that $|\hu(\lambda) - u(\lambda)|$ vanishes for any given $\lambda$, when considering a subsequence of sampling grids that include $\lambda$. We directly proceed to developing an analogue to Theorem \ref{thm:fixed-g} for the dense grid regime in order to prove uniform convergence. As in Theorem \ref{thm:fixed-g}, we adapt generic uniform laws of large numbers to ensure that functional EMUS converges on the gaps within the grid. We state the generic ULLN in Theorem \ref{thm:ullnlip} [Supplement].

Some remarks about the additional assumptions in  Theorem \ref{thm:dense-g} below. First, $(SupCond)$ implies $(Cond)$ from Theorem \ref{thm:stable-errors} due to  Jensen's inequality.  Analogously $(Sup8Int)$ and $(Lip8Int)$ imply $(SupInt)$ and $(LipInt)$ respectively. If $\lambda \mapsto \psi_{\lambda}(\theta) p(\lambda)$ is continuously differentiable for every $\theta$, in order for (Sup8Int) and (Lip8Int) to hold it is sufficient that there exists some $\delta > 0$ such that 
\begin{equation*}
    \int \sup_{\lambda, \lambda' \in [\lambda \pm \delta]} \braces{(\psi_{\lambda}(\theta) p(\lambda))^{8} / (\psi_{\lambda'}(\theta) p(\lambda'))^{7}} \dd{\theta} < \infty
\end{equation*}
and
\begin{equation*}
    \int \sup_{\lambda, \lambda' \in [\lambda \pm \delta]} \braces{|\partial_{\lambda} \psi_{\lambda}(\theta) p(\lambda)|^{8} / (\psi_{\lambda'}(\theta) p(\lambda'))^{7}} \dd{\theta} < \infty.
\end{equation*}
Notice that the result establishes uniform convergence of the functional estimator to the marginal density $\pi(\lambda)$. 

\begin{theorem}
  \label{thm:dense-g}
  We assume (SumToL), (1dRegGrid), (IndSample), (Irred) and additionally
  \begin{equation}
      (SupCond) \quad
      \sup_{\lambda} \int \sn{\pi_{\theta}}^{8} \pi_{\lambda}(\theta) \dd{\theta} < \infty, \quad 
      \sup_{\lambda} \int \lip{\pi_{\theta}}^{8} \pi_{\lambda}(\theta) \dd{\theta} < \infty,
  \end{equation}
  and
  \begin{gather}
    (Sup8Int) \quad
    \limsup_{L \to \infty} \int \frac{\sn{\lambda \mapsto \psi_{\lambda}(\theta) p(\lambda)}^{8}}{\parens*{\sum_{\ell=1}^{L} \psi_{\lambda_{\ell}}(\theta) p(\lambda_{\ell})}^{7}} \dd \theta < \infty, \\
    (Lip8Int) \quad 
    \limsup_{L \to \infty} \int \frac{\lip{\lambda \mapsto \psi_{\lambda}(\theta) \pi(\lambda)}^{8}}{\parens*{\sum_{\ell=1}^{L} \psi_{\lambda_{\ell}}(\theta) p(\lambda_{\ell})}^{7}} \dd \theta < \infty.
  \end{gather}
  Then,
  \begin{equation*}
      \lim_{L\to\infty} \sup_{\lambda \in \Lambda} |\hu(\lambda) - \pi(\lambda)| = 0 \; (a.s.), \quad 
      \lim_{L\to\infty} \E\bracks*{\sup_{\lambda \in \Lambda} |\hu(\lambda) - \pi(\lambda)|} = 0.
  \end{equation*}
  \begin{proof}
    We first notice that by the triangle inequality
    \begin{equation*}
      \sup_{\lambda} |\hu(\lambda) - \pi(\lambda)| \leq \sup_{\lambda} |u(\lambda) - \pi(\lambda)| + \sup_{\lambda} |\hu(\lambda) - u(\lambda)|,
    \end{equation*}
    where the first term vanishes deterministically by Lemma \ref{lem:cvgdensunif} [Supplement]. As for the random term, by the triangle inequality and Cauchy-Schwarz,
    \begin{equation*}
      \begin{aligned}
        \sup_{\lambda} |\hu(\lambda) - u(\lambda)|
        & = \sup_{\lambda} \abs*{\sum_{\ell=1}^{L} \hu(\lambda_{\ell}) \hf(\lambda_{\ell}, \lambda) - u(\lambda)} \\
        & \leq \sup_{\lambda} \abs*{\sum_{\ell=1}^{L} (\hu(\lambda_{\ell}) - u(\lambda_{\ell})) \hf(\lambda_{\ell}, \lambda)} + \sup_{\lambda} \abs*{\sum_{\ell=1}^{L} u(\lambda_{\ell}) \hf(\lambda_{\ell}, \lambda)  - u(\lambda)} \\
        & \leq \sqrt{\frac{|\hbu - \bu|_{2}^{2}}{L} \times \sup_{\lambda} L \sum_{\ell=1}^{L} \hf(\lambda_{\ell}, \lambda)^{2}} + \sup_{\lambda} \abs*{\sum_{\ell=1}^{L} u(\lambda_{\ell}) \hf(\lambda_{\ell}, \lambda)  - u(\lambda)}.
      \end{aligned}
    \end{equation*}
    Since $|\hbu - \bu|_{2}^{2} / L \xrightarrow{a.s.} 0$ by Theorem \ref{thm:stable-errors}, to complete the the proof we need to demonstrate uniform convergence of $L \sum_{\ell=1}^{L} \hf(\lambda_{\ell}, \lambda)^{2}$ to a constant, and of $\sum_{\ell=1}^{L} u(\lambda_{\ell}) \hf(\lambda_{\ell}, \lambda) - u(\lambda)$ to 0. Noting that both sums are over independent terms, this is accomplished by application of a uniform law of large numbers within Lemma \ref{lem:ullnapp}  [Supplement]. The conditions for applying these ULLNs are fulfilled by way of (SupCond) as well as (Sup8Int) and (Lip8Int). 
    As for convergence in $L^{1}$, by Cauchy-Schwarz,
    \begin{equation*}
      \begin{aligned}
        \E \bracks*{\sup_{\lambda} |\hu(\lambda) - u(\lambda)|}
        & \leq \sqrt{\E \bracks*{\frac{|\hbu - \bu|_{2}^{2}}{L}} \E \bracks*{\sup_{\lambda} L \sum_{\ell=1}^{L} \hf(\lambda_{\ell}, \lambda)^{2}}} \\
        & \quad + \E \bracks*{\sup_{\lambda} \abs*{\sum_{\ell=1}^{L} u(\lambda_{\ell}) \hf(\lambda_{\ell}, \lambda) - u(\lambda)}},
      \end{aligned} 
    \end{equation*}
    where once more $\E \bracks{L^{-1} |\hbu - \bu|_{2}^{2}} \to 0$ by Theorem \ref{thm:stable-errors}, and $\E \bracks{\sup_{\lambda} L \sum_{\ell=1}^{L} \hf(\lambda_{\ell}, \lambda)^{2}} \to \mathrm{const}$ as well as $\E \bracks{\sup_{\lambda} \abs{\sum_{\ell=1}^{L} u(\lambda_{\ell}) \hf(\lambda_{\ell}, \lambda) - u(\lambda)}} \to 0$ by Lemma \ref{lem:ullnapp} [Supplement]. This verifies uniform convergence in $L^{1}$.
  \end{proof}
\end{theorem}

\section{Discussion}
\label{sec:disc}

We have developed a method for functional estimation of the marginal likelihood that takes as input samples generated with different values of the hyperparameters on a simulation grid. We have established the consistency of the functional estimator as the sampling effort increases, either by increasing the intensity on a given grid, or by increasing the resolution of the grid. Our analysis is based on a number of semi-non-asymptotic results that could also be used to provide rates of convergence. In the dense-grid case, due to our proof strategy the implied convergence rates are slow when the spectral gap of the Gibbs sampler is small, even though the actual rates can be fast in such regimes. Therefore, a full analysis of convergence rates requires further refinements. In terms of practical implementation, we recommend spreading the simulation effort on finer grids when this is possible, which has been illustrated in simulations. Additionally, we have sketched out an iterative strategy to design the simulation grid optimally. Its practical use relies on reasonable approximations to certain moments, which we leave to future work. Moreover, on finer grids, fewer samples are generated per grid point, which raises some issues when MCMC is used for sampling, and may require a certain amount of interaction between the sampling at different sites. In this work, we have assumed that the sampler used for each hyperparameter value is an independent black box. In such a setting, we can at least use ``warm starts'', where samples obtained after burn-in at one site can be used as starting values for a nearby one. Another common approach to accelerating convergence is to use replica exchange techniques as in \cite{emus_for_cosmo}.

\section*{Acknowledgements}

Timothee Stumpf-Fetizon has received funding from the European Research Council (ERC, PrSc-HDBayLe, Grant agreement No. 101076564).

\bibliographystyle{chicago}
\bibliography{biblio}

\appendix

\section{Appendix on toy examples}
\label{app:toy}

\begin{example}
  \label{ex:toy-ms}
  We consider the toy model  $y | \theta \sim 0.5 \Nor(\theta, q^{-1}) + 0.5 \Nor(-\theta, q^{-1})$ and $\theta | \lambda \sim \Nor(\lambda, \tau^{-1})$, with a flat prior $p(\lambda) \propto 1$. We take  $y = 1$; then for sufficiently large $q$ ($q = 64$ in this example), $u(\cdot)$ is bimodal with modes around $\pm 1$. As $\tau$ increases the modes become more separated and the posterior dependence between $\theta$ and $\lambda$ increases, hence the spectral gap in $\bF$ decreases.
\end{example}

We apply EMUS and Griddy Gibbs for the model in Example \ref{ex:toy-ms} and we set $\Ls=16$ equidistant points on the interval $[-2,2]$. Figure \ref{fig:toy-ms} shows
\begin{equation}
  \label{eq:L1}
  L^{-1} \E\bracks*{\sum_{\ell=1}^{L} |\hbu_{\ell} - \bu_{\ell}|},
\end{equation}
where the expectation is estimated over 128 runs of EMUS and the Gibbs sampler, with $N_{\ell} = 16$, for all $\ll$,  for EMUS and $N_{1} L$ iterations for the Gibbs sampler. In this example we can compute $\bu$ exactly. As $\tau$ increases from 1 to 10, the separation between the two modes in the marginal likelihood increases to the point that the Griddy Gibbs sampler cannot mix at all between the two modes. In that regime the EMUS errors are small, although for $\tau > 100$ the EMUS estimates also miss one of the two modes. As $\tau$ increases further, the dependence between $\lambda$ and $\theta$ gets stronger and stronger and this eventually affects EMUS errors, albeit to a lesser extent than those of the Griddy Gibbs sampler. 
 
\begin{figure}[htbp]
  \centering
  \includegraphics[scale=.66]{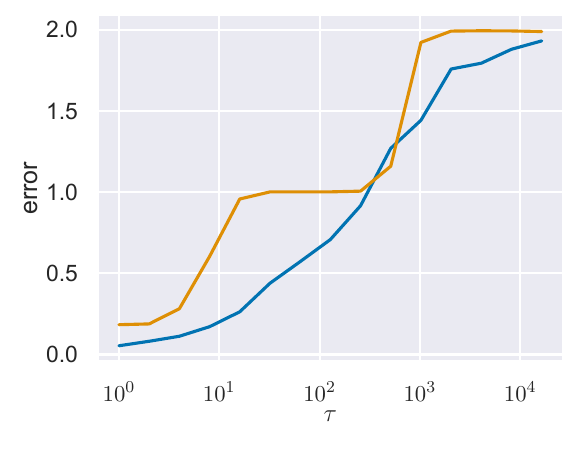}
  \caption{Complexity parameter $\tau$ versus approximation error as defined in \eqref{eq:L1} with EMUS (blue) and Griddy Gibbs (orange), for the model in Example \label{fig:toy-ms}.}
\end{figure}

\begin{example}
  Consider a space for $\theta$ that only takes 5 distinct values, $\theta_{1},\ldots,\theta_{5}$, and $L=2$ with
  \begin{gather*}
    \text{$\psi_{\lambda_{1}}(\theta_{i}) > 0$ for $i=1,2,3$ and $\psi_{\lambda_{1}}(\theta_{i}) = 0$ for $i=4,5$}, \\
    \text{$\psi_{\lambda_{2}}(\theta_{i}) = 0$ for $i=1,2$ and $\psi_{\lambda_{2}}(\theta_{i}) > 0$ for $i=3,4,5$}.
  \end{gather*}
  Let the marginal likelihood in this discrete setting be defined by convention as
  \begin{equation*}
    z(\lambda_{\ell}) = c \sum_{i=1}^5 \psi_{\lambda_{\ell}}(\theta_{i}),\quad \ell=1,2
  \end{equation*}
  where $c>0$ is an arbitrary normalization constant, and let $p(\lambda) \propto 1$.
  \label{ex:toy-is}
\end{example}
Within the setting of Example \ref{ex:toy-is}, we can define $\bF$ via \eqref{eq:kernel-discr}, but as an expectation with respect to a discrete distribution. On the other hand,  in the context of Example \ref{ex:toy-is}, the adaptation of the importance sampling identity \eqref{eq:is-id-grid} as an expectation with respect to a discrete distribution is not well defined since the probability mass functions do not have the same support; analogously, an estimator derived from  \eqref{eq:is-id-grid} would not have finite mean. A direct calculation yields that $\bu_{i}/\bu_{j} = \bF_{j,i}/\bF_{i,j}$ for $i,j =1,2$. Replacing $\bF$ by $\hbF$ we get that  $\hbu_{i}/\hbu_{j} = \hbF_{j,i}/\hbF_{i,j}$ for $i,j =1,2$. While it is possible for the denominator to vanish (if no samples from $\pi_{\lambda_{i}}$ equal $\theta_3$), this occurs with exponentially low probability as the number of samples increases.

\section{Appendix on Computation of EMUS}
\label{app:compemus}

Since the EMUS estimator $\hbu$ is defined in \eqref{eq:emus} as the leading eigenvector of $\hbF\T$, it could be obtained in principle by way of an eigen-decomposition. This would be sub-optimal in terms of numerical precision as it tries to compute the leading eigenvalue of $\hbF\T$, even though it is known to be 1. Instead, we follow \cite{golub1986using} in obtaining the vector of stationary probabilities by way of a QR-decomposition of the stochastic matrix. If $\bA = \bQ\bR$ is a QR-decomposition of $\bA = \bI -\hbF$, then $\hbu$ corresponds to the last column of $\bQ$, up to a multiplicative constant. Moreover, we follow \cite{DinnerThiede2020emus2} recommendation of further polishing this estimate by putting it through a few iterations of the power method.

\section{Appendix on Optimal Design}
\label{app:design}

In what follows, the boldface symbol $\arr{a}^{e}$ refers to the vector of evaluations of the function $a: \Lambda \to \Re$ on $\Le$. We solve the problem of minimizing $\lim_{N\to \infty} \E[N|\hbu^{e} - \bu^{e}|_{2}^{2}]$ as a function of the sampling weights $\bw_{m}^{e}$, subject to $\sum_{m = 1}^{M} \bw_{m}^{e} = 1$. The optimal solution is
\begin{equation}
  \label{eq:wopt}
  \bw_{m}^{e} \propto \bu_{m}^{e} \sqrt{\tr\braces*{(\bI^{e} - {\bF^{e}}\T)^{\#} \arr{\Xi}_{m}^{e} (\bI^{e} - \bF^{e})^{\#}}}, \quad m = 1, \dots, M,
\end{equation}
where $\bF^{e}$ and $\bu^{e}$ are defined as in \eqref{eq:kernel-discr} and \eqref{eq:bfu}, but with $\lambda_{i},\lambda_{j} \in \Le$, and 
\begin{equation*}
    \braces{\arr{\Xi}_{i}^{e}}_{j,k} + \bF_{i,j}^{e} \bF_{i,k}^{e} = \int \frac{\psi_{\lambda_{j}}(\theta) p(\lambda_{j}) \psi_{\lambda_k}(\theta) p(\lambda_k)}{\parens*{\sum_{\lambda_{m} \in \Lambda^{e}} \psi_{\lambda_{m}}(\theta) p(\lambda_{m})}^{2}} \pi_{\lambda_{i}}(\theta) \dd \theta, \quad \lambda_{i}, \lambda_{j}, \lambda_{k} \in \Lambda^{e}.
\end{equation*}
While these integrals could be estimated by sampling from $\braces{\pi_{\lambda_{m}}: \lambda_{m} \in \Le}$ exhaustively, Proposition \ref{prop:c-e} allows the integrals to be estimated with samples from $\Ls \subset \Le$ alone. For example, in the notation of Proposition \ref{prop:c-e}, setting $\phi(\theta) = \psi_{\lambda_{j}}(\theta) p(\lambda_{j}) / (\sum_{\lambda_{m} \in \Lambda^{e}} \psi_{\lambda_{m}}(\theta) p(\lambda_{m}))$, and noting that $\pi_{\lambda_{i}}(\phi) = \bF^{e}_{i,j}$, the proposition yields
\begin{equation*}
  \bF^{e}_{i,j} = \sum_{\lambda_{k} \in \Lambda^{s}} \frac{\bu_{k}}{\bu_{i}} \int \frac{\psi_{\lambda_{j}}(\theta) p(\lambda_{j})}{\sum_{\lambda_{m} \in \Lambda^{e}} \psi_{\lambda_{m}}(\theta) p(\lambda_{m})}  \frac{\psi_{\lambda_{i}}(\theta) p(\lambda_{i})}{\sum_{\lambda_{\ell} \in \Lambda^{s}} \psi_{\lambda_{\ell}}(\theta) p(\lambda_{\ell})} \pi_{\lambda_{k}}(\theta) \dd{\theta}, \quad \lambda_{i}, \lambda_{j} \in \Lambda^{e}.
\end{equation*}
A similar expression holds for the cross-moments $\braces{\arr{\Xi}_{i}^{e}}_{j,k}$.
Replacing all quantities on the right hand sides of these expressions with estimates and plugging them into \eqref{eq:wopt} provides an optimal weight estimate $\hbw^{e} = \braces{\hbw_{m}^{e}: \lambda_{m} \in \Le}$.

Having computed $\hbw^{e}$ from all previously generated samples, we obtain our incremental sampling weights according to the adjustment
\begin{equation}
  \label{eq:select}
  \bar{\bw}_{m}^{e} \propto \braces*{\parens*{B + \sum_{\lambda_{\ell} \in \Ls} N_{\ell}} \hbw_{m}^{e} - N_{m}} \lor 0, \quad \lambda_{m} \in \Lambda^{e},
\end{equation}
where $B$ is the number of new samples to be generated, and $N_{\ell}$ is the number of previously generated samples at $\pi_{\lambda_{\ell}}$. The estimator can be stabilized by replacing $\hbw^{e}$ in \eqref{eq:select} by weights proportional to $\sqrt{\hbw^{e}}$. Since this solution is a continuous relaxation of the discrete effort allocation problem, we discretize by sampling locations from $\Le$ according to the probabilities $\bar{\bw}_{m}^{e}$ using the pivotal sampling algorithm~\cite{deville1998unequal}.

The proposal above is only preliminary, as is our experimental test. In our current implementation the method requires on the order of $M^{3} \sum_{\lambda_{\ell} \in \Ls} N_{\ell}$ operations. Approximations may be necessary to reduce this cost when $M$ is large. We leave further development and testing of this iterative sampling protocol to future work.

\section{Appendix for Asymptotic Variance Bounds}
\label{app:avar-emus}

\subsection{Proof of Theorem \ref{thm:avar-emus}}

\begin{proof}
As in Appendix \ref{app:design}, define
\begin{equation*}
    \braces{\arr{\Xi}_{i}}_{j,k} = \Cov\parens*{\frac{\psi_{\lambda_{j}}(\theta_{i}^{(1)}) p(\lambda_{j})}{\sum_{\ell=1}^{L} \psi_{\lambda_{\ell}}(\theta_{i}^{(1)}) p(\lambda_{\ell})}, \frac{\psi_{\lambda_{k}}(\theta_{i}^{(1)}) p(\lambda_{k})}{\sum_{\ell=1}^{L} \psi_{\lambda_{\ell}}(\theta_{i}^{(1)}) p(\lambda_{\ell})}}, \quad \lambda_{i}, \lambda_{j}, \lambda_{k} \in \Le,
\end{equation*}
which is finite since $\psi_{\lambda_{j}}(\theta_{i}^{(1)}) p(\lambda_{j}) / (\sum_{\ell=1}^{L} \psi_{\lambda_{\ell}}(\theta_{i}^{(1)}) p(\lambda_{\ell})) \leq 1$. Thus, $\sqrt{N_{i}}( \hbF_{i,j} - \bF_{i,j})$ converges in distribution to a Gaussian variable with mean 0 and variance $\bR_{i,j} = \braces{\arr{\Xi}_{i}}_{j,j}$ as $N \to \infty$. As in \cite[Theorem 3.3]{DinnerThiede2020emus2}, the Delta Theorem implies that $\sqrt{N} (\hbu_{\ell} - \bu_{\ell})$ is asymtotically Gaussian with variance
\begin{equation*}
  \begin{aligned}
    \bsig_{\ell}^2
    & = \sum_{i = 1}^{L} w_{i}^{-1} \sum_{j \neq i,\, k \neq i} \pder{\bu_{\ell}}{\bF_{i,j}} \pder{\bu_{\ell}}{\bF_{i,k}} \braces{\arr{\Xi}_{i}}_{j,k} \\
    & = \sum_{i = 1}^{L} w_{i}^{-1} \sum_{j \neq i,\, k \neq i} \pder{\bu_{\ell}}{\bF_{i,j}} \pder{\bu_{\ell}}{\bF_{i,k}} \sqrt{\bR_{i,j} \bR_{i, k}} \frac{\braces{\arr{\Xi}_{i}}_{j,k}}{\sqrt{\braces{\arr{\Xi}_{i}}_{j,j}, \braces{\arr{\Xi}_{i}}_{k,k}}},
  \end{aligned}
\end{equation*}
where the derivatives are to be understood as the sensitivity of the solution of $\bu = \bF\T \bu$ to the entries in $\bF$. A precise definition is given in \cite[Definition D.2]{DinnerThiede2020emus2}. We observe that the RHS is a quadratic form with respect to the positive semi-definite correlation matrix with entries $\braces{\arr{\Xi}_{i}}_{j,k} / \sqrt{\braces{\arr{\Xi}_{i}}_{j,j}, \braces{\arr{\Xi}_{i}}_{k,k}}$. For a positive semi-definite matrix $\bA$, we have the quadratic form bound $\ba\T \bA\ba \leq \tr(\bA) \abs{\ba}_{2}^{2}$, which in this instance yields
\begin{equation*}
    \sum_{j \neq i,\, k \neq i} \pder{\bu_{\ell}}{\bF_{i,j}} \pder{\bu_{\ell}}{\bF_{i,k}} \braces{\arr{\Xi}_{i}}_{j,k} \leq L \sum_{j \neq i} \parens*{\pder{\bu_{\ell}}{\bF_{i,j}}}^{2} \bR_{i,j}.
\end{equation*}
We then apply \cite[Theorem 2]{thiede2015sharp} to bound the derivatives by
\begin{equation*}
 \abs*{\pder{\log \bu_{\ell}}{\bF_{i,j}}} \leq \frac{1}{\bQ_{i,j}},
\end{equation*}
with $\bQ_{i,j}$ previously defined in Theorem \ref{thm:avar-emus}. Hence,
\begin{equation*}
  \begin{aligned}
    L \sum_{j \neq i} \parens*{\pder{\bu_{\ell}}{\bF_{i,j}}}^{2} \bR_{i,j}
    = \bu_{\ell}^{2} L \sum_{j \neq i} \parens*{\pder{\log \bu_{\ell}}{\bF_{i,j}}}^{2} \bR_{i,j}
    = \bu_{\ell}^{2} L \sum_{j \neq i} \frac{\bR_{i,j}}{\bQ_{i,j}^{2}},
  \end{aligned}
\end{equation*}
and in conjunction,
\begin{equation*}
  \frac{\bsig_{\ell}^2}{\bu_{\ell}^2} \leq L \sum_{i=1}^{L} w_{i}^{-1} \sum_{j \neq i} \frac{\bR_{i,j}}{\bQ_{i,j}^{2}}.
\end{equation*}
\end{proof}

\subsection{Proof of Proposition \ref{prop:var-femus}}

\begin{proof}
We observe that
\begin{equation*}
    \abs*{\pder{u(\lambda)}{\bF_{i,j}}} \leq \sup_{\ell} \abs*{\pder{\log \bu_{\ell}}{\bF_{i,j}}} \sum_{\ell=1}^{L} \bu_{\ell} f(\lambda_{\ell}, \lambda) \leq u(\lambda) \sup_{\ell} \abs*{\pder{\log \bu_{\ell}}{\bF_{i,j}}},
\end{equation*}
and $\partial_{f(\lambda_{\ell}, \lambda)} u(\lambda) = \bu_{\ell}$. Following \cite[Theorem 3.5]{DinnerThiede2020emus2}, the delta method yields a CLT with variance bound
\begin{equation*}
  \begin{aligned}
    \sigma^{2}(\lambda)
    & \leq 2 \sum_{i = 1}^{L} w_{i}^{-1} \parens*{\sum_{j \neq i,\, k \neq i} \pder{u(\lambda)}{\bF_{i,j}} \pder{u(\lambda)}{\bF_{i,k}} \braces{\arr{\Xi}_{i}}_{j,k} + \parens*{\pder{u(\lambda)}{f(\lambda_{i}, \lambda)}}^{2} r_{i}(\lambda)} \\
    & \leq 2 \sum_{i = 1}^{L} w_{i}^{-1} \parens*{L \sum_{j \neq i} \parens*{\pder{u(\lambda)}{\bF_{i,j}}}^{2} \bR_{i,j} + \bu_{i}^{2} r_{i}(\lambda)} \\
    & \leq 2 \sum_{i = 1}^{L} w_{i}^{-1} \parens*{L \sum_{j \neq i} u^{2}(\lambda) \parens*{\pder{\log \bu_{\ell}}{\bF_{i,j}}}^{2} \bR_{i,j} + \bu_{i}^{2} r_{i}(\lambda)} \\
    & \leq 2 \sum_{i = 1}^{L} w_{i}^{-1} \parens*{L u^{2}(\lambda) \sum_{j \neq i} \frac{\bR_{i,j}}{\bQ_{i,j}^{2}} + \bu_{i}^{2} r_{i}(\lambda)}.
  \end{aligned}
\end{equation*}
\end{proof}

\section{Appendix for Fixed-grid Consistency}
\label{app:fixed-g}

\subsection{Establishing condition (SupInt) for Gaussian process priors}

\begin{example}[Supremum-Integrability for Gaussian priors]
\label{ex:gaussint}
  Suppose that $\theta \in \mathbf{R}^{d}$ and that $\psi(\theta, \lambda)$ consists of a bounded likelihood function and a Gaussian prior, i.e.
  \begin{equation*}
    \psi(\theta, \lambda) = p(y | \theta) \Nor(\theta; 0, C_{\lambda}), \qquad \sup_{\theta} p(y | \theta) < \infty,
  \end{equation*}
  and that the leading and trailing eigenvalues of $C_{\lambda}$, denoted $\varepsilon_{1}(\lambda)$ and $\varepsilon_{d}(\lambda)$ respectively, satisfy
  \begin{equation*}
    0 < \inf_{\lambda} \varepsilon_{d}(\lambda) \leq \sup_{\lambda} \varepsilon_{1}(\lambda) < \infty.
  \end{equation*}
  Furthermore, let $p(\lambda)$ be bounded on $\Lambda$. Then, (SupInt) applies. Finally, let $\zeta_{d}(\lambda, \lambda')$ be the smallest eigenvalue of $2 C_{\lambda}^{-1} - C_{\lambda'}$. If in addition $\inf_{\lambda \in \Lambda} \sup_{\lambda' \in \Lambda^{s}}  \zeta_{d}(\lambda, \lambda') > 0$, then (Sup2Int) applies.
  \begin{proof}
  We begin by observing that $\theta\T C_{\lambda}^{-1} \theta \geq |\theta|_{2}^{2} / \varepsilon_{1}(\lambda)$ and $\det C_{\lambda} \geq \varepsilon_{d}(\lambda)^{d}$, respectively implying
  \begin{gather*}
    \exp\braces*{-\frac{\theta\T C_{\lambda}^{-1} \theta}{2}} \leq \exp\braces*{-\frac{\abs{\theta}_{2}^{2}}{2 \varepsilon_{1}(\lambda)}}, \qquad \frac{1}{\det C_{\lambda}} \leq \frac{1}{\varepsilon_{d}(\lambda)^{d}}.
  \end{gather*}
  Hence, the prior density is bounded above by
  \begin{equation*}
    \begin{aligned}
      \Nor(\theta; 0, C_{\lambda})
      & = \frac{1}{\sqrt{(2 \pi)^{d} \det C_{\lambda}}} \exp\braces*{-\frac{\theta\T C_{\lambda}^{-1} \theta}{2}} \\
      & \leq \frac{1}{\sqrt{(2 \pi \varepsilon_{d}(\lambda))^{d}}} \exp\braces*{-\frac{\abs{\theta}_{2}^{2}}{2 \varepsilon_{1}(\lambda)}},
    \end{aligned}
  \end{equation*}
  whereby we obtain the required bound
  \begin{equation*}
    \begin{aligned}
      \int \sup_{\lambda} \psi(\theta, \lambda) p(\lambda) \dd{\theta}
      & \leq \int \sup_{\lambda} \parens*{\frac{p(y | \theta) p(\lambda)}{\sqrt{(2 \pi \varepsilon_{d}(\lambda))^{d}}} \exp\braces*{-\frac{\abs{\theta}_{2}^{2}}{2 \varepsilon_{1}(\lambda)}}} \dd{\theta} \\
      & \leq \frac{\sup_{\theta} p(y | \theta) \sup_{\lambda} p(\lambda)}{\sqrt{(2 \pi \inf_{\lambda} \varepsilon_{d}(\lambda))^{d}}} \int \exp\braces*{-\frac{\abs{\theta}_{2}^{2}}{2 \sup_{\lambda} \varepsilon_{1}(\lambda)}} \dd{\theta} \\
      & \leq \sup_{\theta} p(y | \theta) \sup_{\lambda} p(\lambda) \sqrt{\parens*{\frac{\sup_{\lambda} \varepsilon_{1}(\lambda)}{\inf_{\lambda} \varepsilon_{d}(\lambda)}}^{d}} \\
      & < \infty.
    \end{aligned}
  \end{equation*}
  Proceeding with (Sup2Int),
  \begin{equation*}
    \begin{aligned}
      & \int \frac{\sup_{\lambda} \psi(\theta, \lambda)^{2}}{\sum_{i=1}^{L} \psi(\theta, \lambda_{i})} \dd{\theta} \\
      & \leq \sup_{\theta} \pi(y | \theta) \int \frac{\sup_{\lambda \in \Lambda} \Nor(\theta; 0, C_{\lambda})^{2}}{\sup_{\lambda' \in \Lambda^{s}} \Nor(\theta; 0, C_{\lambda'})} \dd{\theta} \\
      & \leq \sup_{\theta} \pi(y | \theta) \parens*{\frac{\sqrt{2 \pi \inf_{\lambda' \in \Lambda^{s}} \varepsilon_{1}(\lambda')}}{2 \pi \inf_{\lambda \in \Lambda} \varepsilon_{d}(\lambda)}}^{d} \int \sup_{\lambda \in \Lambda} \inf_{\lambda' \in \Lambda^{s}} e^{-\theta\T (2 C_{\lambda}^{-1} - C_{\lambda'}^{-1}) \theta / 2} \dd{\theta} \\
      & \leq \sup_{\theta} \pi(y | \theta) \parens*{\frac{\sqrt{2 \pi \inf_{\lambda' \in \Lambda^{s}} \varepsilon_{1}(\lambda')}}{2 \pi \inf_{\lambda \in \Lambda} \varepsilon_{d}(\lambda)}}^{d} \int \exp\braces*{-\frac{|\theta|_{2}^{2}}{2} \inf_{\lambda \in \Lambda} \sup_{\lambda' \in \Lambda^{s}} \zeta_{d}(\lambda, \lambda')} \dd{\theta},
    \end{aligned}
  \end{equation*}
  where $\zeta_{d}(\lambda, \lambda')$ is the smallest eigenvalue of $2 C_{\lambda}^{-1} - C_{\lambda'}^{-1}$, and the integral is finite if $\inf_{\lambda \in \Lambda} \sup_{\lambda' \in \Lambda^{s}}  \zeta_{d}(\lambda, \lambda') > 0$. Thus, (Sup2Int) applies when for every $\lambda \in \Lambda$, there is a $\lambda' \in \Lambda^{s}$ such that $2 C_{\lambda}^{-1} - C_{\lambda'}^{-1}$ is positive definite.
  \end{proof}
\end{example}

\subsection{Proof of Lemma \ref{lem:kernel}}

We make use of a strong law of large numbers on a Banach space, as in the following Theorem, which can be found, e.g., in \cite{giesy2006strong}, see also \cite{mcmle-om} for its use in contexts similar to those considered here.

\begin{theorem}
\label{thm:ulln}
  Consider the separable Banach space $(X, \norm{\cdot})$, with elements $x \in X$ and norm $\norm{\cdot}$. Suppose that
  \begin{equation*}
    \E [\norm{x}] < \infty, \qquad \E[x] = 0.
  \end{equation*}
  Then, if $x_{1}, \ldots x_{N}$ are independent and distributed as $x$,
  \begin{equation*}
    \lim_{N \to \infty} \norm*{N^{-1} \sum_{n=1}^{N} x_{n}} = 0, \; a.s.
  \end{equation*}
\end{theorem}

The application of Theorem \ref{thm:ulln} to proving Lemma \ref{lem:kernel} follows below.

\begin{proof}
  We demonstrate almost sure convergence by appealing to Theorem \ref{thm:ulln}, with random elements
  \begin{equation*}
    \hf(\lambda_{i}, \lambda) - f(\lambda_{i}, \lambda) = N^{-1} \sum_{n=1}^{N} \frac{\psi_{\lambda}(\theta_{i}^{(n)}) p(\lambda)}{\sum_{\ell} \psi_{\lambda_{\ell}}(\theta_{i}^{(n)}) p(\lambda_{\ell})} - f(\lambda_{i}, \lambda)
  \end{equation*}
  on the Banach space $(\mathcal{C}, \abs{\cdot}_{\infty})$, where $\mathcal{C}$ is the set of continuous, real-valued functions on $\Lambda$, and $|h|_{\infty} = \sup_{\lambda \in \Lambda} |h(\lambda)|$ for $h \in \mathcal{C}$. This requires
  \begin{equation*}
    \E [\sup_{\lambda \in \Lambda} |\hf(\lambda_{i}, \lambda) - f(\lambda_{i}, \lambda)|]
    \leq \sup_{\lambda \in \Lambda} f(\lambda_{i}, \lambda) + \E \bracks*{\sup_{\lambda \in \Lambda} \frac{\psi_{\lambda}(\theta_{i}^{(1)}) p(\lambda)}{\sum_{\ell} \psi_{\lambda_{\ell}}(\theta_{i}^{(1)}) p(\lambda_{\ell})}}
  \end{equation*}
  to be finite. However, the expectation on the right-hand side admits the bound
  \begin{equation*}
    \begin{aligned}
      \E \bracks*{\sup_{\lambda \in \Lambda} \frac{\psi_{\lambda}(\theta_{i}^{(1)}) p(\lambda)}{\sum_{\ell} \psi_{\lambda_{\ell}}(\theta_{i}^{(1)}) p(\lambda_{\ell})}}
      & = \frac{1}{z(\lambda_{i})p(\lambda_{i})} \int \sup_{\lambda} \psi_{\lambda}(\theta) p(\lambda) \frac{\psi_{\lambda_{i}}(\theta) p(\lambda_{i})}{\sum_{\ell} \psi_{\lambda_{\ell}}(\theta) p(\lambda_{\ell})} \dd{\theta} \\
      & \leq \frac{1}{u(\lambda_{i})} \int \sup_{\lambda} \psi_{\lambda}(\theta) p(\lambda) \dd{\theta},
    \end{aligned}
  \end{equation*}
  which is finite by (SupInt). Thus, Theorem \ref{thm:ulln} applies and $\sup_{\lambda \in \Lambda} |\hf(\lambda_{i}, \lambda) - f(\lambda_{i}, \lambda)| \xrightarrow{a.s.} 0$.

  We may extend this to uniform convergence in $\mathcal{L}^{2}$ under (Sup2Int). By independence,
  \begin{equation*}
    \begin{aligned}
      \Var \bracks*{\sup_{\lambda \in \Lambda} |\hf(\lambda_{i}, \lambda) - f(\lambda_{i}, \lambda)|}
      & \leq L^{-1} \Var \bracks*{\sup_{\lambda \in \Lambda} \abs*{\frac{\psi_{\lambda}(\theta_{i}^{(1)}) p(\lambda)}{\sum_{\ell} \psi_{\lambda_{\ell}}(\theta_{i}^{(1)}) p(\lambda_{\ell})}}} \\
      & \leq L^{-1} \E \bracks*{\sup_{\lambda \in \Lambda} \parens*{\frac{\psi_{\lambda}(\theta_{i}^{(1)}) p(\lambda)}{\sum_{\ell} \psi_{\lambda_{\ell}}(\theta_{i}^{(1)}) p(\lambda_{\ell})}}^{2}}.
    \end{aligned}
  \end{equation*}
  Moreover, the expectation on the right-hand side admits the bound
  \begin{equation*}
    \E \bracks*{\sup_{\lambda \in \Lambda} \parens*{\frac{\psi_{\lambda}(\theta_{i}^{(1)}) p(\lambda)}{\sum_{\ell} \psi_{\lambda_{\ell}}(\theta_{i}^{(1)}) p(\lambda_{\ell})}}^{2}}
    \leq \frac{1}{u(\lambda_{i})} \int \frac{\sup_{\lambda \in \Lambda} \psi_{\lambda}(\theta)^{2} p(\lambda)^{2}}{\sum_{\ell} \psi_{\lambda_{\ell}}(\theta) p(\lambda_{\ell})} \dd \theta,
  \end{equation*}
  which is finite by (Sup2Int). Thus, $\Var \bracks*{\sup_{\lambda \in \Lambda} |\hf(\lambda_{i}, \lambda) - f(\lambda_{i}, \lambda)|} \to 0$. Having already established that $\E \bracks*{\sup_{\lambda \in \Lambda} |\hf(\lambda_{i}, \lambda) - f(\lambda_{i}, \lambda)|} < \infty$, it follows that $\E \bracks*{\sup_{\lambda \in \Lambda} |\hf(\lambda_{i}, \lambda) - f(\lambda_{i}, \lambda)|^{2}} < \infty$, and $\braces*{\sup_{\lambda \in \Lambda} |\hf(\lambda_{i}, \lambda) - f(\lambda_{i}, \lambda)|: N = 1, 2, \dots}$ is uniformly integrable. In conjunction with almost sure convergence, it follows that $\E \bracks*{\sup_{\lambda \in \Lambda} |\hf(\lambda_{i}, \lambda) - f(\lambda_{i}, \lambda)|} \to 0$. Since the variance vanishes, it finally follows that $\E \bracks*{\sup_{\lambda \in \Lambda} |\hf(\lambda_{i}, \lambda) - f(\lambda_{i}, \lambda)|^{2}} \to 0$.
\end{proof}

\section{Appendix for Dense-grid Consistency}
\label{app:dense-g}

\subsection{Lemmas for Theorem \ref{thm:stable-errors}}

\begin{lemma}[Convergence of $|\hbF (\hbF - \bF)|_{\mathrm{F}}$]
\label{lem:randopnorm}
  Suppose that (Cond) holds. Then,
  \begin{equation*}
    \lim_{L \to \infty} |\hbF (\hbF - \bF)|_{\mathrm{F}} = 0 \; (a.s.), \quad
    \limsup_{L \to \infty} L^{2} \E |\hbF (\hbF - \bF)|_{\mathrm{F}}^{4} < \infty.
  \end{equation*}
  \begin{proof}
    For brevity, we write $\bA = \hbF - \bF$, where $\E[\bA] = 0$. Writing the Frobenius norm as the sum over the squared matrix elements,
    \begin{equation*}
      |\hbF \bA|_{\mathrm{F}}^{4} = \parens*{\sum_{i,j = 1}^{L} \braces{\hbF \bA}_{ij}^{2}}^{2} = \sum_{i,j,i',j' = 1}^{L} \braces{\hbF \bA}_{ij}^{2} \braces{\hbF \bA}_{i'j'}^{2},
    \end{equation*}
    where by Cauchy-Schwarz $\E{\braces{\hbF \bA}_{ij}^{2} \braces{\hbF \bA}_{i'j'}^{2}} = \sqrt{\E\braces{\hbF \bA}_{ij}^{4} \E\braces{\hbF \bA}_{i'j'}^{4}}$. Thus,
    \begin{equation*}
      \begin{aligned}
        \limsup_{L \to \infty} L^{2} \E |\hbF \bA|_{\mathrm{F}}^{4}
        & \leq \limsup_{L \to \infty} L^{-4} \sum_{i,j,i',j' = 1}^{L} \sqrt{L^{12} \E\braces{\hbF \bA}_{ij}^{4} \E\braces{\hbF \bA}_{i'j'}^{4}} \\
        & \leq \limsup_{L \to \infty} \sup_{i, j, i', j'} \sqrt{L^{12} \E\braces{\hbF \bA}_{ij}^{4} \E\braces{\hbF \bA}_{i'j'}^{4}} \\
        & < \infty,
      \end{aligned}
    \end{equation*}
    by (Cond) and Lemma \ref{lem:randop2}. Having controlled the 4-th moment, we may exploit this to show almost sure convergence. The claim corresponds to finding that
    \begin{equation*}
      \P\bracks{|\hbF \bA|_{\mathrm{F}} > \epsilon \ i.o.} = 0, \quad \forall \epsilon > 0.
    \end{equation*}
    By Borel-Cantelli, that corresponds to finding
    \begin{equation*}
      \sum_{L=1}^{\infty} \P\bracks{|\hbF \bA|_{\mathrm{F}} > \epsilon} < \infty.
    \end{equation*}
    We observe that by Markov's inequality,
    \begin{equation*}
      \P\bracks{|\hbF \bA|_{\mathrm{F}} > \epsilon}
      = \P\bracks{L^{2} |\hbF \bA|_{\mathrm{F}}^{4} > L^{2} \epsilon^{4}}
      \leq \frac{L^{2} \E |\hbF \bA|_{\mathrm{F}}^{4}}{L^{2} \epsilon^{4}},
    \end{equation*}
    where the numerator is finite for some $L^{*} < \infty$. Then,
    \begin{equation*}
      \sum_{L=1}^{\infty} \P\bracks{|\hbF \bA|_{\mathrm{F}} > \epsilon}
      \leq L^{*} + \sum_{L=L^{*}}^{\infty} \frac{\limsup_{L \to \infty} L^{2} \E |\hbF \bA|_{\mathrm{F}}^{4}}{L^{2} \epsilon^{4}},
    \end{equation*}
    is a convergent series, as required.
  \end{proof}
\end{lemma}

\begin{lemma}[Convergence of spectral gap]
  \label{lem:convgap}
  Suppose that (Irred), (SupInt), (LipInt), (Cond) hold. Define $v(\lambda) = u(\lambda) / |\bu|_{1}$ such that it is a probability mass function on $\Ls$, let $\gamma$ be the absolute spectral gap of $\bF$, and $\tilde{\gamma}$ the absolute spectral gap of the limiting Markov kernel $g(\lambda, \lambda')$. Then, 
  \begin{equation*}
    \lim_{L \to \infty} \gamma = \tilde{\gamma} > 0, \quad
    \limsup_{L \to \infty} |(I - \bF)^{\#}|_{2} \leq \frac{\sqrt{\sn{\pi} \sn{1/\pi}}}{\tilde{\gamma}}.
  \end{equation*}
  \begin{proof}
    We begin by noticing that $\tilde{\gamma} > 0$ by Lemma \ref{lem:gapexist} under (SupInt), (Irred). In order to avoid potential issues with periodicity, we will work with the squared operators $\bF^{2}$ and $G^{2}$, where $(G\phi)(\lambda) = \int g(\lambda, \lambda') \phi(\lambda') \dd{\lambda'}$, such that their spectrum is non-negative. Then, $\gamma = 1 - \sqrt{\varepsilon_{2}}$, where $\varepsilon_{2}$ is the second eigenvalue of $\bF^{2}$, and $\tilde{\varepsilon}_{2}$ is the second eigenvalue of $G^{2}$. For reversible stochastic matrices, we recall the variational expression for $\varepsilon_{2}$, given by
    \begin{equation*}
      \varepsilon_{2} = \sup_{\braces{\bph: \substack{\iprod{\bone, \bph}_{\bv} = 0 \\ \iprod{\bph, \bph}_{\bv} = 1}}} \iprod{\bF^{2} \bph, \bph}_{\bv},
    \end{equation*}
    where we observe that due to reversibility/self-adjointness of $\bF$ with respect to $\bv$, $\iprod{\bF^{2} \bph, \bph}_{\bv} = \iprod{\bF \bph, \bF \bph}_{\bv}$. We also note that the maximizing argument $\bw$ is an eigenvector corresponding to $\varepsilon_{2}$. Using the analogous expression for the limiting second eigenvalue,
    \begin{equation*}
      \varepsilon_{2} - \tilde{\varepsilon}_{2} = \sup_{\braces{\bph: \substack{\iprod{\bone, \bph}_{\bv} = 0 \\ \iprod{\bph, \bph}_{\bv} = 1}}} \iprod{\bF \bph, \bF \bph}_{\bv} - \sup_{\braces{\phi: \substack{\iprod{1, \phi}_{\pi} = 0 \\ \iprod{\phi, \phi}_{\pi} = 1}}} \iprod{G\phi, G\phi}_{\pi}.
    \end{equation*}
    The strategy consists of expressing both eigenvalues as maxima over the same class, and bounding the difference between the maxima. The first step is to express the constraint on $\bph$ as a constraint on its interpolant $\phi$, i.e.
    \begin{equation*}
      \sup_{\braces{\bph: \substack{\iprod{\bone, \bph}_{\bv} = 0 \\ \iprod{\bph, \bph}_{\bv} = 1}}} \iprod{\bF \bph, \bF \bph}_{\bv} = \sup_{\braces{\phi: \substack{\iprod{\phi, 1}_{\bv} = 0 \\ \iprod{\phi, \phi}_{\bv} = 1}}} \iprod{F \phi, F \phi}_{\bv},
    \end{equation*}
    where $F$ is the operator such that $(F \phi)(\lambda) = \sum_{\ell=1}^{L} f(\lambda, \lambda_{\ell})\phi(\lambda_{\ell})$. As a matter of fact, under (LipInt) and (Cond), Lemma \ref{lem:smoothev} shows that $\iprod{G\phi, G\phi}_{\pi}$ is maximized by a Lipschitz-continuous function $\tilde{\omega}$, and we can restrict solutions to  $\tilde{\Phi} = \braces{\phi: \iprod{\phi, 1}_{\pi} = 0, \iprod{\phi, \phi}_{\pi} = 1, \lip{\phi} \leq \lip{\tilde{\omega}}}$. Analogously, $\iprod{F \phi, F \phi}_{\bv}$ is maximized by a piecewise linear function $w$ with $\limsup_{L \to \infty} \lip{\omega} < \infty$ and $\sn{\omega} \leq \lip{\omega}$. Therefore, we can restrict ourselves to solutions in the class
    \begin{equation*}
      \Phi = \braces{\phi: \iprod{\phi, 1}_{\bv} = 0, \iprod{\phi, \phi}_{\bv} = 1, \lip{\phi} \leq \lip{\omega} \lor \lip{\tilde{\omega}} / \iprod{\tilde{\omega}, \tilde{\omega}}_{\bv}^{1/2}},
    \end{equation*}
    where we have slightly enlarged the class to include $(\tilde{\omega} - \iprod{\tilde{\omega}, 1}_{\bv}) / \iprod{\tilde{\omega}, \tilde{\omega}}_{\bv}^{1/2}$. We note that for both classes, $\sn{\phi} \leq \lip{\phi}$ since $\phi$ must cross 0. Having lifted the first optimization problem to function space, we express $\tilde{\varepsilon}_{2}$ as a maximum over $\Phi$, i.e.
    \begin{equation*}
      \tilde{\varepsilon}_{2} = \sup_{\phi \in \Phi} \frac{\iprod{G\phi, G\phi}_{\pi} - \iprod{\phi, 1}_{\pi}^{2}}{\iprod{\phi, \phi}_{\pi} - \iprod{\phi, 1}_{\pi}^{2}}.
    \end{equation*}
    Then, by the triangle inequality,
    \begin{equation*}
      \begin{aligned}
        |\varepsilon_{2} - \tilde{\varepsilon}_{2}|
        & = \abs*{\sup_{\phi \in \Phi} \iprod{F\phi, F\phi}_{\bv} - \sup_{\phi \in \Phi} \frac{\iprod{G\phi, G\phi}_{\pi} - \iprod{\phi, 1}_{\pi}^{2}}{\iprod{\phi, \phi}_{\pi} - \iprod{\phi, 1}_{\pi}^{2}}} \\
        & \leq \sup_{\phi \in \Phi} \abs*{\iprod{F\phi, F\phi}_{\bv} - \frac{\iprod{G\phi, G\phi}_{\pi} - \iprod{\phi, 1}_{\pi}^{2}}{\iprod{\phi, \phi}_{\pi} - \iprod{\phi, 1}_{\pi}^{2}}} \\
        & \leq \sup_{\phi \in \Phi} |\iprod{F\phi, F\phi}_{\bv} - \iprod{G\phi, G\phi}_{\pi}| \\
        & \quad + \sup_{\phi \in \Phi} \abs*{\iprod{G\phi, G\phi}_{\pi} - \frac{\iprod{G\phi, G\phi}_{\pi} - \iprod{\phi, 1}_{\pi}^{2}}{\iprod{\phi, \phi}_{\pi} - \iprod{\phi, 1}_{\pi}^{2}}}.
      \end{aligned}
    \end{equation*}
    Moreover, the second term admits the bound
    \begin{equation*}
      \begin{aligned}
        & \sup_{\phi \in \Phi} \abs*{\iprod{G\phi, G\phi}_{\pi} - \frac{\iprod{G\phi, G\phi}_{\pi} - \iprod{\phi, 1}_{\pi}^{2}}{\iprod{\phi, \phi}_{\pi} - \iprod{\phi, 1}_{\pi}^{2}}} \\
        & = \sup_{\phi \in \Phi} \abs*{\iprod{\phi, 1}_{\pi}^{2} + \parens*{\iprod{\phi, \phi}_{\pi} - \iprod{\phi, 1}_{\pi}^{2} - 1} \frac{\iprod{G\phi, G\phi}_{\pi} - \iprod{\phi, 1}_{\pi}^{2}}{\iprod{\phi, \phi}_{\pi} - \iprod{\phi, 1}_{\pi}^{2}}} \\
        & \leq \sup_{\phi \in \Phi} \iprod{\phi, 1}_{\pi}^{2} + \sup_{\phi \in \Phi} \abs*{\iprod{\phi, \phi}_{\pi} - \iprod{\phi, 1}_{\pi}^{2} - 1} \sup_{\phi \in \Phi} \abs*{\frac{\iprod{G\phi, G\phi}_{\pi} - \iprod{\phi, 1}_{\pi}^{2}}{\iprod{\phi, \phi}_{\pi} - \iprod{\phi, 1}_{\pi}^{2}}} \\
        & \leq \sup_{\phi \in \Phi} \iprod{\phi, 1}_{\pi}^{2} + \sup_{\phi \in \Phi} |\iprod{\phi, \phi}_{\pi} - \iprod{\phi, 1}_{\pi}^{2} - 1| \\
        & \leq 2 \sup_{\phi \in \Phi} |\iprod{\phi, 1}_{\bv} - \iprod{\phi, 1}_{\pi}|^{2} + \sup_{\phi \in \Phi} |\iprod{\phi, \phi}_{\bv} - \iprod{\phi, \phi}_{\pi}|,
      \end{aligned}
    \end{equation*}
    where in the second to last step, the absolute Rayleigh quotient is at most 1 since all eigenvalues are at most 1 in absolute value, and the last step follows from the triangle inequality, keeping in mind that $\iprod{\phi, 1}_{\bv} = 0$ and $\iprod{\phi, \phi}_{\bv} = 1$ for $\phi \in \Phi$. The problem, then, becomes one of uniformly controlling the convergence of various expectations with respect to $\bv$ to expectations with respect to $\pi$, such that the supremum is controlled as well. In particular, by (LipInt), (Cond) and Lemmas \ref{lem:cvgqf} and \ref{lem:smoothev}, 
    \begin{equation*}
      \begin{aligned}
        & \limsup_{L \to \infty} L \sup_{\phi \in \Phi} |\iprod{F\phi, F\phi}_{\bv} - \iprod{G\phi, G\phi}_{\pi}| \\
        & \leq c_{4} \limsup_{L \to \infty} \sup_{\phi \in \Phi} \lip{\phi}^{2} \\
        & \leq c_{4} \parens*{\frac{\lip{\tilde{\omega}}}{\sqrt{\lim_{L \to \infty} \iprod{\tilde{\omega}, \tilde{\omega}}_{\bv}}} \lor \limsup_{L \to \infty} \lip{\omega}}^{2} \\
        & \leq c_{4} \parens*{\lip{\tilde{\omega}} \lor \limsup_{L \to \infty} \lip{\omega}}^{2} \\
        & < \infty,
      \end{aligned}
    \end{equation*}
    with Lemma \ref{lem:cvgdens} yielding analogous bounds for the other terms. In conjunction,
    \begin{equation*}
      \limsup_{L \to \infty} L|\varepsilon_{2} - \tilde{\varepsilon}_{2}| < \infty.
    \end{equation*}
    Therefore, $\lim_{L \to \infty} \gamma = 1 - \sqrt{\lim_{L \to \infty} \varepsilon_{2}} = 1 - \sqrt{\tilde{\varepsilon}_{2}} = \tilde{\gamma}$. Applying now Lemma \ref{lem:gapbound}, we obtain the group inverse norm bound
    \begin{equation*}
      \begin{aligned}
        \limsup_{L \to \infty} |(I - \bF)^{\#}|_{2}
        & \leq \frac{\sqrt{\sn{\bv} \sn{1/\bv}}}{\lim_{L \to \infty} \gamma}
        & \leq \frac{\sqrt{\sn{\pi} \sn{1/\pi}}}{\tilde{\gamma}}.
      \end{aligned}
    \end{equation*}
  \end{proof}
\end{lemma}

\subsection{Lemmas for Theorem \ref{thm:dense-g}}

Lemma \ref{lem:ullnapp} leverages the following generic uniform law of large numbers for smooth functions. An even more general statement is found in \cite{andrews1992generic}.

\begin{theorem}[Uniform law of large numbers (See e.g. \cite{andrews1992generic}, Theorem 2)]
\label{thm:ullnlip}
  Let $\braces{S_{L}(\lambda): \ell = 1, 2, \dots}$ be a sequence of random functions on the compact set $\Lambda$, and let $T_{L} > \lip{S_{L}}$ (a.s.). $S_{L}$ is \emph{strongly stochastically equicontinuous} SSE if $\limsup_{L \to \infty} \E\bracks{L^{-1} T_{L}} < \infty$ and $|L^{-1} T_{L} - \E\bracks{L^{-1} T_{L}}| \xrightarrow{a.s.} 0$. If $S_{L}$ is SSE, $\limsup_{L \to \infty} \E{|L^{-1} S_{L}(\lambda)|} < \infty$, and $|L^{-1} S_{L}(\lambda) - \E\bracks{L^{-1} S_{L}(\lambda)}| \xrightarrow{a.s.} 0$ for any $\lambda \in \Lambda$, then
  \begin{equation*}
    \lim_{L \to \infty} \sup_{\lambda \in \Lambda} |L^{-1} S_{L}(\lambda) - \E\bracks{L^{-1} S_{L}(\lambda)}| = 0 \; (a.s.)
  \end{equation*}
\end{theorem}

On that basis, Lemma \ref{lem:ullnapp} verifies sufficient conditions to apply Theorem \ref{thm:ullnlip}, via Lemma \ref{lem:ullnstat}.

\begin{lemma}[Application of uniform LLNs]
  \label{lem:ullnapp}
  Suppose that (SupCond), (Sup8Int), (Lip8Int) hold. Then,
  \begin{gather*}
    \lim_{L \to \infty} \sup_{\lambda} L \sum_{i=1}^{L} \hf(\lambda_{i}, \lambda)^{2} < \infty \; (a.s.), \quad
    \lim_{L \to \infty} \sup_{\lambda} \abs*{\sum_{i=1}^{L} \bu_{i} \hf(\lambda_{i}, \lambda) - u(\lambda)} = 0 \; (a.s.), \\ 
    \lim_{L \to \infty} \E\bracks*{\sup_{\lambda} L \sum_{i=1}^{L} \hf(\lambda_{i}, \lambda)^{2}} < \infty, \quad
    \lim_{L \to \infty} \E \bracks*{ \sup_{\lambda} \abs*{\sum_{i=1}^{L} \bu_{i} \hf(\lambda_{i}, \lambda) - u(\lambda)}} = 0.
  \end{gather*}
  \begin{proof}
    For the first sequence, define $s_{i,L}(\theta_{i}, \lambda) = L^{2} \hf(\lambda_{i}, \lambda)^{2}$ and $S_{L}(\lambda) = \sum_{i=1}^{L} s_{i,L}(\theta_{i}, \lambda)$. We observe that by (Sup8Int), (SupCond) and Lemma \ref{lem:cvgkernunif},
    \begin{equation*}
      \begin{aligned}
        \limsup_{L \to \infty} \sup_{i} \E{\sn{s_{i,L}(\theta_{i}, \cdot)}^{2}}
        \leq \limsup_{L \to \infty} \sup_{i} L^{4} \E\bracks*{\sup_{\lambda} \braces*{\frac{\psi_{\lambda}(\theta_{i}) p(\lambda)}{\sum_{\ell=1}^{L} \psi_{\lambda_{\ell}}(\theta_{i}) p(\lambda_{\ell})}}^{4}}
     < \infty,
      \end{aligned}
    \end{equation*}
    with an analogous argument showing $\limsup_{L \to \infty} \sup_{i} \E{\lip{s_{i,L}(\theta_{i}, \cdot)}^{2}} < \infty$. Therefore, we may apply Lemma \ref{lem:ullnstat}, and $\limsup_{L \to \infty} \sup_{\lambda} L \sum_{i=1}^{L} \hf(\lambda_{i}, \lambda)^{2} < \infty$ (a.s.). For the second sequence, define $t_{i,L}(\theta_{i}, \lambda) = L u(\lambda_{i}) \hf(\lambda_{i}, \lambda)$ and $T_{L}(\lambda) = \sum_{i=1}^{L} t_{i,L}(\theta_{i}, \lambda)$. We observe that by Lemmas \ref{lem:cvgdensunif} and \ref{lem:cvgkernunif},
    \begin{equation*}
      \begin{aligned}
        \limsup_{L \to \infty} \sup_{i} \E{\sn{t_{i,L}(\theta_{i}, \cdot)}^{2}}
        \leq \lim_{L \to \infty} \sn{\bu}^{2} \limsup_{L \to \infty} \sup_{i} L^{2} \E\bracks*{\sup_{\lambda} \braces*{\frac{\psi_{\lambda}(\theta_{i}) p(\lambda)}{\sum_{\ell=1}^{L} \psi_{\lambda_{\ell}}(\theta_{i}) p(\lambda_{\ell})}}^{2}}
        < \infty,
      \end{aligned}
    \end{equation*}
    with an analogous argument showing $\limsup_{L \to \infty} \sup_{i} \E{\lip{t_{i,L}(\theta_{i}, \cdot)}^{2}} < \infty$ under (Lip8Int), (SupCond). Therefore, Lemma \ref{lem:ullnstat} applies, and $\sup_{\lambda} |\sum_{\ell=1}^{L} u(\lambda_{\ell}) \hf(\lambda_{\ell}, \lambda)  - u(\lambda)| \xrightarrow{a.s.} 0$.
  \end{proof}
\end{lemma}

\section{Technical Lemmas}
\label{app:supp}

\subsection{Convergence of Maximizers}

\begin{proposition}
\label{prop:const-max}  
  Define $\lambda^{*} = \argmax_{\lambda \in \Lambda} u(\lambda)$ and $\widehat{\lambda} = \argmax_{\lambda \in \Lambda} \hu(\lambda)$, and suppose that $\lambda^{*}$ is identified within the compact set $\Lambda$. If
  \begin{equation*}
    \sup_{\lambda \in \Lambda} |\hu(\lambda) - u(\lambda)| \xrightarrow{a.s.} 0
  \end{equation*}
  as $N \to \infty$ or $L \to \infty$, then
  \begin{equation*}
    |\widehat{\lambda} - \lambda^{*}| \xrightarrow{a.s.} 0, \quad
    \E [|\widehat{\lambda} - \lambda^{*}|] \to 0.
  \end{equation*}
  \begin{proof}
    This is akin to standard arguments for empirical risk minimization. We first observe that if $\lambda^{*}$ is identified within the compact set $\Lambda$, then convergence of the maximizers is equivalent to convergence of the respective evaluations of $u$, i.e.
    \begin{equation*}
       u(\widehat{\lambda}) \xrightarrow{a.s.} u(\lambda^{*}),
    \end{equation*}
    We then observe that by definition,
    \begin{equation*}
      \begin{dcases} u(\lambda^{*}) - u(\widehat{\lambda}) \geq 0 \\ \hu(\widehat{\lambda}) - \hu(\lambda^{*}) \geq 0 \end{dcases},
    \end{equation*}
    and adding up those two inequalities, we find that
    \begin{equation*}
        0 \leq u(\lambda^{*}) - \hu(\lambda^{*}) + u(\widehat{\lambda}) - \hu(\widehat{\lambda}) \leq |u(\lambda^{*}) - \hu(\lambda^{*})| + \sup_{\lambda \in \Lambda} |u(\lambda) - \hu(\lambda)|,
    \end{equation*}
    where both terms on the right-hand side vanish almost surely under uniform convergence. Since the sum of the positive terms $u(\lambda^{*}) - u(\widehat{\lambda})$ and $\hu(\widehat{\lambda}) - \hu(\lambda^{*})$ vanishes, both terms vanish individually, and $u(\widehat{\lambda}) \xrightarrow{a.s.} u(\lambda^{*})$
    as required. Convergence in $L^{1}$ follows by bounded convergence.
  \end{proof}
\end{proposition}

\subsection{Further Lemmas for Theorem \ref{thm:stable-errors}}

\begin{lemma}[Convergence of Riemann sums]
  \label{lem:cvgnc}
  Let $\phi: [0, 1] \to \Re$ be Lipschitz-continuous. Then,
  \begin{equation*}
    \abs*{\int_{0}^{1} \phi \dd{\lambda} - L^{-1} \sum_{\ell=1}^{L} \phi({\ell/L})} \leq \frac{\lip{\phi}}{2L}.
  \end{equation*}
  \begin{proof}
    We split the error over $[0, 1]$ into
    \begin{equation*}
        \abs*{\int_{0}^{1} \phi \dd{\lambda} - L^{-1} \sum_{\ell=1}^{L} \phi({\ell/L})}
        \leq \sum_{l=0}^{L-1} \abs*{\int_{\ell/L}^{(\ell+1)/L} \phi \dd{\lambda} - L^{-1} \phi((\ell+1)/L)}.
    \end{equation*}
    For each of the summands,
    \begin{equation*}
      \begin{aligned}
        \abs*{\int_{\ell/L}^{(\ell+1)/L} \phi \dd{\lambda} - L^{-1} \phi((\ell+1)/L)}
        & \leq \abs*{\int_{\ell/L}^{(\ell+1)/L} (\phi(\lambda) - \phi((\ell+1)/L)) \dd{\lambda}} \\
        & \leq \int_{\ell/L}^{(\ell+1)/L} |\phi(\lambda) - \phi((\ell+1)/L)| \dd{\lambda} \\
        & \leq \int_{\ell/L}^{(\ell+1)/L} \lip{\phi} |\lambda - (\ell+1)/L| \dd{\lambda} \\
        & \leq \lip{\phi} \int_{0}^{1/L} x \dd{x} \\
        & \leq \frac{\lip{\phi}}{2L^{2}}.
      \end{aligned}
    \end{equation*}
    The result follows from summing over the mesh.
  \end{proof}
\end{lemma}

\begin{lemma}[Convergence of subsampled densities on the grid]
  \label{lem:cvgdens}
  Let $q: [0, 1] \to (0, \infty)$ be Lipschitz-continuous, and define the mass function $v = q / \sum_{\ell=1}^{L} q(\ell/L)$ and the density $\tilde{v} = q / \int_{0}^{1} q \dd{\lambda}$. Then, for $\ell = 1, \dots L$,
  \begin{gather*}
    |L v({\ell/L}) - \tilde{v}({\ell/L})| \leq \lip{\tilde{v}} v({\ell/L}) / 2 \leq  L^{-1} (\lip{\tilde{v}} / 2 + \tilde{v}({\ell/L})) \lip{\tilde{v}} / 2, \\
    v({\ell/L}) \leq L^{-1} (\lip{\tilde{v}} v({\ell/L}) / 2 + \tilde{v}({\ell/L})) \leq L^{-1} (\lip{\tilde{v}} / 2 + \tilde{v}({\ell/L})),
  \end{gather*}
  \begin{proof}
    For any $\ell = 1, \dots L$,
    \begin{equation*}
      \begin{aligned}
        |L v({\ell/L}) - \tilde{v}({\ell/L})|
        & = \abs*{\parens*{L - \frac{\sum_{\ell=1}^{L} q({\ell/L})}{\int_{0}^{1} q \dd{\lambda}}} v({\ell/L})}\\
        & \leq L \abs*{\int_{0}^{1} \tilde{v} \dd{\lambda} - L^{-1} \sum_{\ell=1}^{L} \tilde{v}({\ell/L})} v({\ell/L}) \\
        & \leq \lip{\tilde{v}} v({\ell/L}) / 2
      \end{aligned}
    \end{equation*}
    by Lemma \ref{lem:cvgnc}. This implies that
    \begin{equation*}
      \begin{aligned}
        v({\ell/L})
        & \leq L^{-1} (|L v({\ell/L}) - \tilde{v}({\ell/L})| + \tilde{v}({\ell/L})) \\
        & \leq L^{-1} (\lip{\tilde{v}} v({\ell/L}) / 2 + \tilde{v}({\ell/L})) \\
        & \leq L^{-1} (\lip{\tilde{v}} / 2 + \tilde{v}({\ell/L})),
      \end{aligned}
    \end{equation*}
    and substituting back,
    \begin{equation*}
      |L v({\ell/L}) - \tilde{v}({\ell/L})| \leq L^{-1} (\lip{\tilde{v}} / 2 + \tilde{v}({\ell/L})) \lip{\tilde{v}} / 2.
    \end{equation*}
  \end{proof}
\end{lemma}

\begin{lemma}[Convergence of subsampled moments]
  \label{lem:cvgexp}
  Let $q: [0, 1] \to (0, \infty)$ be Lipschitz-continuous, and define the mass function $v = q / \sum_{\ell=1}^{L} q(\ell/L)$ as well as the density $\tilde{v} = q / \int_{0}^{1} q \dd{\lambda}$. Then, for Lipschitz-continuous functions $\phi: [0, 1] \to \Re$ with $\sn{\phi} \leq \lip{\phi}$,
  \begin{equation*}
    |\iprod{\phi, 1}_{\bv} - \iprod{\phi, 1}_{\tilde{v}}| \leq \frac{\lip{\phi}}{2L} c_{1}, \quad
    |\iprod{\phi, \phi}_{\bv} - \iprod{\phi, \phi}_{\tilde{v}}| \leq \frac{\lip{\phi}^{2}}{L} c_{1},
  \end{equation*}
  where $c_{1} = \sn{\tilde{v}} + 2 \lip{\tilde{v}}$.
  \begin{proof}
    By the triangle inequality, the error decomposes to
    \begin{equation*}
      \begin{aligned}
        & \abs*{\sum_{\ell=1}^{L} q({\ell/L}) \phi({\ell/L}) - \int_{0}^{1} \tilde{q} \phi \dd{\lambda}} \\
        & \quad \leq \abs*{\sum_{\ell=1}^{L} (q({\ell/L}) - L^{-1} \tilde{q}({\ell/L})) \phi({\ell/L})} + \abs*{L^{-1} \sum_{\ell=1}^{L} \tilde{q}({\ell/L}) \phi({\ell/L}) - \int_{0}^{1} \tilde{q} \phi \dd{\lambda}}.
      \end{aligned}
    \end{equation*}
    We rearrange the first term by Jensen's inequality such that we can apply Lemma \ref{lem:cvgdens}, obtaining
    \begin{equation*}
      \begin{aligned}
        \abs*{\sum_{\ell=1}^{L} (v({\ell/L}) - L^{-1} \tilde{v}({\ell/L})) \phi({\ell/L})}
        & \leq L^{-1} \sum_{\ell=1}^{L} |L v({\ell/L}) - \tilde{v}({\ell/L})| |\phi({\ell/L})| \\
        & \leq L^{-1} \sn{\phi} \sum_{\ell=1}^{L} |L v({\ell/L}) - \tilde{v}({\ell/L})| \\
        & \leq L^{-1} \sn{\phi} \lip{\tilde{v}} \parens*{\sum_{\ell=1}^{L} v({\ell/L})} / 2 \\
        & \leq \frac{\sn{\phi} \lip{\tilde{v}}}{2L}.
      \end{aligned}
    \end{equation*}
    For the second term, applying Lemma \ref{lem:cvgnc} to $\tilde{v} \phi$ yields
    \begin{equation*}
      \abs*{L^{-1} \sum_{\ell=1}^{L} \tilde{v}({\ell/L}) \phi({\ell/L}) - \int_{0}^{1} \tilde{v} \phi \dd{\lambda}}
      \leq \frac{\lip{\tilde{v} \phi}}{2L}
      \leq \frac{\lip{\tilde{v}} \sn{\phi} + \sn{\tilde{v}} \lip{\phi}}{2L}.
    \end{equation*}
    In particular, where $\sn{\phi} \leq \lip{\phi}$ and defining $c_{1} = \sn{\tilde{v}} + 2 \lip{\tilde{v}}$,
    \begin{equation*}
      \abs*{\sum_{\ell=1}^{L} v({\ell/L}) \phi({\ell/L}) - \int_{0}^{1} \tilde{v} \phi \dd{\lambda}} \leq \frac{\lip{\phi}}{2L} c_{1},
    \end{equation*}
    We may immediately apply this for $\phi^{2}$, where $\lip{\phi^{2}} \leq 2 \sn{\phi} \lip{\phi} \leq 2 \lip{\phi}^{2}$, and hence
    \begin{equation*}
      \abs*{\sum_{\ell=1}^{L} v({\ell/L}) \phi^{2}({\ell/L}) - \int_{0}^{1} \tilde{v} \phi^{2} \dd{\lambda}} \leq \frac{\lip{\phi}^{2}}{L} c_{1}.
    \end{equation*}
  \end{proof}
\end{lemma}

\begin{lemma}[Convergence of Riemann sums in $L^{p}(\mu)$]
  \label{lem:cvgnclp}
  Let $\mu$ be a measure on $\Theta$, and suppose that $\varphi: \Theta \times [0, 1] \to \Re$ satisfies
  \begin{equation*}
    \sup_{\lambda', \lambda''} \frac{|\varphi(\cdot, \lambda'') - \varphi(\cdot, \lambda')|_{L^{p}(\mu)}}{|\lambda'' -  \lambda'|} \leq K_{\mu}^{p}[\varphi].
  \end{equation*}
  Then,
  \begin{equation*}
    \abs*{\int_{0}^{1} \varphi(\cdot, \lambda) \dd{\lambda} - L^{-1} \sum_{\ell=1}^{L} \varphi(\cdot, {\ell/L})}_{L^{p}(\mu)} \leq \frac{K_{\mu}^{p}[\varphi]}{2L}.
  \end{equation*}
  \begin{proof}
    Analogously to Lemma \ref{lem:cvgnc}, we split the error over $[0, 1]$ into
    \begin{equation*}
      \begin{aligned}
        \abs*{\int_{0}^{1} \varphi(\cdot, \lambda) \dd{\lambda} - L^{-1} \sum_{\ell=1}^{L} \varphi(\cdot, {\ell/L})}_{L^{p}(\mu)}
        & \leq \sum_{i=0}^{L-1} \abs*{\int_{\ell/L}^{(\ell+1)/L} \varphi(\cdot, \lambda) \dd{\lambda} - L^{-1} \varphi(\cdot, (\ell+1)/L)}_{L^{p}(\mu)}.
      \end{aligned}
    \end{equation*}
    For each of the summands, we apply the Minkowski integral inequality, giving
    \begin{equation*}
      \begin{aligned}
        & \abs*{\int_{\ell/L}^{(\ell+1)/L} \varphi(\cdot, \lambda) \dd{\lambda} - L^{-1} \varphi(\cdot, (\ell+1)/L)}_{L^{p}(\mu)} \\
        & \leq \abs*{\int_{\ell/L}^{(\ell+1)/L} (\varphi(\cdot, \lambda) - \varphi(\cdot, (\ell+1)/L)) \dd{\lambda}}_{L^{p}(\mu)} \\
        & \leq \int_{\ell/L}^{(\ell+1)/L} |\varphi(\cdot, \lambda) - \varphi(\cdot, (\ell+1)/L)|_{L^{p}(\mu)} \dd{\lambda} \\
        & \leq \int_{\ell/L}^{(\ell+1)/L} K_{\mu}^{p}[\varphi] |\lambda - (\ell+1)/L| \dd{\lambda} \\
        & \leq K_{\mu}^{p}[\varphi] \int_{0}^{1/L} x \dd{x} \\
        & \leq \frac{K_{\mu}^{p}[\varphi]}{2L^{2}}.
      \end{aligned}
    \end{equation*}
    The result follows from summing over the mesh.
  \end{proof}
\end{lemma}

\begin{lemma}[Convergence of subsampled kernels on the grid]
  \label{lem:cvgkern}
  Let $\mu$ be a measure on $\Theta$, $\varphi: \Theta \times [0, 1] \to (0, \infty)$, $\kappa(\theta, \lambda) = \varphi(\theta, \lambda) / \sum_{\ell=1}^{L} \varphi(\theta, \ell/L)$, $\tilde{\kappa}(\theta, \lambda) = \varphi / \int_{0}^{1} \varphi(\theta, \lambda') \dd{\lambda'}$, and define $K$ as in Lemma \ref{lem:cvgnclp}. Then,
  \begin{gather*}
    |L \kappa(\cdot, {\ell/L}) - \tilde{\kappa}(\cdot, {\ell/L})|_{L^{p}(\mu)} \leq \frac{K_{\mu}^{2p}[\tilde{\kappa}] (K_{\mu}^{2p}[\tilde{\kappa}] / 2 + |\tilde{\kappa}(\cdot, {\ell/L})|_{L^{2p}(\mu)})}{2L}, \\
    |\kappa(\cdot, {\ell/L})|_{L^{p}(\mu)} \leq \frac{|\tilde{\kappa}(\cdot, {\ell/L})|_{L^{p}(\mu)} + K_{\mu}^{p}[\tilde{\kappa}] / 2}{L},
  \end{gather*}
  \begin{proof}
    We begin by factorizing the error:
    \begin{equation*}
      \begin{aligned}
        L \kappa(\theta, \lambda) - \tilde{\kappa}(\theta, \lambda)
        & = \parens*{L - \frac{\sum_{\ell=1}^{L} \varphi(\theta, {\ell/L})}{\int_{0}^{1} \varphi(\theta, \lambda') \dd{\lambda'}}} \kappa(\theta, \lambda) \\
        & = L \parens*{\int_{0}^{1} \tilde{\kappa}(\theta, \lambda') \dd{\lambda'} - L^{-1} \sum_{\ell=1}^{L} \tilde{\kappa}(\theta, {\ell/L})} \kappa(\theta, \lambda),
      \end{aligned}
    \end{equation*}
    and appplying the Cauchy-Schwarz inequality and Lemma \ref{lem:cvgnclp},
    \begin{equation*}
      \begin{aligned}
        |L \kappa(\cdot, \lambda') - \tilde{\kappa}(\cdot, \lambda')|_{L^{p}(\mu)}
        & \leq L \abs*{\int_{0}^{1} \tilde{\kappa}(\theta, \lambda) \dd{\lambda} - L^{-1} \sum_{\ell=1}^{L} \tilde{\kappa}(\theta, {\ell/L})}_{L^{2p}(\mu)} |\kappa(\cdot, \lambda')|_{L^{2p}(\mu)} \\
        & \leq |\kappa(\cdot, \lambda')|_{L^{2p}(\mu)} K_{\mu}^{2p}[\tilde{\kappa}] / 2.
      \end{aligned}
    \end{equation*}
    Analogously $|L \kappa(\cdot, \lambda') - \tilde{\kappa}(\cdot, \lambda')|_{L^{p}(\mu)} \leq |\kappa(\cdot, \lambda')|_{L^{\infty}(\mu)} K_{\mu}^{p}[\tilde{\kappa}] / 2$ by Hoelder's inequality. In particular, notice that $\kappa(\theta, {\ell/L}) \leq 1$, so $|L \kappa(\cdot, {\ell/L}) - \tilde{\kappa}(\cdot, {\ell/L})|_{L^{p}(\mu)} \leq K_{\mu}^{p}[\tilde{\kappa}] / 2$, and by the triangle inequality,
    \begin{equation*}
      |\kappa(\cdot, {\ell/L})|_{L^{p}(\mu)} \leq L^{-1} (|\tilde{\kappa}(\cdot, {\ell/L})|_{L^{p}(\mu)} + K_{\mu}^{p}[\tilde{\kappa}] / 2).
    \end{equation*}
    Inserting that back,
    \begin{equation*}
      |L \kappa(\cdot, {\ell/L}) - \tilde{\kappa}(\cdot, {\ell/L})|_{L^{p}(\mu)} \leq \frac{K_{\mu}^{2p}[\tilde{\kappa}] (K_{\mu}^{2p}[\tilde{\kappa}] / 2 + |\tilde{\kappa}(\cdot, {\ell/L})|_{L^{2p}(\mu)})}{2L}.
    \end{equation*}
  \end{proof}
\end{lemma}

\begin{lemma}[Strong convergence of subsampled Markov operators]
  \label{lem:cvgop}
  Define all symbols as in Lemma \ref{lem:cvgkern}, let $\pi_{\lambda}(\dd{\theta}) \propto \varphi(\theta, \lambda) \dd{\theta}$ be measures on $\Theta$, $f(\lambda, \lambda') = |\kappa(\cdot, \lambda')|_{L^{1}(\pi_{\lambda})}$ as well as $\tilde{f}(\lambda, \lambda') = |\tilde{\kappa}(\cdot, \lambda')|_{L^{1}(\pi_{\lambda})}$. Moreover, define the operators $(F\phi)(\lambda) = \sum_{\ell=1}^{L} f(\lambda, \ell/L) \phi(\ell/L)$ and $(\tilde{F}\phi)(\lambda) = \int_{0}^{1} \tilde{f}(\lambda, \lambda') \phi(\lambda') \dd{\lambda'}$. Then, for Lipschitz-continuous $\phi: [0, 1] \to \Re$ with $\sn{\phi} \leq \lip{\phi}$,
  \begin{equation*}
    \sn{F\phi - \tilde{F}\phi} \leq \frac{\lip{\phi}}{2L} c_{2}, \quad
    \sn{F\phi^{2} - \tilde{F}\phi^{2}} \leq \frac{\lip{\phi}^{2}}{L} c_{2},
  \end{equation*}
  where $c_{2} = \sup_{\lambda, \lambda'} \braces{K_{\pi_{\lambda}}^{2}[\tilde{\kappa}] (K_{\pi_{\lambda}}^{2}[\tilde{\kappa}] / 2 + |\tilde{\kappa}(\cdot, \lambda')|_{L^{2}(\pi_{\lambda})}) + \tilde{f}(\lambda, \lambda') + K_{\pi_{\lambda}}^{1}[\tilde{\kappa}]}$.
  \begin{proof}
    By Jensen's inequality,
    \begin{equation*}
      \begin{aligned}
        |(F\phi)(\lambda) - (\tilde{F}\phi)(\lambda)|
        & = \abs*{\int \parens*{\sum_{\ell=1}^{L} \kappa(\theta, {\ell/L}) \phi({\ell/L}) - \int \tilde{\kappa}(\theta, \lambda') \phi(\lambda') \dd{\lambda'}} \dd{\pi_{\lambda}}}, \\
        & \leq \abs*{\sum_{\ell=1}^{L} \kappa(\cdot, {\ell/L}) \phi({\ell/L}) - \int \tilde{\kappa}(\cdot, \lambda') \phi(\lambda') \dd{\lambda'}}_{L^{1}(\pi_{\lambda})},
      \end{aligned}
    \end{equation*}
    so we may follow the course of Lemma \ref{lem:cvgexp} inside the integral. By the triangle inequality, the error decomposes to
    \begin{equation*}
      \begin{aligned}
        & \abs*{\sum_{\ell=1}^{L} \kappa(\cdot, {\ell/L}) \phi({\ell/L}) - \int_{0}^{1} \tilde{\kappa}(\cdot, \lambda') \phi(\lambda') \dd{\lambda'}}_{L^{1}(\pi_{\lambda})} \\
        & \quad \leq \abs*{\sum_{\ell=1}^{L} (\kappa(\cdot, {\ell/L}) \phi({\ell/L}) - L^{-1} \tilde{\kappa}(\cdot, {\ell/L}) \phi({\ell/L}))}_{L^{1}(\pi_{\lambda})} \\
        & \qquad + \abs*{L^{-1} \sum_{\ell=1}^{L} \tilde{\kappa}(\cdot, {\ell/L}) \phi({\ell/L}) - \int_{0}^{1} \tilde{\kappa}(\cdot, \lambda') \phi(\lambda') \dd{\lambda'}}_{L^{1}(\pi_{\lambda})}.
      \end{aligned}
    \end{equation*}
    We rearrange the first term such that we can apply Lemma \ref{lem:cvgkern}, obtaining
    \begin{equation*}
      \begin{aligned}
        \abs*{\sum_{\ell=1}^{L} (\kappa(\cdot, {\ell/L}) \phi({\ell/L}) - L^{-1} \tilde{\kappa}(\cdot, {\ell/L}) \phi({\ell/L}))}_{L^{1}(\pi_{\lambda})}
        & \leq L^{-1} \sum_{\ell=1}^{L} |\phi({\ell/L})| |L \kappa(\cdot, {\ell/L}) - \tilde{\kappa}(\cdot, {\ell/L})|_{L^{1}(\pi_{\lambda})}  \\
        & \leq \sn{\phi} \sup_{\ell = 1,\dots,L} |L \kappa(\cdot, {\ell/L}) - \tilde{\kappa}(\cdot, {\ell/L})|_{L^{1}(\pi_{\lambda})} \\
        & \leq \sn{\phi} \frac{K_{\pi_{\lambda}}^{2}[\tilde{\kappa}] (K_{\pi_{\lambda}}^{2}[\tilde{\kappa}] / 2 + \sup_{\lambda'} |\tilde{\kappa}(\cdot, \lambda')|_{L^{2}(\pi_{\lambda})})}{2L}.
      \end{aligned}
    \end{equation*}
    As to the second term, we observe that
    \begin{equation*}
      \begin{aligned}
        \sup_{\lambda', \lambda''} \frac{|\tilde{\kappa}(\cdot, \lambda'') \phi(\lambda'') - \tilde{\kappa}(\cdot, \lambda') \phi(\lambda')|_{L^{1}(\pi_{\lambda})}}{|\lambda'' - \lambda'|} \leq \lip{\phi} \sup_{\lambda'} \tilde{f}(\lambda, \lambda') + \sn{\phi} \sup_{\lambda} K_{\pi_{\lambda}}^{1}[\tilde{\kappa}],
      \end{aligned}
    \end{equation*}
    and so by Lemma \ref{lem:cvgnclp}
    \begin{equation*}
      \abs*{L^{-1} \sum_{\ell=1}^{L} \tilde{\kappa}(\cdot, {\ell/L}) \phi({\ell/L}) - \int_{0}^{1} \tilde{\kappa}(\cdot, \lambda') \phi(\lambda') \dd{\lambda'}}_{L^{1}(\pi_{\lambda})} \leq \frac{\lip{\phi} \sup_{\lambda'} \tilde{f}(\lambda, \lambda') + \sn{\phi} K_{\pi_{\lambda}}^{1}[\tilde{\kappa}]}{2L}.
    \end{equation*}
    Defining
    \begin{equation*}
      c_{2} = \sup_{\lambda, \lambda'} \braces{K_{\pi_{\lambda}}^{2}[\tilde{\kappa}] (K_{\pi_{\lambda}}^{2}[\tilde{\kappa}] / 2 + |\tilde{\kappa}(\cdot, \lambda')|_{L^{2}(\pi_{\lambda})}) + \tilde{f}(\lambda, \lambda') + K_{\pi_{\lambda}}^{1}[\tilde{\kappa}]},
    \end{equation*}
    and noting that $\sn{\phi} < \lip{\phi}$,
    \begin{equation*}
      |(F\phi)(\lambda) - (\tilde{F}\phi)(\lambda)| \leq \frac{\lip{\phi}}{2L} c_{2}.
    \end{equation*}
    We may apply this to $\phi^{2}$, which has best Lipschitz constant $\lip{\phi^{2}} \leq 2 \sn{\phi} \lip{\phi} \leq 2 \lip{\phi}^{2}$, and therefore $|(F\phi^{2})(\lambda) - (\tilde{F}\phi^{2})(\lambda)| \leq L^{-1} \lip{\phi}^{2} c_{2}$.
  \end{proof}
\end{lemma}

\begin{lemma}[Convergence of subsampled quadratic forms]
  \label{lem:cvgqf}
  Define all symbols as in Lemmas \ref{lem:cvgexp} and \ref{lem:cvgop}, and let $q(\lambda) = \int \varphi(\theta, \lambda) \dd{\theta} < \infty$. Then, for Lipschitz-continuous functions $\phi: [0, 1] \to \Re$ with $\sn{\phi} \leq \lip{\phi}$,
  \begin{equation*}
    |\iprod{F\phi, F\phi}_{\bv} - \iprod{\tilde{F}\phi, \tilde{F}\phi}_{\tilde{v}}| \leq \frac{\lip{\phi}^{2}}{L} c_{3},
  \end{equation*}
  where $c_{3} = c_{2} + \sn{1/\tilde{v}}^{2} (\sup_{\lambda} K_{\pi_{\lambda}}^{1}[\tilde{\kappa}] + \lip{\tilde{v}})^{2} c_{1}$.
  \begin{proof}
    By the triangle inequality,
    \begin{equation*}
      |\iprod{F\phi, F\phi}_{\bv} - \iprod{\tilde{F}\phi, \tilde{F}\phi}_{\tilde{v}}| \leq |\iprod{\tilde{F}\phi, \tilde{F}\phi}_{\bv} - \iprod{\tilde{F}\phi, \tilde{F}\phi}_{\tilde{v}}| + |\iprod{F\phi, F\phi}_{\bv} - \iprod{\tilde{F}\phi, \tilde{F}\phi}_{\bv}|.
    \end{equation*}
    By Lemma \ref{lem:cvgexp}, the first term admits the bound
    \begin{equation*}
      |\iprod{\tilde{F}\phi, \tilde{F}\phi}_{\bv} - \iprod{\tilde{F}\phi, \tilde{F}\phi}_{\tilde{v}}| \leq \frac{\lip{\tilde{F}\phi}^{2}}{L} c_{1},
    \end{equation*}
    so we must analyze $\lip{\tilde{F}\phi}$. Observing that $\tilde{f}(\lambda, \lambda') = \tilde{f}(\lambda', \lambda) \tilde{v}(\lambda') / \tilde{v}(\lambda)$, we find that
    \begin{equation*}
      \begin{aligned}
        \tilde{f}(\lambda', \lambda) - \tilde{f}(\lambda'', \lambda)
        & = \tilde{v}(\lambda) \parens*{\frac{\tilde{f}(\lambda, \lambda')}{\tilde{v}(\lambda')} - \frac{\tilde{f}(\lambda, \lambda'')}{\tilde{v}(\lambda'')}} \\
        & = \tilde{v}(\lambda) \parens*{\frac{\tilde{f}(\lambda, \lambda') - \tilde{f}(\lambda, \lambda'')}{\tilde{v}(\lambda')} + \tilde{f}(\lambda, \lambda'') \parens*{\frac{1}{\tilde{v}(\lambda')} - \frac{1}{\tilde{v}(\lambda'')}}} \\
        & = \tilde{v}(\lambda) \parens*{\frac{\tilde{f}(\lambda, \lambda') - \tilde{f}(\lambda, \lambda'')}{\tilde{v}(\lambda')} + \frac{\tilde{f}(\lambda'', \lambda)}{\tilde{v}(\lambda)} \parens*{\frac{\tilde{v}(\lambda'')}{\tilde{v}(\lambda')} - 1}} \\
        & = \frac{\tilde{v}(\lambda)}{\tilde{v}(\lambda')} (\tilde{f}(\lambda, \lambda') - \tilde{f}(\lambda, \lambda'')) + \frac{\tilde{f}(\lambda'', \lambda)}{\tilde{v}(\lambda')} (\tilde{v}(\lambda'') - \tilde{v}(\lambda')),
      \end{aligned}
    \end{equation*}
    and
    \begin{equation*}
      \begin{aligned}
        \lip{\tilde{F}\phi}
        & = \sup_{\lambda', \lambda''} \abs*{\frac{\int_{0}^{1} \phi(\lambda) (\tilde{f}(\lambda', \lambda) - \tilde{f}(\lambda'', \lambda)) \dd{\lambda}}{\lambda' - \lambda''}} \\
        & = \sup_{\lambda', \lambda''} \abs*{\frac{\int_{0}^{1} \phi(\lambda) \tilde{v}(\lambda) (\tilde{f}(\lambda, \lambda') - \tilde{f}(\lambda, \lambda'')) \dd{\lambda} + \int_{0}^{1} \phi(\lambda) \tilde{f}(\lambda'', \lambda) (\tilde{v}(\lambda'') - \tilde{v}(\lambda')) \dd{\lambda}}{\tilde{v}(\lambda')(\lambda' - \lambda'')}} \\
        & \leq \sn{1/\tilde{v}} \int_{0}^{1} p(\lambda) \abs{\phi(\lambda)} \sup_{\lambda', \lambda''} \abs*{\frac{\tilde{f}(\lambda, \lambda') - \tilde{f}(\lambda, \lambda'')}{\lambda' - \lambda''}} \dd{\lambda} \\
        & \qquad + \sn{1/\tilde{v}} \sup_{\lambda', \lambda''} \abs*{\frac{\tilde{v}(\lambda') - \tilde{v}(\lambda'')}{\lambda' - \lambda''}} \sup_{\lambda''} \abs*{\int_{0}^{1} \phi(\lambda) \tilde{f}(\lambda'', \lambda) \dd{\lambda}} \\
        & \leq \sn{1/\tilde{v}} \parens*{\int_{0}^{1} \tilde{f}(\lambda) \abs{\phi(\lambda)} \lip{\tilde{f}(\lambda, \cdot)} \dd{\lambda} + \lip{\tilde{v}} |\tilde{F}\phi|_{\infty}} \\
        & \leq \sn{1/\tilde{v}} \parens*{\sup_{\lambda} \lip{\tilde{f}(\lambda, \cdot)} \int_{0}^{1} \tilde{v} \abs{\phi} \dd{\lambda} + \lip{\tilde{v}} \lip{\phi}} \\
        & \leq \lip{\phi} \sn{1/\tilde{v}} (\sup_{\lambda} K_{\pi_{\lambda}}^{1}[\tilde{\kappa}] + \lip{\tilde{v}}),
      \end{aligned}
    \end{equation*}
    which, in terms of $\lip{\phi}$, yields the bound
    \begin{equation*}
      |\iprod{\tilde{F}\phi, \tilde{F}\phi}_{\bv} - \iprod{\tilde{F}\phi, \tilde{F}\phi}_{\tilde{v}}| \leq \frac{\lip{\phi}^{2}}{L} \sn{1/\tilde{v}}^{2} (\sup_{\lambda} K_{\pi_{\lambda}}^{1}[\tilde{\kappa}] + \lip{\tilde{v}})^{2} c_{1}.
    \end{equation*}
    As for the second term,
    \begin{equation*}
      \begin{aligned}
        |\iprod{F\phi, F\phi}_{\bv} - \iprod{\tilde{F}\phi, \tilde{F}\phi}_{\bv}|
        & \leq \sum_{\ell=1}^{L} v({\ell/L}) |(F\phi)^{2}({\ell/L}) - (\tilde{F}\phi)^{2}({\ell/L})| \\
        & \leq \sum_{\ell=1}^{L} v({\ell/L}) |(F\phi)({\ell/L}) + (\tilde{F}\phi)({\ell/L})| |(F\phi)({\ell/L}) - (\tilde{F}\phi)({\ell/L})|.
      \end{aligned}
    \end{equation*}
    Since $|(F\phi)({\ell/L}) + (\tilde{F}\phi)({\ell/L})| \leq 2 \sn{\phi} \leq 2 \lip{\phi}$ and $|(F\phi)({\ell/L}) - (\tilde{F}\phi)({\ell/L})| \leq \lip{\phi} c_{2} / (2L)$ by Lemma \ref{lem:cvgop},
    \begin{equation*}
      |\iprod{F\phi, F\phi}_{\bv} - \iprod{\tilde{F}\phi, \tilde{F}\phi}_{\bv}| \leq \frac{\lip{\phi}^{2}}{L} c_{2}.
    \end{equation*}
    We define
    \begin{equation*}
      c_{3} = c_{2} + \sn{1/\tilde{v}}^{2} (\sup_{\lambda} K_{\pi_{\lambda}}^{1}[\tilde{\kappa}] + \lip{\tilde{v}})^{2} c_{1},
    \end{equation*}
    and conclude that
    \begin{equation*}
      |\iprod{F\phi, F\phi}_{\bv} - \iprod{\tilde{F}\phi, \tilde{F}\phi}_{\tilde{v}}| \leq \frac{\lip{\phi}^{2}}{L} c_{3}.
    \end{equation*}
  \end{proof}
\end{lemma}

\begin{lemma}[Smoothness of second eigenvectors]
  \label{lem:smoothev}
  Define all symbols as in Lemma \ref{lem:cvgqf}, and let $\bome$ be an eigenvector of $\bF^{2}$ corresponding to the second-largest eigenvalue $\varepsilon_{2}$, satisfying
  \begin{equation*}
    \varepsilon_{2} \bome = \bF^{2}{}\T \bome.
  \end{equation*}
  If $\lip{\tilde{v}}$, $\sup_{\lambda} K_{\pi_{\lambda}}^{2}[\tilde{\kappa}]$, $\sup_{\lambda, \lambda'} |\tilde{\kappa}|_{L^{2}(\pi_{\lambda})}$ are finite, the linear interpolant $\omega$ of $\bome$ has best Lipschitz constant satisfying
    \begin{equation*}
      \limsup_{L \to \infty} \lip{\omega} < \infty.
    \end{equation*}
    Moreover, $\sn{\omega} \leq \lip{\omega}$. Similarly, the limiting analogue $\tilde{\omega}$ satisfies $\lip{\tilde{\omega}} \leq \tilde{\varepsilon}_{2}^{-1} \sn{1/\tilde{v}} \sup_{\lambda} K_{\pi_{\lambda}}^{1}[\tilde{\kappa}]$ and $|\tilde{\omega}|_{\infty} \leq \lip{\tilde{\omega}}$.
  \begin{proof}
    We begin by investigating the smoothness of $\tilde{\omega}$. Defining $\tilde{f}^{(2)}$ as the kernel corresponding to $\tilde{F}^{2}$,
    \begin{equation*}
      \begin{aligned}
        \lip{\tilde{\omega}}
        & = \tilde{\varepsilon}_{2}^{-1} \lip*{\int_{0}^{1} \tilde{f}^{(2)}(\lambda, \lambda') \tilde{\omega}(\lambda) \dd{\lambda}} \\
        & \leq \tilde{\varepsilon}_{2}^{-1} \int_{0}^{1} \lip{\tilde{f}^{(2)}(\lambda, \cdot)} |\tilde{\omega}(\lambda)| \dd{\lambda} \\
        & \leq \tilde{\varepsilon}_{2}^{-1} \sup_{\lambda} \lip{\tilde{f}^{(2)}(\lambda, \cdot)} \int_{0}^{1} |\tilde{\omega}| \dd{\lambda}.
      \end{aligned}
    \end{equation*}
    Going term by term, by Jensen's and Holder's inequalities,
    \begin{equation*}
      \sup_{\lambda} \lip{\tilde{f}^{(2)}(\lambda, \cdot)} \leq \sup_{\lambda} \int_{0}^{1} \tilde{f}(\lambda, \lambda') \lip{\tilde{f}(\lambda', \cdot)} \dd{\lambda'} \leq \sup_{\lambda} \lip{\tilde{f}(\lambda, \cdot)},
    \end{equation*}
    so smoothness of $\tilde{f}$ transfers to $\tilde{f}^{(2)}$, and again by Jensen's inequality,
    \begin{equation*}
      \begin{aligned}
        \int_{0}^{1} |\tilde{\omega}| \dd{\lambda}
        \leq \sn{1/\tilde{v}} \int_{0}^{1} \tilde{v} |\tilde{\omega}| \dd{\lambda}
        \leq \sn{1/\tilde{v}} \sqrt{\int_{0}^{1} \tilde{v} \tilde{\omega}^{2} \dd{\lambda}}
        = \sn{1/\tilde{v}},
      \end{aligned}
    \end{equation*}
    where we used the constraint $\int_{0}^{1} \tilde{v} \tilde{\omega}^{2} \dd{\lambda} = 1$. Therefore, $\tilde{\omega}$ is Lipschitz-continuous with
    \begin{equation*}
      \lip{\tilde{\omega}} \leq \tilde{\varepsilon}_{2}^{-1} \sn{1/\tilde{v}} \sup_{\lambda} \lip{\tilde{f}(\lambda, \cdot)}.
    \end{equation*}
    We also note that since $\iprod{\tilde{\omega}, 1}_{\tilde{v}} = 0$, $\tilde{\omega}$ necessarily crosses 0, and therefore $\sn{\tilde{\omega}} \leq \lip{\tilde{\omega}}$. Having shown the smoothness of the limiting eigenfunction, we now address the ``smoothness'' of the eigenvectors $\bome$. Consider that the linear interpolant $\omega$ of $\bome$ has best Lipschitz constant
    \begin{equation*}
      \begin{aligned}
        \lip{\omega}
        = \sup_{k = 0, \dots, L-1} \frac{|\omega((k+1)/L) - \omega(k/L)|}{L^{-1}},
      \end{aligned}
    \end{equation*}
    corresponding to the maximum slope in between elements of $\bome$. Defining $f^{(2)}(\lambda, \lambda') = \sum_{\ell=1}^{L} f(\lambda, \ell/L) f(\ell/L, \lambda')$, we use the eigenvector equation $\varepsilon_{2} \bome = \bF^{2}{}\T \bome$ and Holder's inequality to decompose the maximum slope into
    \begin{equation*}
      \begin{aligned}
        \lip{\omega}
        & = \frac{1}{\varepsilon_{2}} \sup_{k = 0, \dots, L-1} \frac{\abs*{\sum_{j=1}^{L} \omega(j/L) (f^{(2)}(j/L, (k+1)/L) - f^{(2)}(j/L, k/L))}}{L^{-1}} \\
        & \leq \frac{|\bome|_{1}}{\varepsilon_{2}} \sup_{\lambda, k = 0, \dots, L-1} L |f^{(2)}(\lambda, (k+1)/L) - f^{(2)}(\lambda, k/L)|.
      \end{aligned}
    \end{equation*}
    Using the triangle inequality, the slope of $f^{(2)}$ may be rewritten as
    \begin{equation*}
      \begin{aligned}
        & L |f^{(2)}(\lambda, (k+1)/L) - f^{(2)}(\lambda, k/L)| \\
        & = L \abs*{\sum_{\ell=1}^{L} f(\lambda, \ell/L) (f(\ell/L, (k+1)/L) - f(\ell/L, k/L))} \\
        & = L \abs*{\sum_{\ell=1}^{L} f(\lambda, \ell/L)} \sup_{\lambda} |f(\lambda, (k+1)/L) - f(\lambda, k/L)| \\
        & \leq \sup_{\lambda} L |f(\lambda, (k+1)/L) - f(\lambda, k/L))| \\
        & \leq 2 \sup_{\lambda, \ell} |L f(\lambda, \ell/L) - \tilde{f}(\lambda, \ell/L)| + \sup_{\lambda, k = 0, \dots, L-1} |\tilde{f}(\lambda, (k+1)/L) - \tilde{f}(\lambda, k/L)| \\
        & \leq 2 \sup_{\lambda, \ell} |L f(\lambda, \ell/L) - \tilde{f}(\lambda, \ell/L)| + L^{-1} \sup_{\lambda} \lip{\tilde{f}(\lambda, \cdot)},
      \end{aligned}
    \end{equation*}
    while by Jensen's inequality,
    \begin{equation*}
      \begin{aligned}
        |\bome|_{1}
        \leq \frac{\sum_{\ell=1}^{L} v_{\ell} |\omega_{\ell}|}{\inf_{\ell} v(\ell/L)}
        \leq \frac{\sqrt{\sum_{\ell=1}^{L} v_{\ell} |\omega_{\ell}|^{2}}}{\inf_{\ell} v(\ell/L)}
        = \frac{1}{\inf_{\ell} v(\ell/L)},
      \end{aligned}
    \end{equation*}
    where we used the constraint $\sum_{\ell=1}^{L} v_{\ell} |\omega_{\ell}|^{2} = 1$. In conjunction,
    \begin{equation*}
        \lip{\omega} \leq \frac{2 \sup_{\lambda, \ell} L |L f(\lambda, \ell/L) - \tilde{f}(\lambda, \ell/L)| + \sup_{\lambda} \lip{\tilde{f}(\lambda, \cdot)}}{\varepsilon_{2} \inf_{\ell} L v(\ell/L)},
    \end{equation*}
    where $\lim_{L \to \infty} \inf_{\ell} Lv(\ell/L) > \infty$ by Lemma \ref{lem:cvgdens} and $\limsup_{L \to \infty} \sup_{\lambda, \ell} L |L f(\lambda, \ell/L) - \tilde{f}(\lambda, \ell/L)| < \infty$ by Lemma \ref{lem:cvgkern}. Regarding the eigenvalue, we observe that for any function $\phi$,
    \begin{equation*}
      \varepsilon_{2} \geq \frac{\iprod{F \phi, F \phi}_{\bv} - \iprod{\phi, 1}_{\bv}^{2}}{\iprod{\phi, \phi}_{\bv} - \iprod{\phi, 1}_{\bv}^{2}}.
    \end{equation*}
    In particular, if we consider the limiting eigenfunction $\tilde{\omega}$, for which we've established $\lip{\omega} < \infty$, then by Lemmas \ref{lem:cvgexp} and \ref{lem:cvgqf}
    \begin{equation*}
      \begin{aligned}
        \liminf_{L \to \infty} \varepsilon_{2}
        & \geq \frac{\lim_{L \to \infty} \iprod{F \tilde{\omega}, F \tilde{\omega}}_{\bv} - \lim_{L \to \infty} \iprod{\tilde{\omega}, 1}_{\bv}^{2}}{\lim_{L \to \infty} \iprod{\tilde{\omega}, \tilde{\omega}}_{\bv} - \lim_{L \to \infty} \iprod{\tilde{\omega}, 1}_{\bv}^{2}} \\
        & \geq \frac{\iprod{\tilde{F}\tilde{\omega}, \tilde{F}\tilde{\omega}}_{\tilde{v}} - \iprod{\tilde{\omega}, 1}_{\tilde{v}}^{2}}{\iprod{\tilde{\omega}, \tilde{\omega}}_{\tilde{v}} - \iprod{\tilde{\omega}, 1}_{\tilde{v}}^{2}} \\
        & \geq \tilde{\varepsilon}_{2}.
      \end{aligned}
    \end{equation*}
    $\tilde{\varepsilon}_{2}$ being positive by Lemma \ref{lem:gapexist}, $\limsup_{L \to \infty} \lip{\omega} < \infty$.
  \end{proof}
\end{lemma}

\begin{lemma}[Existence of limiting spectral gap]
  \label{lem:gapexist}
  Let all symbols be defined as in Lemma \ref{lem:cvgqf}, and suppose that $\tilde{f}$ induces an irreducible Gibbs sampler. If $\int \sup_{\lambda} \varphi(\theta, \lambda) \dd{\theta} < \infty$, the Gibbs sampler has positive absolute spectral gap.
  \begin{proof}
    By \cite{liu1995covariance}, if the Gibbs sampler on $\Theta \times \Lambda$ induced by $\tilde{f}$ is ergodic, a sufficient condition for a positive absolute spectral gap in a 2-component Gibbs sampler is
    \begin{equation*}
      \int \tilde{f}(\lambda, \lambda) \dd{\lambda} < \infty.
    \end{equation*}
    We may bound this quantity by
    \begin{equation*}
      \begin{aligned}
        \int \tilde{f}(\lambda, \lambda) \dd{\lambda}
        & = \iint \frac{\varphi^{2}(\theta, \lambda)}{\int \varphi(\theta, \lambda) \dd{\theta} \int \varphi(\theta, \lambda') \dd{\lambda'}} \dd{\lambda} \dd{\theta} \\
        & \leq \frac{1}{\inf_{\lambda} \int \varphi(\theta, \lambda) \dd{\theta}} \int \sup_{\lambda} \varphi(\theta, \lambda) \frac{\int \varphi(\theta, \lambda') \dd{\lambda'}}{\int \varphi(\theta, \lambda') \dd{\lambda'}} \dd{\theta} \\
        & \leq \frac{\int \sup_{\lambda} \varphi(\theta, \lambda) \dd{\theta}}{\inf_{\lambda} \int \varphi(\theta, \lambda) \dd{\theta}}.
      \end{aligned}
    \end{equation*}
  \end{proof}
\end{lemma}

\begin{lemma}[Spectral gap bound for $|(I - F)^{\#}|_{2}$]
  \label{lem:gapbound}
  Let $F$ be a reversible stochastic matrix with absolute spectral gap $\gamma > 0$ and stationary vector $v$. Then,
  \begin{equation*}
    |(I - F)^{\#}|_{2} \leq \frac{\sqrt{|v|_{\infty} |1/v|_{\infty}}}{\gamma}.
  \end{equation*}
  \begin{proof}
    To begin with, we observe that since $F$ is in detailed balance, it can be diagonalized to
    \begin{equation*}
      F = \sqrt{V} W E W^{T} \sqrt{V^{-1}},
    \end{equation*}
    where $V = \diag{v}$, $W$ is an orthonormal matrix, and $E$ is the diagonal matrix of eigenvalues of $F$. The same diagonalization applies to $(I - F)^{\#}$ by way of its Neumann sum representation,
    \begin{equation*}
      \begin{aligned}
        (I - F)^{\#}
        & = \sum_{k=1}^{\infty} (F - F^{\infty})^{k}.
      \end{aligned}
    \end{equation*}
    where $F^{\infty} = 1 v\T$. Since $F$ and $F^{\infty}$ share eigenvectors, so does their difference, and powers and sums thereof, and $(I - F)^{\#}$ is diagonalizable to
    \begin{equation}
    \label{eq:invdiag}
      (I - F)^{\#} = \sqrt{V} W E^{\#} W\T \sqrt{V^{-1}}, \quad
      E^{\#} = \sum_{k=1}^{\infty} (E - E^{\infty})^{k},
    \end{equation}
    where $E^{\#}$ is again a diagonal matrix. Moreover, $E_{ii}^{\infty} = 1$ iff $E_{ii} = 1$ and $E_{ii}^{\infty} = 0$ otherwise, and we find that the elements of $E^{\#}$ are given by
    \begin{equation*}
      E_{ii}^{\#} =
      \begin{dcases}
        0 & E_{ii} = 1 \\
        1 / (1 - E_{ii}) & \text{otherwise}
      \end{dcases},
    \end{equation*}
    and therefore the spectral radius of $(I - F)^{\#}$ is $1 / (1 - \max_{i: E_{ii} \neq  1} |E_{ii}|) = 1 / \gamma$. We may now use the representation \eqref{eq:invdiag} to bound the spectral norm by the spectral radius, i.e.
    \begin{equation*}
    \begin{aligned}
      |(I - F)^{\#}|_{2}
      & = \abs{\sqrt{V} W E^{\#} W\T \sqrt{V^{-1}}}_{2} \\
      & \leq |\sqrt{V}|_{2} |\sqrt{V^{-1}}|_{2} |W E^{\#} W\T|_{2} \\
      & \leq \gamma^{-1} \sqrt{|v|_{\infty} |1/v|_{\infty}},
    \end{aligned}
  \end{equation*}
  where $W E^{\#} W\T$ is symmetrical, and therefore its spectral norm is equal to its spectral radius.
  \end{proof}
\end{lemma}

\begin{lemma}[Scalings of $\braces{\hbF (\hbF - \bF)}_{ij}$]
  \label{lem:randop2}
  Define all symbols as in Lemma \ref{lem:cvgop} and let $\hbF_{ij} = \kappa(\theta_{i}, j/L)$ for $\theta_{i} \sim \pi_{i/L}$. Then,
  \begin{equation*}
    \limsup_{L \to \infty} L^{6} \E\bracks{\braces{\hbF \bE}_{ij}^{4}}
    \leq \parens*{\sup_{\lambda, \lambda'}|\tilde{\kappa}(\cdot, \lambda')|_{L^{8}(\pi_{\lambda})} + \sup_{\lambda} K_{\pi_{\lambda}}^{8}[\tilde{\kappa}] / 2}^{8}.
  \end{equation*}
  \begin{proof}
    For brevity, we provide the proof for $\limsup_{L \to \infty} L^{3} \E\bracks{\braces{\hat{\bF} (\hbF - \bF)}_{ij}^{2}}
    < \infty$. We write $\bE = \hbF - \bF$, where $\E\bracks{\bE} = 0$. We expand $\braces{\hbF \bE}_{ij}^{2}$ to
    \begin{equation*}
      \braces{\hbF \bE}_{ij}^{2}
      = \parens*{\sum_{k=1}^{L} \hbF_{ik} \bE_{kj}}^{2} = \sum_{k, l} \hbF_{ik} \hbF_{il} \bE_{kj} \bE_{lj},
    \end{equation*}
    where many of the terms are many of the terms are 0 in expectation due to independence of rows. The remaining terms admit the brute force bound
    \begin{equation*}
      \begin{aligned}
        \E\bracks{\hbF_{ik} \hbF_{il} \bE_{kj} \bE_{lj}}
        & \leq \parens*{\E\bracks{\hbF_{ik}^{4}} \E\bracks{\hbF_{il}^{4}} \E\bracks{\bE_{kj}^{4}} \E\bracks{\bE_{lj}^{4}}}^{1/4} \\
        & \leq \parens*{\E\bracks{\hbF_{ik}^{4}} \E\bracks{\hbF_{il}^{4}} \E\bracks{\hbF_{kj}^{4}} \E\bracks{\hbF_{lj}^{4}}}^{1/4} \\
        & \leq \sup_{i, j} |\kappa(\cdot, j/L)|_{L^{4}(\pi_{i/L})}^{4}
      \end{aligned}
    \end{equation*}
    due to the fact that the non-central moment is an upper bound on the central moment for positive RVs. Our problem then reduces to counting non-zero terms. We partition the terms as
    \begin{equation*}
      \begin{aligned}
        \braces{\bF \bE}_{ij}^{2}
        & = \hbF_{ii}^{2} \bE_{ij}^{2} + \sum_{k \neq i, l \neq i} \hbF_{ik} \hbF_{il} \bE_{kj} \bE_{lj} + 2 \hbF_{ii} \bE_{ij} \sum_{k \neq i} \hbF_{ik} \bE_{kj} \\
        & = \sum_{k} \hbF_{ik}^{2} \bE_{kj}^{2} + \sum_{k,l: k \neq l, k \neq i, l \neq i} \hbF_{ik} \hbF_{il} \bE_{kj} \bE_{lj} + 2 \sum_{k \neq i} \hbF_{ii} \hbF_{ik} \bE_{ij} \bE_{kj},
      \end{aligned}
    \end{equation*}
    where we have ensured that except for the first sum, all the terms are products over elements of $\bE$ in different rows. Since $\E\bracks{\bE} = \arr{0}$, only the first $L$ terms remain, and
    \begin{equation*}
      \E\bracks*{\braces*{\hbF \bE}_{ij}^{2}}
      \leq L \sup_{i, j} |\kappa(\cdot, j/L)|_{L^{4}(\pi_{i/L})}^{4},
    \end{equation*}
    where
    \begin{equation*}
      \begin{aligned}
        \limsup_{L \to \infty} L^{3} \E\bracks{\braces{\hbF \bE}_{ij}^{2}}
        & \leq \limsup_{L \to \infty} \sup_{i, j} L^{4} |\kappa(\cdot, j/L)|_{L^{4}(\pi_{i/L})}^{4} \\
        & \leq \parens*{\sup_{\lambda, \lambda'}|\tilde{\kappa}(\cdot, \lambda')|_{L^{4}(\pi_{\lambda})} + \sup_{\lambda} K_{\pi_{\lambda}}^{4}[\tilde{\kappa}] / 2}^{4}
      \end{aligned}
    \end{equation*}
    by Lemma \ref{lem:cvgkern}.
  \end{proof}
\end{lemma}

\subsection{Further Lemmas for Theorem \ref{thm:dense-g}}

\begin{lemma}[Uniform convergence of subsampled densities]
  \label{lem:cvgdensunif}
    Define all symbols as in Lemma \ref{lem:cvgdens}. Then, for any (semi-)norm $\norm{\cdot}$ on $([0, 1] \mapsto \Re)$-functions for which $\norm{q} < \infty$,
  \begin{equation*}
    \norm{Lv - \tilde{v}} \leq \frac{\norm{\tilde{v}}}{2L} c_{4}(L), \qquad \limsup_{L \to \infty} c_{4}(L) = 1.
  \end{equation*}
  \begin{proof}
    For any (semi-)norm $\norm{\cdot}$ on $([0, 1] \mapsto \Re)$-functions,
    \begin{equation*}
      \begin{aligned}
        \norm{Lv - \tilde{v}}
        & = \norm*{\parens*{L - \frac{\sum_{\ell=1}^{L} q({\ell/L})}{\int_{0}^{1} q \dd{\lambda}}} v}\\
        & \leq L \abs*{\int_{0}^{1} \tilde{v} \dd{\lambda} - L^{-1} \sum_{\ell=1}^{L} \tilde{v}({\ell/L})} \norm{v} \\
        & \leq \lip{\tilde{v}} \norm{v} / 2
      \end{aligned}
    \end{equation*}
    by Lemma \ref{lem:cvgnc}. Moreover,
    \begin{equation*}
      \norm{v}
      = L^{-1} \frac{q}{L^{-1} \sum_{\ell=1}^{L} q({\ell/L})}
      = L^{-1} \frac{\norm{\tilde{v}}}{L^{-1} \sum_{\ell=1}^{L} \tilde{v}({\ell/L})}
    \end{equation*}
    where $L^{-1} \sum_{\ell=1}^{L} \tilde{v}({\ell/L}) \geq \inf_{\lambda} \tilde{v}(\lambda)$, and also $L^{-1} \sum_{\ell=1}^{L} \tilde{v}({\ell/L}) \geq 1 - \lip{\tilde{v}} / (2L)$ by Lemma \ref{lem:cvgnc}. Using whichever bound is larger,
    \begin{equation*}
      \norm{v} \leq L^{-1} \frac{\norm{\tilde{v}}}{(\inf_{\lambda} \tilde{v}(\lambda)) \lor (1 - \lip{\tilde{v}} / (2L))}
    \end{equation*}
    where $c_{4}(L) = \braces{(\inf_{\lambda} \tilde{v}(\lambda)) \lor (1 - \lip{\tilde{v}} / (2L))}^{-1}$ is finite, and $\limsup_{L \to \infty} c_{4}(L) = 1$. Inserting back,
    \begin{equation*}
      \norm{Lv - \tilde{v}} \leq \frac{\lip{\tilde{v}} \norm{\tilde{v}}}{2L} c_{4}(L).
    \end{equation*}
  \end{proof}
\end{lemma}

\begin{lemma}[Uniform Convergence of subsampled kernels]
  \label{lem:cvgkernunif}
  Define all symbols as in Lemma \ref{lem:cvgkern}. For a (semi-)norm $\norm{\cdot}$ on $([0, 1] \mapsto \Re)$-functions, define the composite (semi-)norm $|\varphi|_{L^{p}(\pi_{\lambda}) \times \norm{\cdot}} = \parens*{\int \norm{\varphi(\theta, \cdot)}^{p} \dd{\pi_{\lambda}}}^{1/p}$. Then,
  \begin{equation*}
    \begin{aligned}
      |L \kappa - \tilde{\kappa}|_{L^{p}(\pi_{\lambda}) \times \norm{\cdot}}
      & \leq \frac{K_{2p}[\tilde{\kappa}]}{2} |\tilde{\kappa}|_{L^{2p}(\pi_{\lambda}) \times \norm{\cdot}} 
    \end{aligned}
  \end{equation*}
  In particular, if $K_{2p}[\tilde{\kappa}] < \infty$, $\limsup_{L \to \infty} \int \norm{\varphi(\theta, \cdot)}^{2p}/ \parens*{\sum_{\ell=1}^{L} \varphi(\theta, {\ell/L})}^{2p-1} \dd{\theta} < \infty$, $\sup_{\ell} |\tilde{\kappa}|_{L^{2p}(\pi_{\ell/L}) \times \norm{\cdot}} < \infty$,
  \begin{equation*}
    \limsup_{L \to \infty} \sup_{\ell} L |\kappa|_{L^{p}(\pi_{\ell/L}) \times \norm{\cdot}} < \infty.
  \end{equation*}
  \begin{proof}
    As in Lemma \ref{lem:cvgkern},
    \begin{equation*}
      |L \kappa - \tilde{\kappa}|_{L^{p}(\pi_{\lambda}) \times \norm{\cdot}}
      \leq |\kappa|_{L^{2p}(\pi_{\lambda}) \times \norm{\cdot}} K_{2p}[\tilde{\kappa}] / 2,
    \end{equation*}
    and by the triangle inequality,
    \begin{equation*}
      |\kappa|_{L^{p}(\pi_{\lambda}) \times \norm{\cdot}} \leq L^{-1} (|\tilde{\kappa}|_{L^{p}(\pi_{\lambda}) \times \norm{\cdot}} + |\kappa|_{L^{2p}(\pi_{\lambda}) \times \norm{\cdot}} K_{2p}[\tilde{\kappa}] / 2).
    \end{equation*}
    In particular, since $|\kappa|_{L^{8}(\pi_{i/L}) \times \norm{\cdot}}^{8} \int \varphi(\theta, i/L) \dd{\theta} \leq \int \norm{\varphi(\theta, \cdot)}^{2p} / \parens*{\sum_{\ell=1}^{L} \varphi(\theta, \ell/L)}^{2p - 1} \dd{\theta}$, we have that
    \begin{equation*}
      \begin{aligned}
        \sup_{\ell} L |\kappa|_{L^{p}(\pi_{\lambda}) \times \norm{\cdot}}
        & \leq \sup_{\lambda} |\tilde{\kappa}|_{L^{p}(\pi_{\lambda}) \times \norm{\cdot}} + \frac{K_{2p}[\tilde{\kappa}]}{2} \parens*{\frac{\int \norm{\varphi(\theta, \cdot)}^{2p}/ \parens*{\sum_{\ell=1}^{L} \varphi(\theta, \ell/L)}^{2p-1} \dd{\theta}}{\inf_{\lambda} \int \varphi(\theta, \lambda) \dd{\theta}}}^{1/p}
      \end{aligned}
    \end{equation*}
  \end{proof}
\end{lemma}

\begin{lemma}[Uniform LLN for non-stationary summands]
  \label{lem:ullnstat}
  Define $S_{L}(\lambda) = \sum_{\ell=1}^{L} s_{\ell,L}(\theta_{\ell}, \lambda)$, where the $\theta_{\ell}$ are independent and the $s_{\ell,L}$ satisfy
  \begin{equation*}
    \limsup_{L \to \infty} \sup_{\ell = 1, \dots, L} \E\bracks{\sn{s_{\ell,L}(\theta_{\ell}, \cdot)}^{2}} < \infty, \quad
    \limsup_{L \to \infty} \sup_{\ell = 1, \dots, L} \E\bracks{\lip{s_{\ell,L}(\theta_{\ell}, \cdot)}^{2}} < \infty.
  \end{equation*}
  Then, for $\mu_{L}(\lambda) = \E\bracks{L^{-1} S_{L}(\lambda)}$,
  \begin{equation*}
    \lim_{L \to \infty} \sup_{\lambda} |L^{-1} S_{L}(\lambda) - \mu_{L}(\lambda)| = 0 \ (a.s.), \quad
    \lim_{L \to \infty} \E\bracks{\sup_{\lambda} |L^{-1} S_{L}(\lambda) - \mu_{L}(\lambda)|^{2}} = 0.
  \end{equation*}
  \begin{proof}
    We use the framework of Theorem \ref{thm:ullnlip}, where we have to establish that $|L^{-1} S_{L}(\lambda) - \mu_{L}(\lambda)| \xrightarrow{a.s.} 0$ for any given $\lambda \in \Lambda$, and that $L^{-1} S_{L}(\lambda)$ is \emph{strongly stochastically equicontinuous} (SSE). Pointwise convergence is then verified by a small modification to Kolmogorov's SLLN test, where we account for non-stationarity of summands by checking that $\limsup_{L \to \infty} \sup_{\ell} \Var\bracks{s_{\ell,L}(\theta_{\ell}, \lambda)} \leq \limsup_{L \to \infty} \sup_{\ell} \E\bracks{\sn{s_{\ell,L}(\theta_{\ell}, \cdot)}^{2}} < \infty$. This is sufficient to prove a variant of Kolmogorov's inequality, which Kolmogorov's SLLN test relies on. Thus, the test passes, and $|L^{-1} S_{L}(\lambda) - \mu_{L}(\lambda)| \xrightarrow{a.s.} 0$ for any given $\lambda$. It remains to verify SSE, where for any sequence $T_{L}$ with $T_{L} \geq \lip{S_{L}}$ a.s., Theorem \ref{thm:ullnlip} gives the sufficient condition $|L^{-1} T_{L} - \E\bracks{L^{-1} T_{L}}| \xrightarrow{a.s.} 0$ and $\limsup_{L \to \infty} \E\bracks{L^{-1} T_{L}} < \infty$. To obtain a sum of independent terms, it is convenient to use the bound
    \begin{equation*}
      T_{L} = \sum_{\ell=1}^{L} \lip{s_{\ell,L}(\theta_{\ell}, \cdot)} \geq \lip{S_{L}}.
    \end{equation*}
    To show convergence, we again apply the non-stationary Kolmogorov test, where $\limsup_{L \to \infty} \sup_{\ell} \Var\bracks{\lip{s_{\ell,L}(\theta_{\ell}, \cdot)}} \leq \limsup_{L \to \infty} \sup_{\ell} \E\bracks{\lip{s_{\ell,L}(\theta_{\ell}, \cdot)}^{2}} < \infty$, and the test passes. For the finiteness of the limit,
    \begin{equation*}
      \begin{aligned}
        \limsup_{L \to \infty} \E\bracks{L^{-1} T_{L}}
        & \leq \limsup_{L \to \infty} L^{-1} \sum_{\ell=1}^{L} \E\bracks{\lip{s_{\ell,L}(\theta_{\ell}, \cdot)}} \\
        & \leq \limsup_{L \to \infty} \sqrt{\sup_{\ell} \E\bracks{\lip{s_{\ell,L}(\theta_{\ell}, \cdot)}^{2}}} \\
        & < \infty.
      \end{aligned}
    \end{equation*}
    Hence, $|L^{-1} T_{L} - \E\bracks{L^{-1} T_{L}}| \xrightarrow{a.s.} 0$, and $L^{-1} S_{L}(\lambda)$ is SSE. With all conditions verified, by Theorem \ref{thm:ullnlip}, $\sup_{\lambda} |L^{-1} S_{L}(\lambda) - \mu_{L}(\lambda)| \xrightarrow{a.s.} 0$. Next, we verify convergence in $L^{2}$. Firstly,
    \begin{equation*}
      \begin{aligned}
        \limsup_{L \to \infty} \E\bracks{\sup_{\lambda} |L^{-1} S_{L}(\lambda) - \mu_{L}(\lambda)|}
        & \leq \limsup_{L \to \infty} \sn{\mu_{L}} + \limsup_{L \to \infty} L^{-1} \sum_{\ell=1}^{L} \E\bracks{\sn{s_{\ell,L}(\theta_{\ell}, \cdot)}} \\
        & \leq \limsup_{L \to \infty} \sn{\mu_{L}} + \limsup_{L \to \infty} \sup_{\ell} \E\bracks{\sn{s_{\ell,L}(\theta_{\ell}, \cdot)}} \\
        & < \infty,
      \end{aligned}
    \end{equation*}
    as well as by independence,
    \begin{equation*}
      \begin{aligned}
        \limsup_{L \to \infty} \Var\bracks{\sup_{\lambda} |L^{-1} S_{L}(\lambda) - \mu_{L}(\lambda)|}
        & \leq \limsup_{L \to \infty} L^{-2} \sum_{\ell=1}^{L} \Var\bracks{\sn{s_{\ell,L}(\theta_{\ell}, \cdot)}} \\
        & \leq \limsup_{L \to \infty} L^{-1} \sup_{\ell} \E\bracks{\sn{s_{\ell,L}(\theta_{\ell}, \cdot)}^{2}} \\
        & = 0,
      \end{aligned}
    \end{equation*}
    so in conjunction, $\limsup_{L \to \infty} \E\bracks{\sup_{\lambda} |L^{-1} S_{L}(\lambda) - \mu_{L}(\lambda)|^{2}} < \infty$. It follows that $\sup_{\lambda} |L^{-1} S_{L}(\lambda) - \mu_{L}(\lambda)|$ is uniformly integrable, and therefore $\E\bracks{\sup_{\lambda} |L^{-1} S_{L}(\lambda) - \mu_{L}(\lambda)|} \to 0$. Having already established that the variance vanishes, it is immediate that $\E\bracks{\sup_{\lambda} |L^{-1} S_{L}(\lambda) - \mu_{L}(\lambda)|^{2}} \to 0$ as well.
    \end{proof}
\end{lemma}

\end{document}